\documentclass{article}
\usepackage{cite,color,xspace, balance}

\usepackage{kpfonts}
\usepackage{latexsym}
\usepackage{todonotes}

\usepackage{amsthm} 
\usepackage{amsmath}
\usepackage{amsfonts}
\usepackage{epsfig}
\usepackage{cite}
\usepackage{url}
\usepackage{graphicx,tikz,tikz-3dplot}
\usepackage{fancyheadings}
\usepackage{subfigure}
\usepackage[T1]{fontenc}
\usepackage{times}
\usepackage{algorithm}
\usepackage[noend]{algpseudocode}
\usepackage{fullpage}
\usepackage[matha]{mathabx}
\usepackage{mathtools}

\usepackage{txfonts} % for semijoin

\usepackage{hyperref}

\newcommand{\eat}[1]{}

\newcommand{\poly}{\mathrm{poly}}

\newcommand{\eps}{\epsilon}

\newcommand{\la}{\leftarrow}

\newcommand{\be}{\begin{enumerate}}
\newcommand{\ee}{\end{enumerate}}
\newcommand{\bi}{\begin{itemize}}
\newcommand{\ei}{\end{itemize}}
\newcommand{\beq}{\begin{equation}}
\newcommand{\eeq}{\end{equation}}

\newcommand{\bp}{\begin{proof}}
\newcommand{\ep}{\end{proof}}
\newcommand{\bcor}{\begin{cor}}
\newcommand{\ecor}{\end{cor}}
\newcommand{\bthm}{\begin{thm}}
\newcommand{\ethm}{\end{thm}}
\newcommand{\blmm}{\begin{lmm}}
\newcommand{\elmm}{\end{lmm}}
\newcommand{\bdefn}{\begin{defn}}
\newcommand{\edefn}{\end{defn}}
\newcommand{\bprop}{\begin{prop}}
\newcommand{\eprop}{\end{prop}}
\newcommand{\bconj}{\begin{conj}}
\newcommand{\econj}{\end{conj}}
\newcommand{\bopm}{\begin{opm}}
\newcommand{\eopm}{\end{opm}}
\newcommand{\brmk}{\begin{rmk}}
\newcommand{\ermk}{\end{rmk}}

\newcommand{\norm}[1]{\lVert{#1}\rVert}

\newcommand{\suchthat}{\ | \ }

\newcommand{\mv}[1]{\mathbf{#1}}

\theoremstyle{plain}                   % default
\newtheorem{thm}{Theorem}[section]
\newtheorem{lmm}[thm]{Lemma}
\newtheorem{prop}[thm]{Proposition}
\newtheorem{cor}[thm]{Corollary}

\theoremstyle{definition}              % Examples and all

\newtheorem{defn}[thm]{Definition}
\newtheorem{opm}[thm]{Open Problem}
\newtheorem{conj}[thm]{Conjecture}
\newtheorem{example}[thm]{Example}

\newtheorem{rmk}[thm]{Remark}

\newcommand{\bbox}{
\begin{center}
\begin{tabular}{|c|}
\hline
}
\newcommand{\ebox}{
\\
\hline
\end{tabular}
\end{center}
}

\newlength{\toppush}
\def\subjnum{CSE 713}
\def\subjname{Probabilistically Checkable Proofs and Inapproximability}
\def\doheading#1#2#3{\vfill\eject\vspace*{-\toppush}%
  \vbox{\hbox to\textwidth{{\bf}
  \subjnum: \subjname
  \hfil Lecturer: Hung Q. Ngo}%
  \hbox to\textwidth{{\bf} SUNY at Buffalo, Fall 2004\hfil#3\strut}%
  \hrule}
}

\algrenewcommand\algorithmicrequire{\textbf{Input:}}
\algrenewcommand\algorithmicensure{\textbf{Output:}}
\algrenewcommand\algorithmicwhile{\textbf{While}}
\algrenewcommand\algorithmicfor{\textbf{For}}
\algrenewcommand\algorithmicreturn{\textbf{Return}}
\algrenewcommand\algorithmicif{\textbf{If}}

\newcommand{\algo}{\mathcal{A}}

\def\compactify{\itemsep=0pt \topsep=0pt \partopsep=0pt \parsep=0pt}
\let\latexusecounter=\usecounter

\def\compactifytwo{\itemsep=-1pt \topsep=-1pt \partopsep=-2pt \parsep=-1pt
\labelwidth=-2pt \leftmargin=-10pt}

\newcommand{\set}[1]{\left\{ #1 \right\}}

\newcommand{\setof}[2]{\left\{ #1 \mid #2 \right\}}

\newcommand{\dlm}{\textsf{DLM}\xspace}
\newcommand{\nprr}{\textsf{NPRR}\xspace}
\newcommand{\lb}{\textsf{LFTJ}\xspace}

\newcommand{\vc}{\mathbf{c}}

\newcommand{\arr}{\textsc{SortedList}}
\newcommand{\iarr}{\textsc{IntervalList}}
\newcommand{\fnd}{\textsc{Find}}
\newcommand{\nxt}{\textsc{Next}}

\newcommand{\del}{\textsc{Delete}}
\newcommand{\deli}{\textsc{DeleteInterval}}

\newcommand{\ins}{\textsc{insert}}
\newcommand{\insconst}{\textsc{InsConstraint}}
\newcommand{\inst}{\textsc{InsertTree}}

\newcommand{\chld}{\textsc{equalities}}

\newcommand{\intv}{\textsc{intervals}}
\newcommand{\ctree}{\textsc{ConstraintTree}}

\newcommand{\covers}{\textsc{covers}}
\newcommand{\nextvalue}{\textsc{nextChainVal}}
\newcommand{\nextvaluefromshadow}{\textsc{nextShadowChainVal}}

\newcommand{\getpp}{\textsc{getProbePoint}}
\newcommand{\NULL}{\textsc{null}}
\newcommand{\fndlub}{\textsc{FindLub}}

\newcommand{\inter}{\textsc{Intersect}}
\newcommand{\nxtu}{\textsc{NextUnion}}
\newcommand{\nxtsbl}{\textsc{NextSibling}}
\newcommand{\sbl}{\textsc{Sibling}}
\newcommand{\getcache}{\textsc{GetCache}}
\newcommand{\cache}{\textsc{Cache}}

\newcommand{\outspace}{\mathcal O}

\newcommand{\atoms}{\mathrm{atoms}}
\newcommand{\arity}{\text{arity}}

\newcommand{\ms}{\text{Minesweeper}\xspace}
\newcommand{\cds}{\text{CDS}\xspace}
\newcommand{\cert}{\mathcal C}
\newcommand{\btree}{\text{Search-Tree}\xspace}

\newcommand{\findgap}{\textsc{FindGap}}

\newcommand{\tw}{\text{tw}}

\newcommand{\calH}{\mathcal H}
\newcommand{\calA}{\mathcal A}
\newcommand{\calP}{\mathcal P}
\newcommand{\calV}{\mathcal V}
\newcommand{\calE}{\mathcal E}

\renewenvironment{quote}{%
  \list{}{%
    \leftmargin0.5cm   % this is the adjusting screw
    \rightmargin\leftmargin
  }
  \item\relax
}
{\endlist}

\begin{document}
\title{Beyond Worst-case Analysis for Joins with Minesweeper\footnote{This is
the full version of our PODS'2014 paper.}}

\author{Hung Q. Ngo\\
Computer Science and Engineering\\
University at Buffalo, SUNY
\and Dung T. Nguyen\\
Computer Science and Engineering\\
University at Buffalo, SUNY
\and 
Christopher R\'e\\
Computer Science\\
Stanford University
\and Atri Rudra\\
Computer Science and Engineering\\
University at Buffalo, SUNY\\
}

\date{}

\maketitle

\begin{abstract}
We describe a new algorithm, \ms, that is able to satisfy stronger
runtime guarantees than previous join algorithms (colloquially,
`beyond worst-case guarantees') for data in indexed search trees. Our
first contribution is developing a framework to measure this stronger
notion of complexity, which we call {\it certificate complexity}, that
extends notions of Barbay et al. and Demaine et al.; a certificate is
a set of propositional formulae that certifies that the output is
correct. This notion captures a natural class of join algorithms. In
addition, the certificate allows us to define a strictly stronger
notion of runtime complexity than traditional worst-case
guarantees. Our second contribution is to develop a dichotomy theorem
for the certificate-based notion of complexity. Roughly, we show that
\ms evaluates $\beta$-acyclic queries in time linear in the
certificate plus the output size, while for any $\beta$-cyclic query
there is some instance that takes superlinear time in
the certificate (and for which the output is no larger than the
certificate size). We also extend our certificate-complexity analysis
to queries with bounded treewidth and the triangle query.
%\yell{With the help of LogicBlox, we implemented 
%\ms, and our preliminary experience is that it can be faster than
%state-of-the-art join algorithms on real data.}
\end{abstract}

%%%%%%%%%%%%%%%%%%%%%%%%%%%%%%%%%%%%%%%%%%%%%%%%%%%%%%%%%%%%%%%%%%%%%%%%%%
\section{Introduction}
\label{section:intro}
%%%%%%%%%%%%%%%%%%%%%%%%%%%%%%%%%%%%%%%%%%%%%%%%%%%%%%%%%%%%%%%%%%%%%%%%%%
Efficiently evaluating relational joins is one of the most
well-studied problems in relational database theory and
practice. Joins are a key component of problems in constraint
satisfaction, artificial intelligence, motif finding, geometry, and
others. This paper presents a new join algorithm, called \ms, for
joining relations that are stored in order data structures, such as
B-trees. Under some mild technical assumptions, \ms is able to achieve
stronger runtime guarantees than previous join algorithms.

The \ms algorithm is based on a simple idea. When data are stored in
an index, successive tuples indicate {\it gaps}, i.e., regions in the
output space of the join where no possible output tuples exist.  \ms
maintains gaps that it discovers during execution and infers where to
look next. In turn, these gaps may indicate that a large number of
tuples in the base relations cannot contribute to the output of the
join, so \ms can efficiently skip over such tuples without reading
them. By using an appropriate data structure to store the gaps, \ms
guarantees that we can find at least one point in the output space
that needs to be explored, given the gaps so far.  The key technical
challenges are the design of this data structure, called the {\it
  constraint data structure}, and the analysis of the join algorithm
under a more stringent runtime complexity measure.

To measure our stronger notion of runtime, we introduce the notion of
a {\em certificate} for an instance of a join problem: essentially, a
certificate is a set of comparisons between elements of the input
relations that certify that the join output is exactly as claimed. We
use the certificate as a measure of the difficulty of a particular
instance of a join problem. That is, our goal is to find algorithms
whose running times can be bounded by some function of the {\em
smallest certificate size} for a particular input instance. Our
notion has two key properties:

\begin{itemize}
\item {\it Certificate complexity captures the computation performed
  by widely implemented join algorithms.} We observe that the set of
  comparisons made by any join algorithm that interacts with the data
  by comparing elements of the input relations (implicitly) constructs
  a certificate. Examples of such join algorithms are
  index-nested-loop join, sort-merge join, hash join,\footnote{Within a
    $\log$-factor, an ordered tree can simulate a hash
    table.} grace join, and block-nested loop join. Hence, our results
  provide a lower bound for this class of algorithms, as any such
  algorithm must take at least as many steps as the number of
  comparisons in a smallest certificate for the instance.

\item {\it Certificate complexity is a strictly finer notion of
  complexity than traditional worst-case data complexity}. In
  particular, we show that there is always a certificate that is no
  larger than the input size. In some cases, the certificate may be
  much smaller (even constant-sized for arbitrarily large
  inputs). 
\end{itemize}

\noindent These two properties allow us to model a common situation in
which indexes allow one to answer a query {\em without} reading all of
the data---a notion that traditional worst-case analysis is too coarse
to capture. We believe ours is the first {\it beyond worst-case
  analysis} of join queries.

Throughout, we assume that all input relations are indexed
consistently with a particular ordering of all attributes called the
{\em global attribute order} (GAO). In effect, this assumption means
that we restrict ourselves to algorithms that compare elements in GAO
order. This model, for example, excludes the possibility that a
relation will be accessed using indexes with multiple search keys
during query evaluation.

With this restriction, our main technical results are as
follows. Given a $\beta$-acyclic query we show that there is some GAO such
that \ms runs in time that is essentially optimal in the
certificate-sense, i.e., in time $\tilde O(|\cert| + Z)$, where
$\cert$ is a smallest certificate for the problem instance, $Z$ is the
output size, and $\tilde{O}$ hides factors that depend (perhaps
exponentially) on the query size and at most logarithmically on the
input size.\footnote{The exponential dependence on the query is
  similar to traditional data complexity; the logarithmic dependence
  on the data is an unavoidable technical necessity (see
  Appendix~\ref{app:sec:rt:analysis}).}  Assuming the 3\textsc{SUM}
conjecture, this boundary is tight, in the sense that any
$\beta$-cyclic query (and any GAO) there are some family of instances
that require a run-time of $\Omega(|\cert|^{4/3-\eps} + Z)$ for any
$\eps >0$ where $Z = O(|\cert|)$. For $\alpha$-acyclic join queries,
which are the more traditional notion of acyclicity in database theory
and a strictly larger class than $\beta$-acyclic queries, Yannakakis's
seminal join algorithm has a worst-case running time that is linear in
the input size plus output size (in data complexity). However, we show
that in the certificate world, this boundary has changed: assuming the
exponential time hypothesis, the runtime of any algorithm for
$\alpha$-acyclic queries cannot be bounded by any polynomial in
$|\cert|$.\footnote{In Appendix~\ref{app:sec:counter-examples}, we
  show that both worst-case optimal
  algorithms~\cite{DBLP:conf/pods/NgoPRR12,DBLP:journals/corr/abs-1210-0481}
  and Yannakakis's algorithm run in time $\omega(|\cert|)$ for
  $\beta$-acyclic queries on some family of instances.}

We also describe how to extend our results to notions of
treewidth. Recall that any ordering of attributes can be used to
construct a tree decomposition. Given a GAO that induces a tree
decomposition with an (induced) {\em treewidth} $w$, \ms runs in time
$\tilde O(|\cert|^{w+1} + Z)$.  In particular, for a query with {\em
  treewidth} $w$, there is always a GAO that achieves $\tilde
O(|\cert|^{w+1}+Z)$.  Moreover, we show that no algorithm
(comparison-based or not) can improve this exponent by more than a
constant factor in $w$. However, our algorithm does not have an
optimal exponent: for the special case of the popular triangle query,
we introduce a more sophisticated data structure that allows us to run
in time $\tilde{O}(|\cert|^{3/2}+Z)$, while $\ms$ runs in time
$\tilde{O}(|\cert|^2+Z)$.
 
\paragraph*{Outline of the Remaining Sections} 
In Section~\ref{section:formulation}, we describe the notion of a
certificate and formally state our main technical problem and
results. In Section~\ref{section:ms:alg}, we give an overview of the
main technical ideas of \ms, including a complete description of our
algorithm and its associated data structures. In
Section~\ref{sec:beta}, we describe the analysis of \ms for
$\beta$-acyclic queries. In Section~\ref{sec:tw}, we then describe how
to extend the analysis to queries with low-treewidth and the triangle
query. In Section~\ref{sec:related:work}, we discuss related
work. Most of the technical details are provided in the appendix.

%%%%%%%%%%%%%%%%%%%%%%%%%%%%%%%%%%%%%%%%%%%%%%%%%%%%%%%%%%%%%%%%%%%%%%%%%%
\section{Problem and Main Result}
\label{section:formulation}
%%%%%%%%%%%%%%%%%%%%%%%%%%%%%%%%%%%%%%%%%%%%%%%%%%%%%%%%%%%%%%%%%%%%%%%%%%

Roughly, the main problem we study is:
\begin{quote}
Given a natural join query $Q$ and a database instance $I$, compute
$Q$ in time $f(|\cert|,Z)$, where $\cert$ is the smallest
``certificate" that certifies that the output $Q(I)$ is as claimed by
the algorithm and $Z=|Q(I)|$.
\end{quote}

\noindent
 We will assume that all relations in the input are already
 indexed. Ideally, we aim for
 $f(|\cert|,Z)=O\left(|\cert|+Z\right)$. We make this problem precise
 in this section.

%% We now specify how the
%%  database instance $I$ is pre-processed and a precise definition of
%%  certificate. 

% ------------------------------------------------------------------------
\subsection{The inputs to \ms}
\label{subsec:ms-inputs}
% ------------------------------------------------------------------------

We assume a set of attributes $A_1,\dots,A_n$ and denote the domain of
attribute $A_i$ as $\mv D(A_i)$. Throughout this paper, without loss
of generality, we assume that all attributes are on domain $\mathbb
N$. We define three items: (1) the global attribute order; (2) our
notation for order; and (3) our model for how the data are indexed.

\paragraph*{The Global Attribute Order} 
\ms evaluates a given natural join query $Q$ consisting of a set
$\atoms(Q)$ of relations indexed in a way that is consistent with an
ordering $A_1,\dots,A_n$ of all attributes occurring in $Q$ that we
call the {\em global attribute order} (GAO). To avoid burdening the
notation, we assume that the GAO is simply the order
$A_{1},\dots,A_{n}$. We assume that all relations are stored in
ordered search trees (e.g., B-trees) where the search key for this
tree is consistent with this global order. For example, $(A_1,A_3)$ is
consistent, while $(A_3,A_2)$ is not.

\paragraph*{Tuple-Order Notation} 
We will extensively reason about the relative order of tuples and
describe notation to facilitate the arguments. For a relation
$R(A_{s(1)}, \dots, A_{s(k)})$ where $s : [k] \to [n]$ is such that $s(i)
< s(j)$ if $i < j$, we define an {\em index tuple} $\mv x =
(x_1,\cdots,x_j)$ to be a tuple of positive integers, where $j \leq
k$. Such tuples index tuples in the relation $R$. We define their
meaning inductively. If $\mv x = (x_1)$, then $R[\mv x]$ denotes the
$x_1$'th smallest value in the set $\pi_{A_{s(1)}}(R)$. 
%For example, $R[2]$ denotes the second smallest value in $\pi_{A_{s(1)}}(R)$. 
Inductively, define $R[\mv x]$ to be the $x_j$'th smallest value in the set
\[
R[x_1,\dots,x_{j-1}, *] := \pi_{A_{j}}
\bigl( \sigma_{A_{s(1)}=R[x_1], \cdots, A_{s(j-1)}=R[x_1,\dots,x_{j-1}]}(R) 
\bigr).
\]
For example, if $R(A_1,A_2) = \set{ (1,1), (1,8), (2,3), (2,4)}$ then
$R[*] = \{1,2\}$, $R[1,*] = \{1,8\}$, $R[2] = 2$, and $R[2,1] = 3$.

We use the following convention to simplify the algorithm's
description: for any index tuple
$(x_1,\dots,x_{j-1})$,
\begin{eqnarray}
R[x_1,\dots,x_{j-1},0] &=& -\infty \label{eqn:-infty-convention}\\
R[x_1,\dots,x_{j-1},|R[x_1,\dots,x_{j-1}, *]|+1] &=& +\infty .
\label{eqn:+infty-convention}
\end{eqnarray}

\paragraph*{Model of Indexes}
The relation $R$ is indexed such that the values of various attributes
of tuples from $R$ can be accessed using index tuples. We
assume appropriate size information is stored so that we know what
the correct ranges of the $x_j$'s are; for example, following the notation
described above, the correct range is 
$1 \leq x_j \leq |R[x_1,\dots,x_{j-1},*]|$ for every $j \leq \arity(R)$. 
With the convention specified in 
\eqref{eqn:-infty-convention} and \eqref{eqn:+infty-convention},
$x_j=0$ and $x_j=|R[x_1,\dots,x_{j-1},*]|+1$ are {\em out-of-range} coordinates.
These coordinates are used for the sake of brevity only; an index tuple, by
definition, cannot contain out-of-range coordinates.

The index structure for $R$ supports the query $R.\findgap(\mv x, a)$, which
takes as input an index tuple $\mv x = (x_1,\dots,x_j)$ of length $0
\leq j < k$ and a value $a \in \mathbb Z$, and returns a pair of
coordinates $(x_{-},x_{+})$ such that
\begin{itemize} 
\item $0 \leq x_{-} \leq x_{+} \leq |R[(\mv x, *)]|+1$
\item $R[(\mv x, x_{-})] \leq a \leq R[(\mv x, x_{+})]$, and
\item $x_{-}$ (resp. $x_{+}$) is the maximum (resp. minimum)
  index satisfying this condition.
\end{itemize}
\noindent Note that it is possible for $x_- = x_+$, which holds when
$a \in R[(\mv x,*)]$. Moreover, we assume throughout that $\findgap$
runs in time $O(k\log |R|)$. This model captures widely used indexes
including a B-tree~\cite[Ch.10]{Ramakrishnan:2002:DMS:560733} or a
Trie~\cite{DBLP:journals/corr/abs-1210-0481}.

% ------------------------------------------------------------------------
\subsection{Certificates}
% ------------------------------------------------------------------------
\label{sec:certificate}
%\subsubsection{Definitions}

%% In this section, we define a combinatorial notion called
%% ``certificate'' that aims to capture the set of comparisons a
%% comparison-based join algorithm must perform to correctly compute the
%% output. Given the input instance as described in the previous section,
%% our class of comparison-based algorithms can only access input
%% relations using index tuples, and to perform comparisons of the form
%% $R[\mv x] \theta S[\mv y]$, where $\theta \in \{<,=,>\}$, $\mv x$ and
%% $\mv y$ are valid index tuples so that $R[\mv x]$ and $S[\mv y]$ are
%% $A_k$-values for some $k\in [n]$. In particular, these algorithms have
%% to compute the join while being oblivious to the specific domain
%% values stored in the $R[\mv x]$. The set of comparisons that such an
%% algorithm issues during execution must satisfy some properties for the
%% algorithm to be correct. Intuitively, any two database instances
%% satisfying the same set of comparisons must have the same output for
%% such an algorithm to work. Such a set of comparisons is called a
%% certificate for the problem instance. In what follows, we make the
%% above intuition precise.  

We define a {\em certificate}, which is a set of comparisons that
certifies the output is exactly as claimed. We do not want the
comparisons to depend on the specific values in the instance, only
their order. To facilitate that, we think of $R[\mv x]$ as a variable
that can be mapped to specific domain value by a database
instance.\footnote{We use variables as a perhaps more intuitive,
  succinct way to describe the underlying morphisms.}  These variables
are only defined for valid index tuples as imposed by the input
instance described in the previous section.

A {\em database instance} $I$ instantiates all variables $R[\mv x]$, 
where $\mv x = (x_1,\dots,x_j),~ 1\le j\leq \arity(R)$, is an index tuple in relation
$R$. (In particular, the input database instance described in the previous
section is such a database instance.)
We use $R^I[\mv x]$ to denote the instantiation of the variable
$R[\mv x]$. Note that each such variable is on the domain of some
attribute $A_k$; for short, we call such variable an {\em
  $A_k$-variable}. A database instance $I$ fills in specific values to the
nodes of the search tree structures of the input relations.

\begin{example}
Consider the query $Q = R(A) \Join T(A, B)$ on the input instance
$I(N)$ defined by $R^{I(N)} = [N]$ and $T^{I(N)} = \{ (1, 2i)
\suchthat i \in [N] \} \cup \{ (2, 3i) \suchthat i \in [N] \}$.  This
instance can be viewed as defining the following variables: $R[i]$,
$i\in [N]$, $T[1]$, $T[2]$, $T[1,i]$, and $T[2,i]$, $i\in [N]$.
Another database instance $J$ can define the same index variables but
using different constants, in particular, set $R^J[i] = \{ 2i
\suchthat i\in [N] \}$, $T^J[1] = 2$, $T^J[2]=4$, $T^J[1,i] = i$, and
$T^J[2, i] = 10i$, $i\in [N]$.
\label{ex:very-basic}
\end{example}

%%The following combinatorial notion called ``certificate'' is meant to
%%capture the runtime of a comparison-based join algorithm.
We next formalize the notion of certificates.
Consider an input instance to \ms, consisting of the query $Q$,
the GAO $A_1,\dots,A_n$, and a set of relations $R \in \atoms(Q)$ already 
indexed consistently with the GAO.

\bdefn[Argument]
An {\em argument} for the input instance is a set of symbolic comparisons of 
the form
\begin{equation}
 R[\mv x] \ \theta \ S[\mv y], \ \text{ where } R, S\in \atoms(Q)
\label{eqn:comparison-forms}
\end{equation}
and $\mv x$ and $\mv y$ are two index tuples\footnote{Note again that the index 
tuples are constructed from the input instance as described in the previous 
section.} such that $R[\mv x]$ and
$S[\mv y]$ are both $A_k$-variables for some $k \in [n]$, and $\theta
\in \{<, =, >\}$.  Note that we allow $R=S$.\footnote{Equality constraints
  between index tuples from same relation is {\em required} to
  guarantee that certificates are no longer than the input, see
  property (ii) below.} 
  A database instance $I$ {\em satisfies an
  argument} ${\cal A}$ if $R^I[\mv x] \ \theta \ S^I[\mv y]$ is true
for every comparison $R[\mv x] \ \theta \ S[\mv y]$ in the argument
${\cal A}$.  \edefn

An index tuple $\mv x = (x_1,\dots,x_r)$ for a relation $S$ is called a {\em
full index tuple} if $r = \arity(S)$.  Let $I$ be a database instance for the
problem.  Then, the full index tuple $\mv x$ is said to {\em contribute} to an
output tuple $\mv t \in Q(I) = \ \Join_{R\in\atoms(Q)} R^I$ if the tuple
$(S[x_1], S[x_1,x_2], \dots,S[\mv x])$ is exactly the projection of
$\mv t$ onto attributes in $S$. 
A collection $X$ of full index tuples is said to
be a {\em witness} for $Q(I)$ if $X$ has exactly one full
index tuple from each relation $R\in \atoms(Q)$, and all index tuples
in $X$ contribute to the same $\mv t \in Q(I)$.

\bdefn[Certificate]\label{defn:certificate}
An argument $\mathcal A$ for the input instance is called a 
{\em certificate} iff the following condition is satisfied: 
if $I$ and $J$ are two database
instances of the problem both of which satisfy $\mathcal A$, then 
{\em every} witness for $Q(I)$ is a witness for $Q(J)$ and vice versa.
The {\em size} of a certificate is the number of comparisons in it.
\edefn

\begin{example}
Continuing with Example~\ref{ex:very-basic}. Fix an $N$, the argument
$\{ R[1] = T[1],~ R[2] = T[2] \}$ is a certificate for $I(N)$. For
every database, such as $I=I(N)$ and $J$ in the example, that
satisfies the two equalities, the set of witnesses are the same, i.e.,
the sets $\{ 1, (1, i) \}$ and $\{ 2, (2, i) \}$ for $i \in
[N]$. Notice we do not need to spell out all of the outputs in the
certificate.

Consider the instance $K$ in which $R^K = [N]$, $T^K = \{ (1, 2i)
\suchthat i \in [N] \} \cup \{ (3, 3i) \suchthat i \in [N] \}$. While
$K$ is very similar to $I$, $K$ does {\em not} satisfy the certificate
since $R^{K}[2] \neq T^{K}[2]$. The certificate also does not apply to
$I(N+1)$ from Example~\ref{ex:very-basic}, since $I(N+1)$ defines a
different set of variables from $I(N)$, e.g., $T[1,N+1]$ is defined in
$I(N+1)$, but not in $I(N)$.
\end{example}

% ------------------------------------------------------------------------
\paragraph*{Properties of optimal certificates}
% ------------------------------------------------------------------------

We list three important facts about $\cert$, a minimum-sized certificate:
\bi
 \item[(i)] The set of comparisons issued by a very natural class of
(non-deterministic) comparison-based join algorithms {\em is} a certificate; 
this result not only justifies the definition of certificates, but also shows that
$|\cert|$ is a lowerbound for the runtime of any comparison-based join algorithm.
 
\item[(ii)] $|\cert|$ can be shown to be at most linear in the input size
{\em no matter what the data and the GAO are},
and in many cases $|\cert|$ can even be of constant size.
Hence, running time measured in $|\cert|$ is a strictly finer notion of 
runtime complexity than input-based runtimes; and

 \item[(iii)] $|\cert|$ depends on the data and the GAO.
\ei

We explain the above facts more formally in the following two
propositions.  The proofs of the propositions can be found in 
Appendix~\ref{app:sec:certificates}.

\bprop[Certificate size as run-time lowerbound of comparison-based
  algorithms] Let $Q$ be a join query whose input relations are
already indexed consistent with a GAO as described in
Section~\ref{subsec:ms-inputs}.  Consider any comparison-based join
algorithm that only does comparisons of the form shown in
\eqref{eqn:comparison-forms}. Then, the set of comparisons performed
during execution of the algorithm {\em is} a certificate.  In
particular, if $\cert$ is an optimal certificate for the problem, then
the algorithm must run in time at least $\Omega(|\cert|)$.
\label{prop:Omega(|C|)}
\eprop

\bprop[Upper bound on optimal certificate size]
Let $Q$ be a general join query on $m$ relations and $n$ attributes. 
Let $N$ be the total number of tuples from all input relations.
Then, no matter what the input data and the GAO are, we have
$|\cert| \leq r\cdot N$, where 
$r = \max \{\arity(R) \suchthat R \in \atoms(Q)\} \leq n$.
\label{prop:optimal-certificate-upperbound}
\eprop

In Appendix~\ref{app:sec:certificates}, we present examples to
demonstrate that $|\cert|$ can vary any where from $O(1)$ to
$\Theta(|\text{input-size}|)$, that the input data or the GAO can
change the certificate size, and that same-relation comparisons are
needed.

% ------------------------------------------------------------------------
\subsection{Main Results}
\label{sec:problem-results}
% ------------------------------------------------------------------------

Given a set of input relations already indexed consistent with a fixed
GAO, we wish to compute the natural join of these relations as quickly
as possible.  As illustrated in the previous section, a runtime
approaching $|\cert|$ is optimal among comparison-based
algorithms. Furthermore, runtimes as a function of $|\cert|$ can be
sublinear in the input size.  Ideally, one would like a join algorithm
running in $\tilde O(|\cert|)$-time.  However, such a runtime is
impossible because for many instances the output size $Z$ is {\em
  super}linear in the input size, while $|\cert|$ is at most linear in
the input size. Hence, we will aim for runtimes of the form $\tilde O(
g(|\cert|) + Z)$, where $Z$ is the output size and $g$ is some
function; a runtime of $\tilde O(|\cert|+Z)$ is essentially optimal.

Our algorithm, called \ms, is a general-purpose join algorithm. Our
main results analyze its runtime behavior on various classes of
queries in the certificate complexity model.  Recall that
$\alpha$-acyclic (often just acyclic) is the standard notion of
(hypergraph) acyclicity in database
theory~\cite[p.~128]{DBLP:books/aw/AbiteboulHV95}. A query is
$\beta$-acyclic, a stronger notion, if every subquery of $Q$ obtained
by removing atoms from $Q$ remains $\alpha$-acyclic. For completeness,
we include these definitions and examples in
Appendix~\ref{app:sec:gao}.

Let $N$ be the input size, $n$ the number of attributes, $m$ the
number of relations, $Z$ the output size, $r$ the maximum arity of
input relations, and $\cert$ any optimal certificate for the instance.
Our key results are as follows.

\begin{thm}
Suppose the input query is $\beta$-acyclic. Then there is some GAO
such that $\ms$ computes its output in time $O\left(
2^nm^2n\left(4^r|\cert| + Z \right) \log N \right)$.
\label{thm:io-beta-neo}
\end{thm}

As is standard in database theory, we ignore the dependency on the
query size, and the above theorem states that \ms runs in time $\tilde
O(|\cert|+Z)$.\footnote{For $\beta$-acyclic queries with a fixed GAO,
  our results are loose; our best upper bound the complexity uses the
  treewidth from Section~\ref{sec:tw}.}

What about $\beta$-{\em cyclic} queries? The short answer is {\em no}:
we cannot achieve this guarantee.  It is obvious that any join
algorithm will take time $\Omega(Z)$. Using {\em
  $3$\textsf{SUM}-hardness}, a well-known complexity-theoretic
assumption~\cite{3sum}, we are able to show the following.

\bprop
\label{prop:no-io-non-beta}
Unless the $3$\textsf{SUM} problem can be solved in sub-quadratic
time, for any $\beta$-{\em cyclic} query $Q$ in any GAO, there does
not exist an algorithm that runs in time $O(|\cert|^{4/3-\eps} + Z)$
for any $\eps>0$ on all instances.  
\eprop

We extend our analysis of \ms to queries that have bounded treewidth
and to triangle queries in Section~\ref{sec:tw}. These results are
technically involved and we only highlight the main technical
challenges.

%%%%%%%%%%%%%%%%%%%%%%%%%%%%%%%%%%%%%%%%%%%%%%%%%%%%%%%%%%%%%%%%%%%%%%%%%%
\section{The \ms Algorithm}
\label{section:ms:alg}
We begin with an overview of the main ideas and technical challenges
of the \ms algorithm.  Intuitively, \ms probes into the space of all
possible output tuples, and explores the gaps in this space where
there is no output tuples.  These gaps are encoded by a technical
notion called constraints, which we describe next. (For illustration,
we present complete end-to-end results for set intersection and the
{\em bow-tie} query in Appendix~\ref{app:sec:intersection}
and~\ref{app:sec:bowtie}.)

\subsection{Notation for \ms}

We need some notation to describe our algorithm. Define the {\em
  output space} $\outspace$ of the query $Q$ to be the space
$\outspace = \mv D(A_1) \times \mv D(A_2) \times \cdots \times \mv
D(A_n)$, where $\mv D(A_i)$ is the domain of attribute
$A_i$.\footnote{Recall, we assume $\mv D(A_i) = \mathbb N$ for
  simplicity.} By definition, a tuple $\mv t$ is an {\em output tuple}
if and only if $\mv t = (t_1, \dots, t_n) \in \outspace$, and
$\pi_{\bar A(R)}(\mv t) \in R$, for all $R\in \atoms(Q)$,
where $\bar A(R)$ is the set of attributes in $R$.

\paragraph*{Constraints}
A {\em constraint} $\mv c$ is an $n$-dimensional vector of the
following form: $\mv c = \langle c_1,\cdots, c_{i-1},
(\ell,r), \{*\}^{n-i} \rangle,$ where $c_j\in \mathbb N \cup
\{*\}$ for every $j\in [i-1]$.  In other words, each constraint $\mv c$ is a
vector consisting of three types of components:
\begin{itemize}
\item[(1)] {\em open-interval} component $(\ell, r)$ on the attribute $A_i$ 
(for some $i\in [n]$) and $\ell, r \in \mathbb N \cup \{-\infty, +\infty\}$,
\item[(2)] {\em wildcard} or $*$ component, and 
\item[(3)] {\em equality} component of the type $p \in \mathbb N$. 
\end{itemize}
In any constraint, there is exactly one interval component. All
components after the interval component are wildcards.  Hence, we will
often not write down the wildcard components that come after the
interval component.  The prefix that comes before the interval
component is called a {\em pattern}, which consists of any number of
wildcards and equality components.  The equality components encode the
coordinates of the axis parallel affine planes containing the gap. For
example, in three dimensions the constraint $\langle *, (1, 10), *
\rangle$ can be viewed as the region between the affine hyperplanes
$A_2=1$ and $A_2=10$; and the constraint $\langle 1, *, (2,5)\rangle$
can be viewed as the strip inside the plane $A_1=1$ between the line
$A_3=2$ and $A_3=5$.  We encode these gaps syntactically to
facilitate efficient insertion, deletion, and merging.

% What does it mean for t to satisfy a constraint?
Let $\mv t = (t_1, \dots, t_n) \in \outspace$ be an arbitrary tuple
from the output space, and $\mv c = \langle c_1,\dots, c_n \rangle$ be
a constraint.  Then, $\mv t$ is said to {\em satisfy} constraint $\mv
c$ if for every $i\in [n]$ one of the following holds: (1) $c_i = *$,
(2) $c_i \in \mathbb N$ and $t_i = c_i$, or (3) $c_i = (\ell, r)$ and
$t_i \in (\ell, r)$. We say a tuple $\mv t$ is {\em active} with
respect to a set of constraints if $\mv t$ does not satisfy any
constraint in the set (Geometrically, no constraint
covers the point $\mv t$).

\subsection{A High-level Overview of \ms}
\label{sec:algorithm}
%%%%%%%%%%%%%%%%%%%%%%%%%%%%%%%%%%%%%%%%%%%%%%%%%%%%%%%%%%%%%%%%%%%%%%%%%%

We break \ms in two components: (1) a special data structure called
the {\em constraint data structure} (\cds), and (2) an algorithm that uses
this data structure. Algorithm~\ref{algo:outer-shell} gives a
high-level overview of how \ms works, which we will make precise in the
next section.

\textbf{The \cds stores the constraints already discovered during
execution}. For example, consider the query $$R(A,B) \Join S(B).$$ If \ms
determines that $S[4] = 20$ and $S[5] = 28$, then we can deduce that
there is no tuple in the output that has a $B$ value in the open interval
$(20,28)$. This observation is encoded as a constraint $\langle *,
(20,28) \rangle$. A key challenge with the \cds is to
efficiently find an active tuple $\mv t$, given a set of constraints already
stored in the \cds.

\textbf {The outer algorithm queries the \cds to find active tuples
  and then probes the input relations.}  If there is no active $\mv t$, the
algorithm terminates.  Given an active $\mv t$, \ms makes queries into
the index structures of the input relations. These queries either
report that $\mv t$ is an output tuple, in which case $\mv t$ is
output, or they discover constraints that are then inserted into the
\cds. Intuitively, the queries into the index structures are crafted
so that at least one of the constraints that is returned is
responsible for ruling out $\mv t$ in any optimal certificate.

\begin{algorithm}[!htp]
\caption{High-level view: \ms algorithm}
\label{algo:outer-shell}
\begin{algorithmic}[1]
\State $\cds \gets \emptyset$ \Comment{No gap discovered yet}
\While {$\cds$ can find $\mv t$ not in any stored gap}
  \If {$\pi_{\bar A(R)}(\mv t) \in R$ for every $R\in \atoms(Q)$}
    \State Report $\mv t$ and tell \cds that $\mv t$ is ruled out
  \Else
    \State Query all $R \in \atoms(Q)$ for gaps around $\mv t$
    \State Insert those gaps into \cds
  \EndIf
\EndWhile
\end{algorithmic}
\end{algorithm}

We first describe the interface of the \cds and then the outer
algorithm which uses the \cds. 
 
% ------------------------------------------------------------------------
\subsection{The \cds}
\label{sec:cds-outline}
% ------------------------------------------------------------------------

The \cds is a data structure that implements two functions as
efficiently as possible: (1) $\insconst(\mv c)$ takes a new constraint
$\mv c$ and inserts it into the data structure, and (2) $\getpp()$
returns an active tuple $\mv t$ with respect to all constraints that
have been inserted into the \cds, or $\NULL$ if no such $\mv t$
exists.

\paragraph*{Implementation}
To support these operations, we implement the \cds using a tree
structure called $\ctree$, which is a tree with at most $n$ levels,
one for each of the attributes following the GAO.
Figure~\ref{fig:ctree} illustrates such a tree.
More details are provided in Appendix~\ref{app:sec:cds}.
\begin{figure}
\centerline{\includegraphics[width=2.5in]{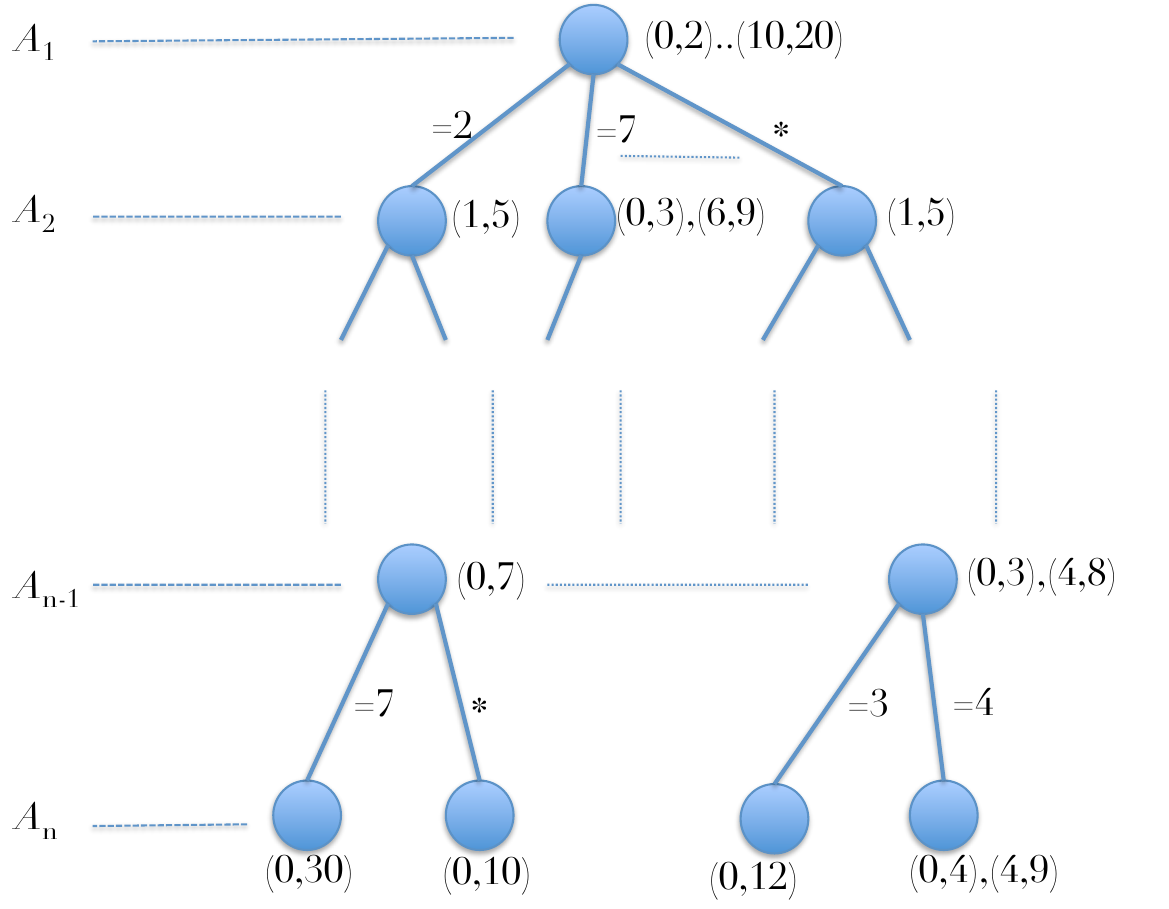}}
\caption{Example of $\ctree$ data structure}
\label{fig:ctree}
\end{figure}
Each node $v$ in the \cds corresponds to a prefix (i.e. pattern) of constraints;
each node has two data structures:
%\begin{itemize}
%\item 

{\bf (1)}
$v.\chld$ is a sorted list with one entry per child of $v$ in the 
underlying tree. Each entry in the sorted list is labeled with an element of 
$\mathbb N$ and has a pointer to the subtree rooted at the corresponding child. 
There are two exceptions: (1) if $v$ is a leaf then $v.\chld=\emptyset$, and
(2) each $v$ has at most one additional child node labeled with $*$.

%\item 
{\bf (2)} $v.\intv$ is a sorted list of disjoint open intervals under that
  corresponding attribute. A key property is that {\em given a value
    $u$ we can, in logarithmic time, output the smallest value $u'\ge
    u$ that is not covered by any interval in $v.\intv$ (via the
    $\nxt$ function)}.  We will maintain the invariant that, for every
  node $v$ in a $\ctree$, none of the labels in $v.\chld$ is contained
  in an interval in $v.\intv$.
%\end{itemize}

The following lemma is straightforward hence we omit the proof.
Note that when we insert a new interval that overlaps existing
intervals and/or contains values in $\chld$,
we will have to merge them and/or remove the entries in $\chld$; 
and hence the cost is amortized.

\begin{prop}
\label{prop:insert-ctree}
The operation $\insconst(\vc)$ can be implemented in amortized time 
$O(n\log{W})$, where $W$ is total number of constraint vectors already inserted.
\end{prop}

The key challenge is to design an efficient implementation of
$\getpp()$; the heart of Sections~\ref{sec:beta} and \ref{sec:tw} is
to analyze $\getpp()$ using properties of the query $Q$.

% ------------------------------------------------------------------------
\subsection{The outer algorithm}
\label{subsec:outer}
% ------------------------------------------------------------------------

Algorithm \ref{algo:outer-algorithm} contains all the details that
were missing from the high-level view of
Algorithm~\ref{algo:outer-shell}. 
Appendix~\ref{app:sec:ms-worked-example} has a complete run of 
\ms on a specific query. Appendices~\ref{app:sec:intersection} and
\ref{app:sec:bowtie} have the complete end-to-end descriptions of two specific
queries, which help clarify the general algorithm.
We prove the following result.

\begin{algorithm}[t]
\caption{$\ms$ for evaluating the query $Q = \ \Join_{R\in \atoms(Q)} R(\bar A(R))$}
\label{algo:outer-algorithm}
\begin{algorithmic}[1]
\Require{We use the conventions defined in \eqref{eqn:-infty-convention} and \eqref{eqn:+infty-convention}}
\State Initialize the constraint data structure $\cds = \emptyset$
\While{$( (\mv t \la \cds.\getpp()) \neq \NULL)$}
  \State Denote $\mv t = (t_1,\dots,t_n)$
  \For {each $R\in \atoms(Q)$}
    \State $k \la \arity(R)$; 
    \State Let $\bar A(R)=(A_{s(1)},\dots,A_{s(k)})$ be $R$'s attributes, 
    where $s: [k]\to [n]$ is such that $s(i)<s(j)$ for $i<j$.
    \For {$p=0$ {\bf to} $k-1$} \Comment{Explore around $\mv t$ in $R$}
      \For {each vector $\mv v \in \{\ell, h\}^p$} \Comment{$\ell,h$ are just
      symbols, and $\{\ell, h\}^0$ has only the empty vector}
      \State Let $\mv v = (v_1,\dots,v_p)$ \Comment{$v_j \in \{\ell, h\}, \forall
      j \in [p]$}
        %% \State $i^{(\mv v, h)}_R \la \min \left\{ i \suchthat 
        %%         R\left[i^{(v_1)}_R, i^{(v_1,v_2)}_R, \dots,
        %%         i^{(v_1,\dots,v_p)}_R, i\right]
        %%         \geq t_{a_{p+1}(R)} \right\}$
        %% \State $i^{(\mv v, \ell)}_R \la \max \left\{ i \suchthat 
        %%         R\left[i^{(v_1)}_R, i^{(v_1,v_2)}_R, \dots,
        %%         i^{(v_1,\dots,v_p)}_R, i\right]
        %%         \leq t_{a_{p+1}(R)} \right\}$
        \State $(i^{(\mv v, \ell)}_R,i_R^{(\mv v, h)}) \la R.\findgap\left(
        \bigl(i^{(v_1)}_R, i^{(v_1,v_2)}_R, \dots,
        i^{(v_1,\dots,v_p)}_R\bigr), t_{s(p+1)}\right)$ \Comment{Gap
        around $(R[\mv i_R^{(\mv v)}],t_{s(p+1)})$ in $R$.}
      \EndFor
    \EndFor
  \EndFor
  \If {$R\left[i^{(h)}_R, i^{(h,h)}_R, \dots,
       i^{\{h\}^p}_R\right] = t_{s(p)}$
       for all $p \in [\arity(R)]$ and for all $R \in \atoms(Q)$}
       \label{main-algo-line:if}
    \State \textbf{Output} the tuple $\mv t$ \label{main-algo-line:output-t}
    \State $\cds.\insconst\left(\langle t_1, t_2,\dots, t_{n-1}, (t_n-1,t_n+1)
    \rangle\right)$
    \label{main-algo-line:constraint-tt}
  \Else \label{main-algo-line:else}
    \For {each $R\in \atoms(Q)$}
      \State $k \la \arity(R)$
      \For {$p=0$ {\bf to} $k-1$}
        \For {each vector $\mv v \in \{\ell, h\}^p$}
           \If {(all the indices $i^{(v_1)}_R, \dots, i^{(v_1,\dots,v_p)}_R$ are
           {\bf not} out of range)}
          \State $\cds.\insconst\left(\left\langle 
                  R\left[i^{(v_1)}_R\right], \dots,
                  R\left[i^{(v_1)}_R, \cdots, i^{(v_1,\dots,v_p)}_R\right], 
                  \left( R[i^{(\mv v, \ell)}_R], R[i^{(\mv v, h)}_R]
                  \right)
                 \right\rangle\right)$ \label{line:additional-constraint}
          \State \Comment{Note that the constraint is empty if $R[i^{(\mv v, \ell)}_R] = R[i^{(\mv v, h)}_R]$}
           \EndIf
        \EndFor
      \EndFor
    \EndFor
  \EndIf
\EndWhile
\end{algorithmic}
\end{algorithm}

\bthm
Let $N$ denote the input size, $Z$ the number of output tuples, 
$m =|\atoms(Q)|$, and $$r = \max_{R\in \atoms(Q)} \arity(R).$$
Let $\cert$ be any optimal certificate for the input instance.
Then, the total runtime of Algorithm \ref{algo:outer-algorithm} is 
\[O\left( \left( 4^r|\cert| + rZ \right) m \log(N)\right) + T(\cds),\]
where $T(\cds)$ is the total time taken by the constraint data structure.
The algorithm inserts $O(m4^r|\cert| + Z)$ constraints to $\cds$ and
issues $O(2^r|\cert| + Z)$ calls to $\getpp()$.
\label{thm:analyze-outer-algorithm}
\ethm

Our proof strategy bounds the number of iterations of the algorithm
using an amortized analysis. We pay for each probe point $\mv t$
returned by the $\cds$ by either charging a comparison in the certificate
$\cert$ {\em or} by charging an output tuple.  If $\mv t$ is an output
tuple, we charge the output tuple. If $\mv t$ is not an output tuple,
then we observe that at least one of the constraints we discovered
must rule out $\mv t$. Recall that each constraint is essentially a pair of
elements from some base relation. If one element from each such pair is not
involved in any comparison in $\cert$, then we can perturb the instance
slightly by moving the comparison-free element to align with $\mv t$.
This means $\cert$ does not have enough information to rule out $\mv t$ as an
output tuple, reaching a contradiction. Hence when $\mv t$ is not an output
tuple, essentially some gap must map to a pair of comparisons. Finally, using
the geometry of the gaps, we show that each comparison is charged at most 
$2^r$ times and each output tuple is
charged $O(1)$ times. Thus, in total the number of iterations is
$O(2^r|\cert|+Z)$. 

When $\cert$ is an optimal-size certificate, the runtime above is
about linear in $|\cert|+Z$ {\em plus} the total time the \cds
takes. Note, however, that $|\cert|$ can be very small, even constant.
Hence, we basically shift all of the burden of join evaluation to
the \cds. Thus, one should not hope that there is an efficient \cds for
general queries:

\bthm[Limitation of any \cds] Unless the
exponential time hypothesis is wrong, no constraint data structure can
process the constraints and the probe point accesses in time
polynomial (independent of the query) in the number of constraints
inserted and probe points accessed.  
\label{thm:cds-limit}
\ethm

Complete proofs of the above theorems are included in
Appendix~\ref{app:sec:outer-analysis}.
In the next sections, we analyze the \cds, specifically the function
$\getpp()$. Our analysis exploits properties of the query and the GAO
for $\beta$-acyclic and bounded treewidth queries.

%%%%%%%%%%%%%%%%%%%%%%%%%%%%%%%%%%%%%%%%%%%%%%%%%%%%%%%%%%%%%%%%%%%%%%%%%%%%%%%
\section{$\beta$-acyclic queries}
\label{sec:beta}
%%%%%%%%%%%%%%%%%%%%%%%%%%%%%%%%%%%%%%%%%%%%%%%%%%%%%%%%%%%%%%%%%%%%%%%%%%%%%%%

We describe how to implement $\getpp$ for $\beta$-acyclic queries. In
particular, we show that there is some GAO that helps implement $\getpp$
in amortized logarithmic time. Hence, by
Theorem~\ref{thm:analyze-outer-algorithm} our running time is $\tilde
O(|\cert|+Z)$, which we argued previously is essentially optimal.

\subsection{Overview}
Recall that given a set of intervals, $\getpp$ returns an active tuple
$\mv t = (t_1,\dots,t_n) \in \outspace$, i.e., a tuple $\mv t$ that does
not satisfy any of the constraints stored in the \cds. Essentially,
during execution there may be a large number of constraints, and
$\getpp$ needs to answer an alternating sequence of constraint
satisfaction problems and insertions. The question is: how do we split
this work between insertion time and querying time? 

In \ms, we take a lazy approach: we insert all the constraints without
doing any cleanup on the \cds. Then, when $\getpp$ is called \ms
might have to do hard work to return a new active tuple, applying
memoization along the way so the heavy labor does not have to be
repeated in the future.  When the GAO has a special structure, this
strategy helps keep every $\cds$ operation at amortized logarithmic 
time. We first give an example to build intuition about how our lazy
approach works.

\begin{example} 
\label{ex:memoization}
Consider a query with three attributes $(A,B,C)$, and suppose the constraints 
that are inserted into the \cds are 
\bi 
 \item[(i)] $\langle a, b, (-\infty, 1)\rangle$ for all $a, b\in [N]$, 
 \item[(ii)] $\langle *, b, (2i-2, 2i)\rangle$ for all $b,i \in [N]$,
 \item[(iii)] $\langle *, *, (2i-1, 2i+1)\rangle$ for $i \in [N]$,
 \item[(iv)] and $\langle *, *, (2N, +\infty)\rangle$.
\ei
There are $O(N^2)$ constraints, and there is no active tuple of the form
$(a, b, c)$ for $a,b \in [N]$. Without memoization, the brute-force strategy
will take time $\Omega(N^3)$, because for every pair $(a,b) \in [N]^2$, 
the algorithm will have to verify in $\Omega(N)$ time that the constraints
$(ii)$ forbid all $c = 2i-1, i\in [N]$, the constraints $(iii)$ 
forbid all $c=2i, i\in [N]$, and the constraint $(iv)$ forbid $c > 2N$.

But we can do better by remembering inferences that we have made.  Fix
a value $a=1, b=1$. \ms recognizes in $O(N)$-time that there is no $c$
for which $(a,b,c)$ is active. \ms is slightly smarter: it looks at
constraints of the type $(ii), (iii), (iv)$ (for $b=1$) and concludes
in $O(N)$-time that every tuple satisfying those constraints also
satisfies the constraint $\langle *, 1, (0, +\infty)\rangle$.  \ms
remembers this inference by inserting the inferred constraint into the
\cds. Then, for $a \geq 2$, it takes only $O(1)$-time to conclude that
no tuple of the form $(a, 1, c)$ can be active. It does this inference
by inserting constraint $\langle a, 1, (0,+\infty)\rangle$, which is
merged with $(i)$ to become $\langle a, 1, (-\infty,
+\infty)\rangle$. Overall, we need only $O(N^2)$-time to reach the
same conclusion as the $\Omega(N^{3})$ brute-force strategy.

%% Do we need to make the argument that eager won't work? It seems
%% distracting to me

%% The above explains why memoization makes sense,
%% but one might come to the conclusion that we can do the inference
%% (i.e. specialization) at constrain insertion time instead of at a
%% call to $\getpp$. However, aggressive inference does not quite work
%% because the \cds has to handle dynamic constraints inserted in an
%% online fashion. Suppose, for example, constraints $(i)$ and $(ii)$
%% are inserted before $(iii)$ and $(iv)$. If we specializes all
%% constraints $(ii)$ down, we will have to insert, for each $a, b, i
%% \in [N]$, the constrains $\langle a, b, (2i-2, 2i)\rangle$, for a
%% total of $\Omega(N^3)$-time.
\end{example}

\subsection{Patterns}
\label{subsec:pattern-neo}

Recall that $\getpp$ returns a tuple $\mv t = (t_1,\dots,t_n)
\in \outspace$ such that $\mv t$ does not satisfy any of the
constraints stored in the \cds. We find $\mv t$ by computing $t_1,
t_2, \dots, t_n$, one value at a time, backtracking if necessary. We
need some notation to describe the algorithm and the properties that
we exploit.

Let $0 \leq k \leq n$ be an integer. A vector $\mv p = \langle p_1,
\dots, p_k \rangle$ for which $p_i \in \mathbb N \cup \{*\}$ is
called a {\em pattern}.  The number $k$ is the {\em length} of the
pattern.  If $p_i \in \mathbb N$ then it is an {\em equality
  component} of the pattern, while $*$ is a {\em wildcard component}
of the pattern.

A node $u$ at depth $k$ in the tree $\ctree$ can be identified by a 
pattern of length $k$ corresponding naturally to the labels on the path from 
the root of $\ctree$ down to node $u$. 
The pattern for node $u$ is denoted by $P(u)$.
In particular, $P(\text{root}) = \epsilon$, the empty pattern.

Let $\mv p = \langle p_1, \dots, p_k\rangle$ be a pattern. Then, a 
{\em specialization} of $\mv p$ is another pattern $\mv p' = \langle p'_1,
\dots, p'_k\rangle$ {\em of the same length} 
for which $p'_i = p_i$ whenever $p_i \in \mathbb N$. 
In other words, we can get a specialization of $\mv p$ by changing some of the $*$ 
components into equality components.
If $\mv p'$ is a specialization of $\mv p$, then $\mv p$ is a 
{\em generalization} of $\mv p'$.
For two nodes $u$ and $v$ of the \cds, if $P(u)$ is a specialization of $P(v)$, 
then we also say that node $u$ is a specialization of node $v$.

The specialization relation defines a partially ordered set.  When
$\mv p'$ is a specialization of $\mv p$, we write $\mv p' \preceq \mv
p$.  If in addition we know $\mv p' \neq \mv p$, then we write $\mv p'
\prec \mv p$. 

Let $G(t_1, \dots, t_i)$ be the {\em principal filter} generated by
$(t_1,\dots,t_i)$ in this partial order, i.e., it is the set of all nodes $u$ of
the \cds such that $P(u)$ is a generalization of $\langle t_1, \dots,
t_i \rangle$ and that $u.\intv \neq \emptyset$.  The key property of
constraints that we exploit is summarized by the following
proposition.

\begin{prop}
Using the notation above, for a $\beta$-acyclic query, there exists a
GAO such that for each $t_1,\dots,t_i$ the principal filter
$G(t_1,\dots,t_i)$ is a chain.
\label{prop:chain}
\end{prop}

Recall that a chain is a totally ordered set. In particular, $G =
G(t_1,\dots,t_i)$ has a smallest pattern $\bar{\mv p}$ (or bottom
pattern).  Note that these patterns in $G$ might come from constraints
inserted from relations, constraints inserted by the outputs of the
join, or even constraints inserted due to backtracking. Thinking of
the constraints geometrically, this condition means that the
constraints form a collection of axis-aligned affine subspaces of
$\outspace$ where one is contained inside another.

In Appendix~\ref{app:sec:beta-acyclic}, we prove
Proposition~\ref{prop:chain} using a result of Brouwer and
Kolen~\cite{brouwer-kolen-1980}. The class of GAOs in the proposition is called
a {\em nested elimination order}. We show that there exists a GAO that
is a nested elimination order if and only if the query is
$\beta$-acyclic. We also show that $\beta$-acyclicity and this GAO can
be found in polynomial time.

\subsection{The $\getpp$ Algorithm} Algorithm~\ref{algo:getpp-beta} 
describes $\getpp$ algorithm specialized to $\beta$-acyclic
queries. In turn, this algorithm uses Algorithm~\ref{algo:ctree-nextvalue-beta}, 
which is responsible for efficiently inferring constraints imposed by patterns 
above this level.  We walk through the steps of the algorithm below.

Initially, let $v$ be the root node of the \cds.  We set $t_1 =
v.\intv.\nxt(-1)$. This is the smallest value $t_1$ that
does not belong to any interval stored in $v.\intv$. We work under the
implicit assumption that any interval inserted into $\ctree$ that
contains $-1$ must be of the form $(-\infty, r)$, for some $r \geq
0$. This is because the domain values are in $\mathbb N$.  In
particular, if $t_1 = +\infty$ then the constraints cover the entire
output space $\outspace$ and $\NULL$ can be returned.

 Inductively, let $(t_1, \dots, t_i)$, $i \geq 1$, be the {\em prefix}
 of $\mv t$ we have built thus far. Our goal is to compute
 $t_{i+1}$. What we need to find is a value $t_{i+1}$ such that
 $t_{i+1}$ does not belong to the intervals stored in $u.\intv$ {\em
   for every} node $u \in G(t_1,\dots,t_i)$. For this, we call
 algorithm~\ref{algo:ctree-nextvalue-beta} that uses
 Prop.~\ref{prop:chain} to efficiently find $t_{i+1}$ or return that
 there is no such $t_{i+1}$. We defer its explanation for the
 moment. We only note that if such a $t_{i+1}$ cannot be found
 (i.e. if $t_{i+1} = +\infty$ is returned after the search), then we
 have to {\em backtrack} because what that means is that every tuple
 $\mv t$ that begins with the prefix $(t_1,\dots,t_i)$ satisfies some
 constraint stored in $\ctree$.  Line \ref{line:backtrack} of
 Algorithm~\ref{algo:getpp-beta} shows how we backtrack. In
 particular, we save this information (by inserting a new constraint
 into the CDS) in Line~\ref{line:backtrack} to avoid ever
 exploring this path again.

\paragraph*{Next Chain Value.} 
The key to Algorithm~\ref{algo:ctree-nextvalue-beta} 
is that such a $t_{i+1}$ can be found efficiently since one only needs
to look through a chain of constraint sets. We write $\mv p \precdot \mv p'$
if $\mv p \prec \mv p'$ and there is no pattern $\mv p''$ such that
$\mv p \prec \mv p'' \prec \mv p'$. 
Every interval from a node $u\in G$ higher up in the chain infers 
an interval at a node lower in the chain. For instance, in
Example~\ref{ex:memoization}, the chain $G$ consists of three nodes
$\langle a, b\rangle$, $\langle *, b\rangle$, and $\langle *, *\rangle$.
Further, every constraint of the form $\langle *, *, (2i-1,2i+1)\rangle$
infers a more specialized constraint of the form 
$\langle *, b, (2i-1,2i+1)\rangle$,
which in turns infers a constraint of the form 
$\langle a, b, (2i-1,2i+1) \rangle$.
Hence, if we infer every single constraint downward from the top pattern
to the bottom pattern, we will be spending a lot of time. 
The idea of Algorithm~\ref{algo:ctree-nextvalue-beta} is to infer as large
of an interval as possible from a node higher in the chain before specializing
it down.
Our algorithm will ensure that whenever we infer a new constraint 
(line~\ref{line:new-constraint-beta} 
%\yell{Line number is not showing up. Don't know why.}
of 
Algorithm~\ref{algo:ctree-nextvalue-beta}), this constraint subsumes an old
constraint which will never be charged again in a future inference.

\begin{algorithm}[t]
\caption{$\cds.\getpp()$ for $\beta$-acyclic queries}
\label{algo:getpp-beta}
\begin{algorithmic}[1]
\Require{A $\ctree$ $\cds$}
%\Ensure{Returns a probe point $\mv t=(t_1,\dots,t_n)$}
\Statex

\State $i \gets 0$
\While {$i < n$}
  \State {\small $G \gets \setof{u \in \cds}{(t_1,\dots,t_i) \preceq P(u) \text{ and }
  u.\intv \neq \emptyset}$}
  %\Comment{Short for $G(t_1,\dots,t_i)$}
  \If {($G = \emptyset$)} %\Comment{$G$ remains $\emptyset$ for all later $i$}
    \State $t_{i+1} \gets -1$ \label{line:beta-Gempty}
    \State $i \gets i+1$
  \Else
    \State Let $\bar{\mv p} = \langle \bar p_1, \dots, \bar p_i \rangle$ be the
    {\em bottom} of $G$
    \State Let $\bar u \in \cds$ be the node for which $P(\bar u) = \bar{\mv p}$
    \State $t_{i+1} \gets \cds.\nextvalue(-1, \bar u, G)$
    \label{line:initial-call-beta}
    \State $i_0 \gets \max\{ k \suchthat k \leq i, \bar p_k \neq *\}$
    \If {($t_{i+1} = +\infty$) and $i_0=0$}
      \State \Return \NULL \Comment{No tuple $\mv t$ found}
%\label{step:getpp-full}
    \ElsIf {($t_{i+1} = +\infty$)} 
      \State {\small $\cds.\insconst(\langle \bar p_1,\dots, \bar p_{i_0-1}, (\bar
              p_{i_0}-1, \bar p_{i_0}+1) \rangle)$} %\Comment{Subtle!}
      \label{line:backtrack}
      \State $i \gets i_0-1$ \Comment{Back-track} 
    \Else
      \State $i \gets i+1$ \Comment{Advance $i$}
    \EndIf
  \EndIf
\EndWhile
\State \Return $\mv t = (t_1,\dots,t_n)$
\end{algorithmic}
\end{algorithm}

\begin{algorithm}[t]
\caption{$\cds.\nextvalue(x, u, G)$, where $G$ is a chain}
\label{algo:ctree-nextvalue-beta}
\begin{algorithmic}[1]
\Require{A $\ctree$ $\cds$, a node $u \in G$}
\Require{A chain $G$ of nodes, and a starting value $x$}
\Ensure{the smallest value $y\geq x$ not covered by {\em any}
$v.\intv$, for all $v \in G$ such that $P(u) \preceq P(v)$}
\Statex
\If {there is no $v \in G$ for which $P(u) \precdot P(v)$} \Comment{At the top
of the chain $G$}
\label{step:nextval-end-beta}
  \State \Return $u.\intv.\nxt(x)$ \label{line:nextvalue-basecase}
\Else
  \State $y \gets x$
  \Repeat
    \State Let $v\in G$ such that $P(u) \precdot P(v)$ 
    \State \Comment{Next node up the
    chain}
    \State $z \gets \cds.\nextvalue(y, v, G)$ \label{line:z-beta}
    \State \Comment{first ``free value'' $\geq y$ at all nodes up the chain}
    \State $y \gets u.\intv.\nxt(z)$ \label{line:y-beta}
    \State \Comment{first ``free value'' $\geq z$ at $u$}
  \Until {$y = z$}
  \State $\cds.\insconst(\langle P(u), (x-1,y)\rangle)$ 
  \label{line:new-constraint-beta}
  %\State \Comment{All values from $x$ to $y-1$ are not available}
  \State \Return $y$
\EndIf
\end{algorithmic}
\end{algorithm}

\subsection{Runtime Analysis} 

The proofs of the following main results are in 
Appendix~\ref{app:sec:beta-acyclic}.

\begin{lmm}
Suppose the input query $Q$ is $\beta$-acyclic. Then, there exists a
GAO such that each of the two operations $\getpp$ and $\insconst$ of
$\ctree$ takes amortized time $O(n2^n\log W)$, where $W$ is the total
number of constraints ever inserted.
\label{lmm:amortized-analysis-acyclic-Q}
\end{lmm}

The above lemma and Theorem~\ref{thm:analyze-outer-algorithm} leads directly 
to one of our main results.

\begin{cor}[Restatement of Theorem~\ref{thm:io-beta-neo}]
Suppose the input query is $\beta$-acyclic then there exists a GAO
such that $\ms$ computes its output in time
\[ O\left( 2^nm^2n\left(4^r|\cert| + Z \right) \log N \right). \]
In particular, its data-complexity runtime is essentially optimal
in the certificate complexity world: $\tilde O(|\cert|+Z)$.
\label{cor:thm:io-beta-neo}
\end{cor}

Beyond $\beta$-acyclic queries, we show that we cannot do better modulo a
well-known complexity theoretic assumption.

\bprop[Re-statement of Proposition~\ref{prop:no-io-non-beta}]
\label{restatement:prop:no-io-non-beta}
Unless the $3$\textsf{SUM} problem can be solved in sub-quadratic
time, for any $\beta$-{\em cyclic} query $Q$ in any GAO, there does
not exist an algorithm that runs in time $O(|\cert|^{4/3-\eps} + Z)$
for any $\eps>0$ on all instances.  
\eprop

\paragraph*{Comparison with Worst-Case Optimal Algorithms}
It is natural to wonder if Yannakakis' worst-case optimal algorithm
for $\alpha$-acyclic queries or the worst-case optimal algorithms
of~\cite{DBLP:conf/pods/NgoPRR12} (henceforth, \nprr)
or~\cite{DBLP:journals/corr/abs-1210-0481} (henceforth \lb) can
achieve runtimes of $\tilde{O}(|\cert|+Z)$ for $\beta$-acyclic
queries. We outline the intuition about why this cannot be the case.

Yannakakis' algorithm performs pairwise semijoin reducers.
If we pick an instance where $|\cert|=o(N)$ such that there is a relation pair
involved each with size $\Omega(N)$, then Yannakakis's algorithm will
exceed the bound. For $\nprr$ and $\lb$, consider the family of
instances in which one computes all paths of length $\ell$ (some constant)
%(for simplicity, let us assume that $\ell$ is a constant) 
in a directed graph $G=(V,E)$ (this can be realized by a ``path" query of length
$\ell$ where the relations are the edge set of $G$). Now consider the
case where the longest path in $G$ has size at most $\ell-1$. In this
case the output is empty and since each relation is $E$, we have
$|\cert|\le O(|E|)$ and by Corollary~\ref{cor:thm:io-beta-neo}, we
will run in time $\tilde{O}(|E|)$. Hence, when $G$ has many
paths (at least $\omega(|E|)$) of length at most $\ell$, 
then both $\nprr$ and $\lb$ will have to explore all $\omega(|E|)$ paths
leading to an $\omega(|\cert|)$ runtime.

Appendix~\ref{app:sec:counter-examples} presents a rich family of
$\beta$-acyclic queries and a family of instances that combines both
of the ideas above to show that all the three worst-case optimal
algorithms can have arbitrarily worse runtime than \ms. In particular,
even running those worst-case algorithms in parallel
is not able to achieve the certificate-based guarantees.

%%%%%%%%%%%%%%%%%%%%%%%%%%%%%%%%%%%%%%%%%%%%%%%%%%%%%%%%%%%%%%%%%%%%%%%%%%
\section{Extensions}
\label{sec:tw}

We extend in two ways: queries with bounded tree width and we describe
faster algorithms for the triangle query.

\subsection{Queries with bounded tree-width}
%%%%%%%%%%%%%%%%%%%%%%%%%%%%%%%%%%%%%%%%%%%%%%%%%%%%%%%%%%%%%%%%%%%%%%%%%%
While Proposition~\ref{prop:no-io-non-beta} shows that $O(|\cert|^{4/3-\eps}+Z)$-time
is not achievable for $\beta$-cyclic queries, we are able to show the following
analog of the treewidth-based runtime under the traditional worst-case
complexity notion \cite{journals/ai/DechterP89, MR985145}.

\bthm[\ms for bounded treewidth queries]
\label{thm:tw}
Suppose the GAO has {\em elimination width} bounded by $w$,
Then, $\ms$ runs in time
\[ O\left( m^3n^34^n \left( nm^{w+1} 8^{n(w+1)} |\cert|^{w+1} + Z\right)       
     \log N \right).
\]
In particular, if we ignore the dependence on the query size, the runtime is
$\tilde O \left( |\cert|^{w+1} + Z \right)$.
Further, if the input query $Q$ has treewidth bounded by $w$, then there
exists a GAO for which $\ms$ runs in the above time.
\ethm

The overall structure of the algorithm remains identical to the
$\beta$-acyclic case, the only change is in $\getpp$.  The $\getpp$
algorithm for general queries remains very similar in structure to
that of the $\beta$-acyclic case (Algorithm~\ref{algo:getpp-beta}),
and if the input query is $\beta$-acyclic (with a nested elimination
order as the GAO), then the general $\getpp$ algorithm is {\em
  exactly} Algorithm~\ref{algo:getpp-beta}.  The new issue we have to
deal with is the fact that the poset $G$ at each depth is not
necessarily a chain.  Our solution is simple: we mimic the behavior of
Algorithm~\ref{algo:getpp-beta} on a shadow of $G$ that {\em is} a
chain and make use of both the algorithm and the analysis for the
$\beta$-acyclic case.
%
%%%% Removing discussion about shadow to save space. --Atri
\iffalse
The shadow of $G$ is constructed as follows. 
Let $u_1,\dots,u_k$ be an arbitrary linearization of nodes in $G$, 
i.e. if $1 \leq i<j\leq k$, then either $P(u_i) \preceq P(u_j)$ or 
$P(u_i)$ and $P(u_j)$ are incomparable using the relation $\preceq$.
A linearization always exists because $\preceq$ is a partial order.
Now, for $j \in [k]$, define the patterns
$\bar P(u_j) = \bigwedge_{i=j}^k P(u_i).$
Here $\wedge$ denotes meet under the partial order $\preceq$.
Then, obviously the shadow patterns form a chain:
\[ \bar P(u_1) \preceq \bar P(u_2) \preceq \cdots \preceq \bar P(u_k). \]
Note that it is possible for $\bar P(u_i) = \bar P(u_j)$ for $i\neq j$.
For each node $u$, we will operate as if its pattern was actually $\bar P(u)$.
\fi
Appendix~\ref{app:sec:tw} contains all the algorithm details, and the proofs
of the above theorem, along with the following negative result.

It is natural to wonder if Theorem~\ref{thm:tw} is tight. In addition to the
obvious $\Omega(Z)$ dependency, the next result indicates that the dependence 
on $w$ also cannot be avoided, {\em even if} we just look at the class of
$\alpha$-acyclic queries. 

\bprop
\label{prop:tw-eth}
Unless the {\em exponential time hypothesis} is false, for every large
enough constant $k>0$, there is an $\alpha$-acyclic query $Q_k$ for
which there is no algorithm with runtime $|\cert|^{o(k)}$. Further,
$Q_k$ has treewidth $k-1$.  \eprop

Our analysis of \ms is off by at most $1$ in the exponent.
\bprop
\label{prop:ms-tw-lb}
For every $w\ge 2$, there exists an ($\alpha$-acyclic) query $Q_w$ with treewidth $w$ with the following property. For every possible global ordering of attributes, there exists an (infinite family of) instance on which the \ms\ algorithm takes $\Omega(|\cert|^w)$ time.
\eprop

\subsection{An implementation of \ms}
% Input size v.s. Cert size table
\begin{figure}
\centering
\begin{tabular}{|c||c|c||c|c||c|c||}
\hline
Query & \multicolumn{2}{c||}{com-Orkut} & \multicolumn{2}{c||}{soc-Epinions1}
& \multicolumn{2}{c||}{soc-LiveJournal1}
\\
\cline{2-7}
& $N$ & $|\cert|$ & $N$ & $|\cert|$ & $N$ & $|\cert|$ \\
\hline
%Star & {\small 351,567,487} & {\small 214,086} & {\small 1,526,790} & {\small 1,067} & {\small 207,000,694} & {\small 172,109} \\
%$3$-path & {\small 351,567,487} & {\small 118,637} & {\small 1,526,790} & {\small 842}  & {\small 207,000,694} & {\small 138,011}  \\
%Tree & {\small 468,752,570} & {\small 2,785,035} &  {\small 2,035,627} & {\small 3,441} & {\small 275,994,467} &  {\small 2,683,626} \\
Star & 352M & 214K & 1.5M & 1,067 & 207M & 172K \\
$3$-path & 352M & 119K & 1.5M & 842  & 207M & 138K  \\
Tree & 469M & 2.8M &  2M & 3,441 &  276M &  2.7M \\
\hline
\end{tabular}
\caption{Input size ($N$) versus Certificate size ($|\cert|$). Units are Million(M) and Thousand(K). The
  three graph datasets are from Orkut, Epinions,
  and LiveJournal network
  \url{http://snap.stanford.edu/data/}.}
\label{fig:real:CertSize}
\end{figure}

With the help of LogicBlox, we implemented \ms inside the LogicBlox
engine. Our results are preliminary: it is implemented for main memory
data and all experiments are run in a multi-threaded mode. We run
three queries: a star query, a small path query, and a tree query,
which are described below, on three data sets Orkut online social
network, Who-trusts-whom network of Epinions.com, and LiveJournal
online social network.
\begin{itemize}
\item Star query: $Q = R_1(A) \Join S(A,B) \Join S(A,C) \Join S(A,D) \Join R_2(B) \Join R_3(C) \Join R_4(D)$.
\item $3$-path query: $Q = S(A,B) \Join S(B,C) \Join S(C,D) \Join R_5(A) \Join R_6(B) \Join R_7(C) \Join R_8(D)$.
\item Tree query: $Q = S(A,B) \Join S(B,C) \Join S(B,D) \Join S(D,E) \Join R_9(A) \Join R_{10}(C) \Join R_{11}(D) \Join R_{12}(E)$. 
\end{itemize}

For each query and each dataset, relation $S$ is a graph dataset,
while every $R_i$ relation contains a subset of vertices from that
graph dataset, where every vertex is chosen with a probability
$0.001$.  Figure~\ref{fig:real:CertSize} shows the input size versus
certificate size on different queries and different graph
datasets. The certificate size is measured by counting the number of
$\findgap$ operations during computing join queries. These numbers
show that certificate size is very small compared to input size and so
it indicates that a practical implementation might be obtained.

\subsection{The Triangle Query}

We consider the triangle query $Q_{\triangle}=R(A,B)\Join S(B,C)\Join
T(A,C)$ that can be viewed as enumerating triangles in a given
graph. Using the \cds described so far, \ms computes this query in
time $\tilde{O}(|\cert|^2+Z)$, and this analysis is tight.\footnote{A
  straightforward application of our more general analysis given in
  Theorem~\ref{thm:tw}, which gives $\tilde O(|\cert|^3+Z)$.} The
central inefficiency is that the $\cds$ wastes time determining that
many tuples with the same prefix $(a,b)$ have been ruled out by
existing constraints. In particular, the $\cds$ considers all possible
pairs $(a,b)$ (of which there can be $\Omega(|\cert|^2)$ of them). By
designing a smarter \cds, our improved \cds explores $O(|\cert|)$ such
pairs. We can prove the following result.
(The details are in Appendix~\ref{app:triang}.)
\begin{thm}
\label{thm:triang}
We can solve the triangle query, $Q_{\Delta}$ in time
$O\left(\left(|\cert|^{3/2}+Z\right)\log^{7/2}{N}\right)$.
\end{thm}

%%%%%%%%%%%%%%%%%%%%%%%%%%%%%%%%%%%%%%%%%%%%%%%%%%%%%%%%%%%%%%%%%%%%%%%%%%
\section{Related Work}
\label{sec:related:work}
%%%%%%%%%%%%%%%%%%%%%%%%%%%%%%%%%%%%%%%%%%%%%%%%%%%%%%%%%%%%%%%%%%%%%%%%%%

Our work touches on a few different areas, and we structure the related
work around each of these areas: join processing, certificates for set
intersection, and complexity measures that are finer than worst-case
complexity.

\subsection{Join Processing}
Many positive and negative results regarding conjunctive query
evaluation also apply to natural join evaluation. On the negative
side, both problems are $\mathsf{NP}$-hard in terms of expression
complexity \cite{DBLP:conf/stoc/ChandraM77}, but are easier in terms
of data complexity \cite{DBLP:conf/stoc/Vardi82} (when the query is
assumed to be of fixed size). They are $\mathsf{W}[1]$-complete and
thus unlikely to be fix-parameter tractable
\cite{DBLP:conf/pods/PapadimitriouY97,DBLP:conf/pods/Grohe01}.
%%%% Removing the following comment because of space reasons
%:
%the $\mathsf{W}[1] \neq \mathsf{FPT}$ assumption is the fixed-parameter analog
%of $\mathsf{P} \neq \mathsf{NP}$ in the parameterized complexity world
%\cite{MR1656112,MR2238686}.

On the positive side, a large class of conjunctive queries (and thus
natural join queries) are tractable. In particular, the classes of
acyclic queries and bounded treewidth queries can be evaluated
efficiently \cite{DBLP:conf/vldb/Yannakakis81,
  DBLP:journals/tcs/ChekuriR00, DBLP:journals/jcss/GottlobLS02,
  DBLP:journals/jacm/FlumFG02, DBLP:journals/jcss/Willard02}.  For
example, if $|q|$ is the query size, $N$ is the input size, and $Z$ is
the output size, then Yannakakis' algorithm can evaluate acyclic
natural join queries in time $\tilde O(\poly(|q|)(N\log N + Z))$.
Acyclic conjunctive queries can also be evaluated efficiently in the
I/O model \cite{DBLP:conf/pods/PaghP06}, and in the RAM model even
when there are inequalities \cite{DBLP:journals/jcss/Willard02}.  For
queries with treewidth $w$, it was recognized early on that a runtime
of about $\tilde O(N^{w+1} + Z)$ is
attainable~\cite{journals/ai/DechterP89,
  Freuder:1990:CKS:1865499.1865500}. our result strictly generalizes
these results. In Appendix~\ref{app:sec:counter-examples}, we show
that Yannakakis' algorithm does not meet our notion of certificate
optimality.

The notion of treewidth is loose for some queries. For instance, if we
replicate each attribute $x$ times for every attribute, then the
treewidth is inflated by a factor of $x$; but by considering all
duplicate attributes as one big compound attribute the runtime should
only be multiplied by a polynomial in $x$ and there should not be a
factor of $x$ in the exponent of the runtime.  Furthermore, there is
an inherent incompatibility between treewidth and acyclicity: an
acyclic query can have very large treewidth, yet is still tractable. A
series of papers~\cite{DBLP:journals/tcs/ChekuriR00, MR2351517,
  Gottlob:2009:GHD:1568318.1568320, DBLP:journals/jcss/GottlobLS02,
  DBLP:journals/jacm/FlumFG02} refined the treewidth notion leading to
generalized hyper treewidth \cite{DBLP:journals/jcss/GottlobLS02} and
ultimately {\em fractional hypertree width}
\cite{DBLP:journals/talg/Marx10}, which allows for a unified view of
tractable queries.  (An acyclic query, for example, has fractional
hypertree width at most $1$.)

The fractional hypertree width notion comes out of a recent tight
worst-case output size bound in terms of the input relation sizes
\cite{FOCS:AtsGroMar08}. An algorithm was presented that runs in time
matching the bound, and thus it is worst-case optimal in
\cite{DBLP:conf/pods/NgoPRR12}.  Given a tree decomposition of the
input query with the minimum fractional edge cover over all bags, we
can run this algorithm on each bag, and then Yannakakis algorithm
\cite{DBLP:conf/vldb/Yannakakis81} on the resulting bag relations,
obtaining a total runtime of $\tilde O(N^{w^*}+Z)$, where $w^*$ is the
fractional hyper treewidth.  The {\em leap-frog triejoin} algorithm
\cite{DBLP:journals/corr/abs-1210-0481} is also worst-case optimal and
runs fast in practice; it is based on the idea that we can efficiently
skip unmatched intervals.  The indices are also built or selected to
be consistent with a chosen GAO. In the
Appendix~\ref{app:sec:counter-examples}, we show that neither Leapfrog
nor the algorithm from~\cite{DBLP:conf/pods/NgoPRR12} can achieve the
certificate guarantees of \ms for $\beta$-acyclic queries.

\vspace{-7pt}
\paragraph*{Notions of acyclicity} 
There are at least five notions of acyclic hypergraphs, four of which
were introduced early on in database theory (see e.g,
\cite{DBLP:journals/jacm/Fagin83}), and at least one new one
introduced recently \cite{Duris2012}. The five notions are {\em not}
equivalent, but they form a strict hierarchy in the following way:
%\begin{multline*}
\[ 
    \text{Berge-acyclicity} \varsubsetneq \gamma\text{-acyclicity}               
   \varsubsetneq \text{jtdb} \varsubsetneq \beta\text{-acyclicity} 
   \varsubsetneq \alpha\text{-acyclicity}% = \text{acyclicity} 
\]
%\end{multline*}
Acyclicity or $\alpha$-acyclicity \cite{Beeri:1981:PAD:800076.802489,
  Beeri:1983:DAD:2402.322389, Fagin:1982:SUR:319732.319735,
  Goodman:1982:TQS:319758.319775, Maier:1982:CAH:588111.588118} was
recognized early on to be a very desirable property of data base
schemes; in particular, it allows for a data-complexity optimal
algorithm in the worst case \cite{DBLP:conf/vldb/Yannakakis81}.
However, an $\alpha$-acyclic hypergraph may have a sub-hypergraph that
is not $\alpha$-acyclic. For example, if we take {\em any} hypergraph
and add a hyperedge containing all vertices, we obtain an
$\alpha$-acyclic hypergraph.  This observation leads to the notion of
$\beta$-acyclicity: a hypergraph is $\beta$-acyclic if and only if
every one of its sub-hypergraph is ($\alpha$-) acyclic
\cite{DBLP:journals/jacm/Fagin83}.  It was shown (relatively) recently
\cite{ordyniak_et_al:LIPIcs:2010:2855} that {\sc sat} is in
$\mathsf{P}$ for $\beta$-acyclic CNF formulas and is
$\mathsf{NP}$-complete for $\alpha$-acyclic CNF formulas.  Extending
the result, it was shown that negative conjunctive queries are
poly-time solvable if and only if it is $\beta$-acyclic
\cite{braultbaron:LIPIcs:2012:3669}.  The separation between
$\gamma$-acyclicity and $\beta$-acyclicity showed up in logic
\cite{Duris:2008:HAE:1381308.1382241}, while Berge-acyclicity is
restrictive and, thus far, is of only historical interest
\cite{Berge:1985:GH:1096893}.

\vspace{-7pt}
\paragraph*{Graph triangle enumeration}
In social network analysis, computing and 
listing the number of triangles in a graph is at the heart of the clustering
coefficients and transitivity ratio. There are four decades of research on
computing, estimating, bounding, and lowerbounding the number of triangles and the 
runtime for such algorithms
\cite{DBLP:conf/www/SuriV11,
DBLP:journals/im/KolountzakisMPT12,
Vassilevska:2009:FMC:1536414.1536477,
STOC:VasWil06, 
MR599482, 
DBLP:journals/siamcomp/ItaiR78}. 
This problem can easily be reduced to a join
query of the form $Q = R(A,B) \Join S(B, C) \Join T(A, C)$.

\subsection{Certificates for Intersection} 
The problem of finding the union and intersection of two sorted arrays
using the fewest number of comparisons is well-studied, dated back to
at least Hwang and Lin \cite{MR0297088} since 1972. In fact, the idea
of skipping elements using a binary-search jumping (or leap-frogging)
strategy was already present in \cite{MR0297088}.  
Demaine et
al. \cite{DBLP:conf/soda/DemaineLM00} used the leap-frogging strategy
for computing the intersection of $k$ sorted sets. They introduced the
notion of proofs to
capture the intrinsic complexity of such a problem.  Then, the idea of
gaps and certificate encoding were introduced to show that their
algorithm is average case optimal. (See Appendix~\ref{app:sec:C-vs-dlm} for a
more technical discussion.)

\dlm's notion of proof inspired another adaptive complexity notion for the 
set intersection problem called partition certificate by Barbay and Kenyon in
\cite{DBLP:conf/soda/BarbayK02, DBLP:journals/talg/BarbayK08}, where instead
of a system of inequalities essentially a set of gaps is used to encode and
verify the output. Barbay and Kenyon's idea of a partition
certificate is very close to the set of intervals that \ms outputs. 
In the analysis of \ms in Appendix~\ref{app:sec:intersection} for the set 
intersection problem, we (implicitly) show a correspondence between these partition 
certificates and \dlm's style proofs.
In addition to the fact that join queries are more general than set
intersection, our notion of certificate is value-oblivious; our 
certificates do not depend on specific values in the domain, while
Barbay-Kenyon's partition certificate does.

It should be noted that these lines of inquiries are not only of theoretical
interest. They have yielded good experimental results in text-datamining and
text-compression\cite{DBLP:conf/sccc/BarbayL09}.\footnote{We thank J\'er\'emy Barbay
for bringing these references to our attention.}

%%%

\subsection{Beyond Worst-case Complexity}

There is a fairly large body of work on analyzing algorithms with more refined
measures than worst-case complexity. (See, e.g., the
excellent lectures by Roughgarden on this topic~\cite{tim-notes}.)
This section recalls the related works that are most closely related to ours.

A fair amount of work has been done in designing {\em adaptive} algorithms for
sorting~\cite{adaptive-sort}, where the goal is to design a sorting algorithm
whose runtime (or the number of comparisons) matches a notion of difficulty of
the instance (e.g. the number of inversions, the length of longest monotone
subsequence and so on -- the survey~\cite{adaptive-sort} lists at least eleven
such measures of {\em disorder}). This line of work is similar to ours in the
sense that the goal is to run in time proportional to the difficulty of the
input. The major difference is that in these lines of work the main goal is to
avoid the logarithmic factor over the linear runtime whereas in our work, our
potential gains are of much higher order and we ignore log-factors. 

Another related line of work is on self-improving algorithms of Ailon et
al.~\cite{self-improving}, where the goal is to have an algorithm that runs on inputs that are drawn i.i.d. from an {\em unknown} distribution and in expectation converge to a runtime that is related to the entropy of the distribution. In some sense this setup is similar to online learning while our work requires worst-case per-instance guarantees.

The notion of instance optimal join algorithms was (to the best of our
knowledge) first explicitly studied in the work of Fagin et
al.~\cite{fagin-io}. The paper studies the problem of computing the
top-$k$ objects, where the ranking is some aggregate of total ordering
of objects according to different attributes. (It is assumed that the
algorithm can only iterate through the list in sorted order of
individual attribute scores.) The results in this paper are stronger
than ours since Fagin et al. give $O(1)$-optimality ratio (as opposed
to our $O(\log{N})$-optimality ratio). On the other hand the results
in the Fagin et al. paper are for a problem that is arguably narrower
than the class we consider of join algorithms.

The only other paper with provable instance-optimal guarantees that we are aware of is the Afshani et al. results on some geometric problems~\cite{geometric-io}. Their quantitative results are somewhat incomparable to ours. On the one hand their results get a constant optimality ratio: on the other hand, the optimality ratio is only true for {\em order oblivious} comparison algorithms (while our results with $O(\log N)$ optimality ratio hold against all comparison-based algorithms).

%\cmr{Put in the adaptive sorting stuff here ? (\# of inversions etc.)}
%\cmr{Put Instance Optimality here, Fagin's algorithm and references to
%  Tim's notes. OK to repeat citations. We have to cite that geometric
%  paper that was pointed out to us. }

%\begin{itemize}
%\item Instance optimality and related notions (Fagin, recent geometric
%  results, Tim's lecture notes). Reviewer point out a paper, and we
%  need to not forget to cite it!! (I think we do :) )
%
%\item Mention property testing/sub-linear algorithms stuff. Mention that they're still worst-case results and no sub-linear algorithm in worst-case is possible even for set intersection (cf.~\cite{GKS07}). 
%
%\end{itemize}

%% \paragraph*{Sub-linear algorithms} \ms\ runs in sub-linear time when $|\cert|$ is $o(N)$. Sub-linear algorithms (and property testing) are active fields of research (see e.g. surveys by Kumar and Rubinfeld~\cite{sublinear-survey} and Ron~\cite{pt-survey}). In such algorithms, the goal is compute the output in either sub-linear time or with sub-linear many accesses to the input. The place where these works differ from ours is that these sublinear algorithm are for worst-case inputs. It is very easy to see that there can be no worst-case sublinear time join algorithm. This is impossible even for the set intersection case and can be easily proven given the strong lower bounds on communication complexity for the set disjointness problem. (This was explicitly observed in~\cite{GKS07} in the context of data stream algorithms.)

%%%%%%%%%%%%%%%%%%%%%%%%%%%%%%%%%%%%%%%%%%%%%%%%%%%%%%%%%%%%%%%%%%%%%%%%%%
\section{Conclusion and Future Work}
\label{sec:conclusions}
%%%%%%%%%%%%%%%%%%%%%%%%%%%%%%%%%%%%%%%%%%%%%%%%%%%%%%%%%%%%%%%%%%%%%%%%%%

%% \cmr{We need to add statements about why previous algorithms do not
%%   meet the certificate complexity run-times.}

 We described the \ms algorithm for processing join queries on data
 that is stored ordered in data structures modeling traditional
 relational databases. We showed that \ms can achieve stronger runtime
 guarantees than previous algorithms; in particular, we believe \ms is
 the first algorithm to offer beyond worst-case guarantees for
 joins. Our analysis is based on a notion of certificates, which
 provide a uniform measure of the difficulty of the problem that is
 independent of any algorithm. In particular, certificates are able to
 capture what we argue is a natural class of comparison-based join
 algorithms.

Our main technical result is that, for $\beta$-acyclic queries there
is some GAO such that \ms runs in time that is linear in the
certificate size. Thus, \ms is optimal (up to an $O(\log N)$ factor)
among comparison-based algorithms. Moreover, the class of
$\beta$-acyclic queries is the boundary of complexity in that we show
no algorithm for $\beta$-{\em cyclic} queries runs in time linear in
the certificate size. And so, we are able to completely characterize
those queries that run in linear time for the certificate and hence
are optimal in a strong sense. Conceptually, certificates change the
complexity landscape for join processing as the analogous boundary for
traditional worst-case complexity are $\alpha$-acyclic queries, for
which we show that there is no polynomial bound in the certificate
size (assuming the strong form of the exponential time hypothesis).
We then considered how to extend our results using treewidth. We
showed that our same \ms algorithm obtains $\tilde O( |\cert|^{w+1} +
Z)$ runtime for queries with treewidth $w$. For the triangle query
(with treewidth $2$), we presented a modified algorithm that runs in
time $\tilde{O}(|\cert|^{3/2}+Z)$.

\vspace{-7pt}
\paragraph*{Future Work}
We are excited by the notion of certificate-based complexity for join
algorithms; we see it as contributing to an emerging push beyond
worst-case analysis in theoretical computer science. We hope there is
future work in several directions for joins and certificate-based
complexity.

\vspace{-7pt}
\paragraph*{Indexing and Certificates}
The interplay between indexing and certificates may provide fertile
ground for further research. For example, the certificate size depends
on the order of attributes. In particular, a certificate in one order
may be smaller than in another order. We do not yet have a handle on
how the certificate-size changes for the same data in different
orders. Ideally, one would know the smallest certificate size for any
query and process in that order. Moreover, we do not know how to use
of multiple access paths (eg. Btrees with different search keys) in
either the analysis or the algorithm. These indexes may result in
dramatically faster algorithms and new types of query optimization.

\vspace{-7pt}
\paragraph*{Fractional Covers}
A second direction is that join processing has seen a slew of
  powerful techniques based on increasingly sophisticated notions of
  covers and decompositions for queries. We expect that such covers
  (hypergraph, fractional hypergraph, etc.) could be used to tighten
  and improve our bounds. For the triangle query, we have the
  fractional cover bound, i.e., $\tilde{O}( |\cert|^{3/2})$. But is this
  possible for all queries?

%% We plan to investigate applications of \ms in graphical
%% models, coding theory, motif listing, and graph databases.

\section*{Acknowledgments} 
We thank LogicBlox, Mahmoud Abo Khamis, Semih Salihoglu and Dan Suciu for many
helpful conversations.

HQN's work is partly supported by NSF grant CCF-1319402 and a gift from Logicblox. DTN's work is partly supported by NSF grant CCF-0844796 and a gift from Logicblox.
CR's work on this project is generously supported by NSF CAREER Award under     
No. IIS-1353606, NSF award under No. CCF-1356918, the ONR                       
under awards No.  N000141210041 and No. N000141310129,                          
Sloan Research Fellowship, Oracle, and Google.                                  
AR's work is partly supported by                                                
NSF CAREER Award CCF-0844796, NSF grant CCF-1319402 and a                   
gift from Logicblox.

%\balance
%%%%%%%%%%%%%%%%%%%%%%%%%%%%%%%%%%%%%%%%%%%%%%%%%%%%%%%%%%%%%%%%%%%%%%%%%%
{%\small
\bibliographystyle{siam}
%\bibliography{minesweeper}

}
%%%%%%%%%%%%%%%%%%%%%%%%%%%%%%%%%%%%%%%%%%%%%%%%%%%%%%%%%%%%%%%%%%%%%%%%%%

%%%%%%%%%%%%%%%%%%%%%%%%%%%%%%%%%%%%%%%%%%%%%%%%%%%%%%%%%%%%%%%%%%%%%%%%%%
\pagebreak
%\onecolumn

\appendix
% ------------------------------------------------------------------------
\section{The GAO and query's structure}
\label{app:sec:gao}
% ------------------------------------------------------------------------

The input query $Q$ can be represented by a hypergraph $\calH=(\calV, \calE)$,
where $\calV$ is the set of all attributes, and $\calE$ is the collection of
input relations' attribute sets. Any global attribute order (GAO) 
is just a permutation of vertices of $\calV$. 
In the logic, constraint satisfaction, and databases \cite{Schafhauser2006}, 
graphical models, 
and sparse matrix computation \cite{Heggernes:2008:FCM:1654201.1654211} 
literature, any permutation of vertices of a hypergraph is called an {\em
elimination order}, which can be used to characterize many important
properties of the hypergraph.  

Our algorithm is no different: its performance intimately relates to properties
of the GAO, which characterizes structural properties of the hypergraph $\calH$.
In this section we state some relevant known results and derive two 
slightly new results regarding the relationship between elimination
orders and notions of widths and acyclicity of hypergraphs.

% ------------------------------------------------------------------------
\subsection{Basic concepts}
\label{app:subsec:gao-basics}

There are many definitions of acyclic hypergraphs. A hypergraph $(\cal
V, \cal E)$ is $\alpha$-acyclic if the GYO procedure returns
empty~\cite[p.~128]{DBLP:books/aw/AbiteboulHV95}. Essentially, in GYO
one iterates two steps: (1) remove any edge that is empty or contained
in another hyperedge, or (2) remove vertices that appear in at most
one hyperedge. If the result is empty, then the hypergraph is
$\alpha$-acyclic. A query is $\beta$-acyclic if the graph formed by
any subset of hyperedges is $\alpha$-acyclic. Thus, the requirement
that a hypergraph be $\beta$-acyclic is (strictly) stronger than
$\alpha$-acyclic.  We illustrate this with an example.

\begin{example}
We map freely between hypergraphs and queries. The query $Q_{\Delta} =
R(A,B) \Join S(A,C) \Join T(B,C)$ is both $\alpha$-cyclic and
$\beta$-cyclic. However, if one adds the relation $U(A,B,C)$ to form
$Q_{\Delta + U} = R(A,B) \Join S(A,C) \Join T(B,C) \Join U(A,B,C)$
this query is $\alpha$-acyclic, but it is still $\beta$-cyclic.
\end{example}

We can also define these concepts via a notion of tree decomposition.

\bdefn[Tree decomposition]
Let $\calH = (\calV, \calE)$ be a hypergraph.
A {\em tree-decomposition} of $\calH$ is a pair $(T, \chi)$
where $T = (V(T), E(T))$ is a tree and $\chi : V(T) \to 2^{\calV}$ assigns to
each node of the tree $T$ a set of vertices of $\calH$.
The sets $\chi(t)$, $t\in V(T)$, are called the {\em bags} of the 
tree-decomposition.  There are two properties the bags must satisfy
\bi
 \item[(a)] For every hyperedge $F \in \calE$, there is a bag $\chi(T)$ 
 such that $F\subseteq \chi(t)$.
 \item[(b)] For every vertex $v \in \calV$, the set 
 $\{ t \suchthat t \in T, v \in \chi(t) \}$ is not empty and forms a 
 connected subtree of $T$.
\ei
\edefn

There are at least five notions of acyclic hypergraphs, four of which were
introduced very early on 
(see e.g, \cite{DBLP:journals/jacm/Fagin83}), and at least one 
new one introduced recently \cite{Duris2012}. The five notions are {\em not}
equivalent, but they form a strict hierarchy as discussed in
Section~\ref{sec:related:work}.
Of interest to us in this paper are $\beta$-acyclicity and acyclicity.

\bdefn[Acyclicity]\label{defn:acyclic}
A hypergraph $\calH = (\calV, \calE)$ is $\alpha$-{\em acyclic} or just {\em acyclic} if and only if
there exists a tree decomposition 
$(T= (V(T), E(T)), \{\chi(t) \suchthat t \in V(T)\} )$
in which every bag $\chi(t)$ is a hyperedge of $\calH$.
When $\calH$ represents a query $Q$, the tree $T$ is also called the {\em join
tree} of the query. A query is acyclic if and only if its hypergraph is acyclic.
\edefn

\bdefn[$\beta$-acyclicity]\label{defn:beta-acyclic}
A hypergraph $\calH = (\calV, \calE)$ is {\em $\beta$-acyclic} 
if and only if there is {\bf no} sequence $$(F_1, u_1, F_2, u_2, \cdots, F_m, u_m,
F_{m+1}=F_1)$$ with the following properties
\bi
 \item $m \geq 3$
 \item $u_1,\dots,u_m$ are distinct vertices of $\calH$
 \item $F_1, \dots, F_m$ are distinct hyperedges of $\calH$
 \item for every $i \in [m]$, $u_i \in F_i \cap F_{i+1}$, and $u_i \notin F_j$
 for every $j \in [m+1] -\{i,i+1\}$.
\ei
A query is $\beta$-acyclic if and only if its hypergraph is $\beta$-acyclic.
\edefn

The rest of this section roughly follows the definitions given in 
\cite{DBLP:journals/talg/Marx10}.
For a more detailed discussion of (generalized) hypertree decomposition, the
reader is referred to \cite{DBLP:conf/wg/GottlobGMSS05}.

The {\em width} of a tree-decomposition is the quantity 
\[ \max_{t \in V(T)} |\chi(t)|-1. \]
The {\em treewidth} of a hypergraph $\calH$, denoted by $\tw(\calH)$,
is the minimum width over all tree decompositions of the hypergraph.

% ------------------------------------------------------------------------
\subsection{Elimination orders, prefix posets, acyclicity, and hypergraph
widths}
\label{app:subsec:prefix-poset}

An {\em elimination order} of a hypergraph $\calH=(\calV,\calE)$ is simply
a total order $v_1,\dots,v_n$ of all vertices in $\calV$.
Fix an elimination order $\rho = v_1,\dots,v_n$ of $\calH$, for
$j=n,n-1,\dots,1$ we recursively define $n$ hypergraphs 
$\calH_n,\calH_{n-1},\dots,\calH_1$, and $n$ set collections 
$\calP_n,\calP_{n-1},\dots,\calP_1$, as follows.

\bi
 \item[(a)] Let $\calH_n = \calH = (\calV_n, \calE_n=\calE)$ and define
 \begin{eqnarray*}
  \partial(v_n) &=& \{ F \in \calE_n \suchthat v_n \in F \},\\
  \calP_n &=& \{ F - \{v_n\} \suchthat F \in \partial(v_n)\},\\
  U(\calP_n) &=& \bigcup_{F \in \calP_n} F.
 \end{eqnarray*}
 In other words, $\partial(v_n)$ is the collection of hyperedges of $\calH_n$ each of
 which contains $v_n$. (The notation $\partial(v)$ is relatively standard in graph
 theory, denoting the set of edges incident to the vertex $v$.)
 Next, $\calP_n$ is the same set of hyperedges in $\partial(v_n)$ with $v_n$
 removed. Note that the empty set might be a member of $\calP_n$.
 Finally, $U(\calP_n)$ is the ``universe'' of sets in $\calP_n$.
 \item[(b)] For each $j=n-1, n-2, \dots, 1$, define $\calH_j = (\calV_j,
 \calE_j)$ as follows.
 \begin{eqnarray*}
 \calV_j &=& \{v_1,\dots,v_j\}\\
     \calE_j &=& \{ F - \{v_{j+1}\} \suchthat F \in \calE_{j+1} \} \cup 
                 \{ U(\calP_{j+1}) \}\\
 \partial(v_j)&=&\left\{ F \in \calE_j \suchthat v_j \in F\right\}\\
 \calP_j &=& \{ F - v_j \suchthat F \in \partial(v_j)\}\\
 U(\calP_j) &=& \bigcup_{F\in\calP_j} F.
 \end{eqnarray*}
 In other words, let $\calH_j=(\calV_j,\calE_j)$ be the hypergraph obtained from 
 $\calH_{j+1}$ by removing $v_{j+1}$ from $\calH_{j+1}$ from all hyperedges, 
 adding a new hyperedge which is the union of all sets in $\calP_{j+1}$.
 Finally, let $\calP_j$ be the collection of all hyperedges of $\partial(v_j)$
 with $v_j$ removed.
\ei
In particular, the hypergraph $\calH_j$ is on vertex set $\{v_1,\dots,v_j\}$,
and the hypergraph $\calH_1$ has only $\{v_1\}$ as a hyperedge.
The universe $U(\calP_k)$ of $\calP_k$ is a subset of $\{v_1,\dots,v_{k-1}\}$,
and in particular $U(\calP_1) = \emptyset$.

\paragraph*{Prefix posets} 
For each $k\in [n]$, the set collection $\calP_k$ is a collection of subsets 
of $\{v_1,\dots,v_{k-1}\}$.
We will view $\calP_k$ as a {\em partially ordered set} (poset) using the {\em
reversed inclusion order}. In particular, for any
$S_1,S_2\in \calP_k$, we write $S_1 \preceq S_2$ if and only if 
$S_2 \subseteq S_1$. 

%% (It will become clear later in
%% Section~\ref{sec:beta-acyclic} why
%% reversed inclusion is used here instead of inclusion.)

These posets $\calP_k$ are called the {\em prefix posets} with respect to the
elimination order $v_1,\dots,v_n$ of $\calH$.
The {\em bottom element} of a poset $\calP$ is an element $F \in \mathcal
P$ such that $F \preceq F'$ for every $F' \in \calP$.
The poset $\calP$ is a {\em chain} if all members of $\calP$ can
be linearly ordered using the $\preceq$ relation. In other words, $\calP$
is a chain if its members form a nested inclusion collection of sets.
%An {\em anti-chain} of a poset is a collection of elements which are pairwise
%incomparable using the $\preceq$ relation.
It turns out that we can characterize $\beta$-acyclicity
and treewidth of $\calH$ using the prefix posets of some elimination order.

\bdefn[Nested elimination order]
For any $\beta$-acyclic hypergraph $\calH$, a vertex ordering 
$v_1,\dots,v_n$ of $\calH$ is called
a {\em nested elimination order} if and only if every prefix poset $\calP_k$ 
is chain. 
\label{defn:neo}
\edefn

\bprop[$\beta$-acyclicity and the GAO]
A hypergraph $\calH = (\calV, \calE)$ is {\em $\beta$-acyclic} 
{\em if and only if} there exists a vertex ordering $v_1,\dots, v_n$ which is
a nested elimination order for $\calH$.
\label{prop:beta-gao}
\eprop
\bp
For the forward direction, suppose $\calH$ is $\beta$-acyclic.
A {\em nest point} of $\calH$ is a vertex $v \in \calH$ such that the 
collection of hyperedges containing $v$ forms a nested sequence of subsets, 
one contained in the next. In 1980, Brouwer and Kolen \cite{brouwer-kolen-1980}
proved that any $\beta$-acyclic hypergraph $\calH$ has at least two nest 
points. Let $v_n$ be a nest point of $\calH$. Then, from the definition of nest
point, $\partial(v_n)$ is a chain each of whose members contains $v_n$.
The set $\calP_n$ is thus also a chain as it is the same as $\partial(v_n)$ with
$v_n$ removed, and $\calP_n$'s bottom element is precisely $U(\calP_n)$.
Consequently, $\calH_{n-1}$ is precisely $\calH - \{v_n\}$. 
The graph $\calH- \{v_n\}$ is $\beta$-acyclic because $\calH$ is
$\beta$-acyclic. By induction there exists an elimination order
$v_1,\dots,v_{n-1}$ such that every prefix poset $\calP_k$, $k\in [n-1]$,
is a chain. Thus, the elimination order $v_1,\dots,v_n$ satisfies the desired
property.

Conversely, suppose there exists an ordering $v_1,\dots,v_n$ of all vertices of
$\calH$ such that every poset $\calP_k$ is a chain. Assume to the contrary that
$\calH$ is not $\beta$-acyclic. Then, there is a sequence
$$(F_1,u_1,F_2,u_2,\dots,F_m,u_m,F_{m+1}=F_1)$$ satisfying the conditions stated
in Definition~\ref{defn:beta-acyclic}.
Without loss of generality, suppose $u_m$ comes last in the elimination 
order $v_1,\dots,v_n$, and that $u_m = v_k$ for some $k$. 
Then, the poset $\calP_k$ contains the set 
$F_m \cap \{v_1,\dots,v_{k-1}\}$ and the set $F_1 \cap
\{v_1,\dots,v_{k-1}\}$. Since both $u_2$ and $u_{m-1}$ come before $u_m$ in the
ordering, we have
\begin{eqnarray*}
u_2 &\in& \left(F_1 \cap \{v_1,\dots,v_{k-1}\}\right) \setminus 
        \left(F_m \cap \{v_1,\dots,v_{k-1}\}\right)\\
u_{m-1} &\in& \left(F_m \cap \{v_1,\dots,v_{k-1}\}\right) \setminus 
        \left(F_1 \cap \{v_1,\dots,v_{k-1}\}\right).
\end{eqnarray*}
Consequently, $\calP_k$ is not a chain.
\ep

We have just characterized $\beta$-acyclicity with a polynomial-time verifiable
property of the GAO. We next characterize the treewidth of a hypergraph using
the best ``elimination width'' of its GAO. This result is well-known in the
probabilistic graphical model literature.

\bprop[Treewidth and the GAO]
Let $\calH = (\calV, \calE)$ be a hypergraph with treewidth $w$.
Then there exists an elimination order $v_1,\dots, v_n$ of all vertices of 
$\calH$ such that for every $k \in [n]$ we have 
$\left| U(\calP_k) \right| \leq w.$
\label{prop:gao-tw}
\eprop
\bp
This follows from the well-known fact that the smallest induced treewidth 
(over all elimination orders) of $\calH$ is the same as the treewidth of
$\calH$ (see, e.g., \cite{journals/ai/DechterP89, MR985145}). 
The maximum size of the universes $U(\calP_k)$, $k\in [n]$, is
precisely the induced treewidth of the Gaifman graph of $\calH$ with respect to
the given elimination order.
\ep

% ------------------------------------------------------------------------
\section{Certificates}
\label{app:sec:certificates}
% ------------------------------------------------------------------------

The notion of certificate is subtle. In this section we give a series of examples
and proofs of propositions exploring its properties. 

% ------------------------------------------------------------------------
\subsection{Illustrations and basic examples}
\label{app:sec:certificate-basics}

To understand the notion of certificates, it is important to understand the
input to $\ms$ and how relations are accessed. Figure~\ref{fig:R-trie} gives
an illustration of index tuples to access a relation $R$. In this example,
the nodes (except for the root) are the variables $R[\mv x]$ whose contents have
been filled out by a database instance.
\begin{figure*}
\centerline{\includegraphics[width=4in]{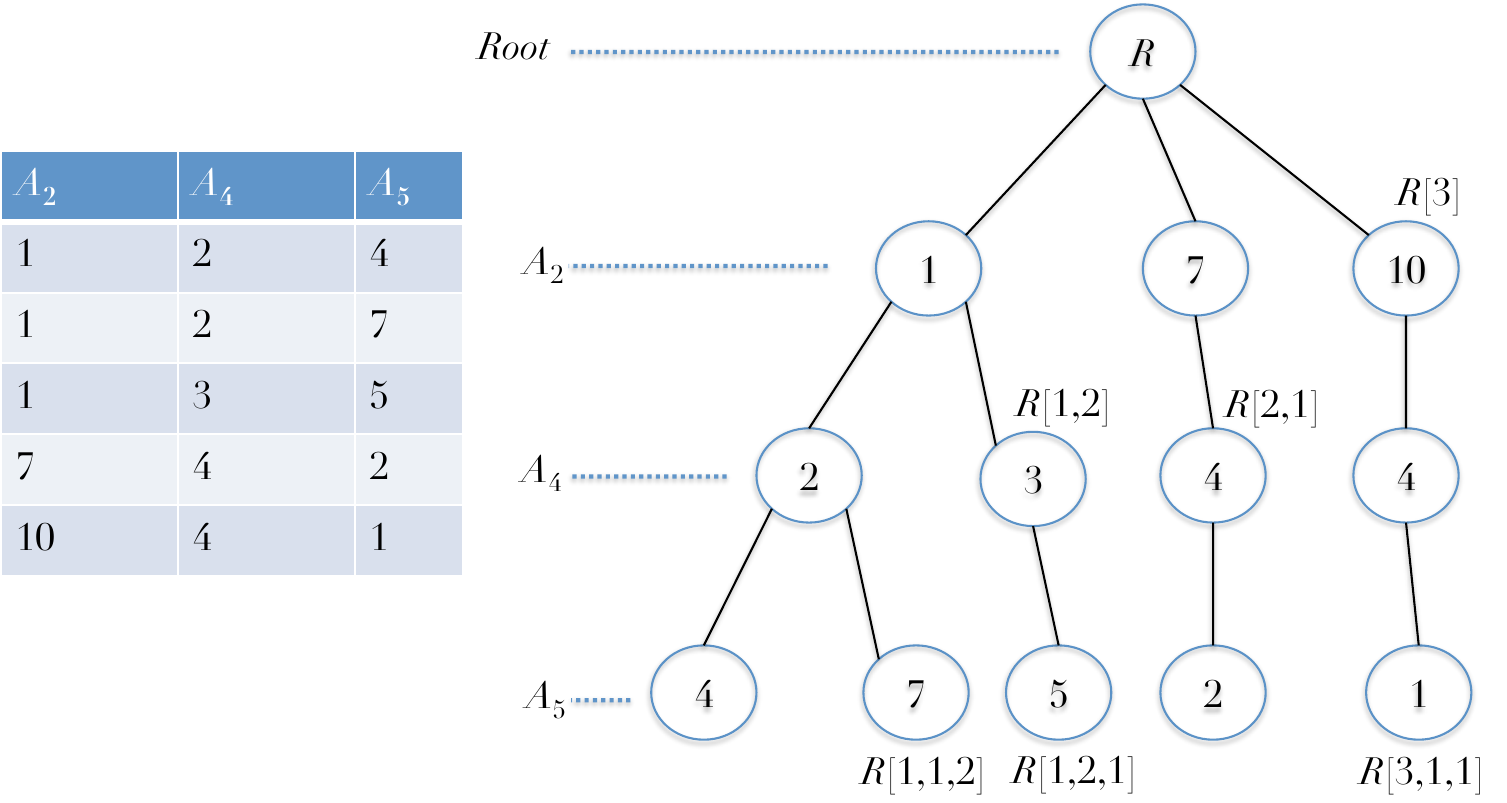}}
\caption{The (unbounded fanout) $\btree$ data structure. 
Here $R[x_1,x_2]$ is the value of the node where we take the $x_1$th branch of
the first level, then the $x_2$th branch at the second level of $R$'s $\btree$.
For this example, $|R[*]| = 3$, $|R[1, *]| = 2$, $|R[2, *]| = 1$.}
\label{fig:R-trie}
\end{figure*}

Next, we start with an extremely simple join query to illustrate the
notion of certificates, showing that certificates can have constant size
and they can be a lot smaller than the output size.

\begin{example}[Constant size certificates]
Consider the query $R(A) \Join S(A, B)$ where 
\begin{eqnarray*}
    R &=& [N]\\
    S &=& \{ (N+1, i + N) \suchthat i \in  [N]\}.
\end{eqnarray*}
In this case, $\{ R[N] < S[1] \}$ is a certificate showing that the output is
empty: for every database $I$ in which $R^I[N] < S^I[1]$ there is no tuple in
the output.
\end{example}

\begin{example}[$|\cert| \ll Z$]
Next, consider the following instance of the same query as above.
\begin{eqnarray*}
    R &=& [N]\\
    S &=& \{(N, 10i), i \in [N]\}.
\end{eqnarray*}    
In this case, $\{ R[N] = S[1] \}$ is a certificate because, for every input
database $I$ for which $R^I[N] = S^I[1]$, the outputs are tuples of the form
$(R^I[N], S^I[1, i]), i \in [N]$. And, the witnesses are pairs of index tuples
$\{ N, (1, i) \}$.
These two examples show that certificates can be of constant size, and they can
be arbitrarily smaller than the output size.
\end{example}

Extrapolating from the above example, it is not hard to show that for any join
query we can construct an instance whose optimal certificate size is only a
function of the query size and not the data. In essence, such certificates are
of constant size in data complexity. Consequently, algorithms whose runtime is a
function of the optimal certificate can be extremely fast!.

% ------------------------------------------------------------------------
\subsection{Certificate subtleties}
\label{app:sec:certificate-subtleties}

The comparisons of the forms shown in \eqref{eqn:comparison-forms} allow for 
comparisons between tuples of the same relation. Comparisons between tuples from
the same relation and the equalities can help 
tremendously in reducing the size of the overall certificate. This fact will
make the job of the algorithm designer more difficult if we aim for a runtime
proportional to the optimal certificate size. Consider the
following example.

\begin{example}[Equalities and same-relation comparisons are important]
\label{ex:same-relation-and-equalities}
Consider the following query, where the global attribute order is $A, B, C$
\[ Q = R(A, C) \Join S(B, C), \]
where
\begin{eqnarray*}
R(A, C) &=& [N] \times \{2k \suchthat k \in [N]\}\\
S(B, C) &=& [N] \times \{2k-1 \suchthat k \in [N]\}
\end{eqnarray*}
The join is empty, and there is a certificate of size $O(N^2)$ showing
that the output is empty.  Note that both the relations have size $N^2$. 
The certificate consists of the following comparisons:
\begin{eqnarray*}
R[1,c] &=& R[a,c], \text{ for } a, c \in [N], a > 1\\
S[1,c] &=& S[b,c], \text{ for } b, c \in [N], b > 1,
\end{eqnarray*}
\[ S[1,1] < R[1,1] < S[1,2] < R[1,2] < \cdots < S[1, N] < R[1,N]. \]

If we don't use any equality, or if we only compare tuples from different
relations, any certificate will have to be of size
$\Omega(N^3)$ because it will have to show for each pair
$a, b$ that $R[a,*] \cap S[b,*] = \emptyset$ which takes $2N-1$ inequalities,
for a grand total of $N^2(2N-1) = \Omega(N^3)$ comparisons.
\end{example}

A certificate is a function of the GAO (and of course, the data).
For the same input data, changing the GAO can 
dramatically change the optimal certificate size, and for non-trivial queries
we cannot predict the dramatic difference between optimal certificate sizes
of different GAOs without examining the data values.

\begin{example}[Certificate's dependency on the GAO]
\label{ex:cert-different-GAO}
Consider the same query as in Example~\ref{ex:same-relation-and-equalities} 
but with the global
attribute order of $C, A, B$. In this case, 
\begin{eqnarray*}
R(C, A) &=& \{2k \suchthat k \in [N]\} \times [N]\\
S(C, B) &=& \{2k-1 \suchthat k \in [N]\} \times [N]
\end{eqnarray*}
The following is an $O(N)$-sized certificate proving that the output is empty:
\[ S[1] < R[1] < S[2] < R[2] < \cdots < S[N] < R[N]. \]
This GAO is a nested elimination order for this query, and thus \ms 
runs in time $\tilde O(N)$ on this instance, thanks to
Theorem~\ref{thm:io-beta-neo}.
\end{example}

Examples \ref{ex:same-relation-and-equalities} and \ref{ex:cert-different-GAO} 
indicate a trend that we can prove rigorously. 

\bprop
Let $\rho$ be any GAO. Let $B$ be any private attribute of some relation $R$,
i.e. $B$ does not belong to any other relation.  Let $\rho'$ be an attribute
order obtained from $\rho$ by removing $B$ from $\rho$ and adding it to the end of
$\rho$.  Let $\cert(\rho)$ denote an optimal certificate with respect to the
GAO $\rho$. Similarly, define $\cert(\rho')$. Then,
$|\cert(\rho')| \leq |\cert(\rho)|$.
\eprop
\bp
First, we observe that in an optimal certificate $\cert$ (for any GAO), there is
no comparison involving $B$-variables. If $\cert$ does contain such comparison, 
let $\calA$ be the argument obtained from $\cert$ by removing all comparisons 
involving $B$-variables. We want to show that $\calA$ remains a certificate,
still, contradicting the optimality of $\cert$. 
Let $K$ be any database instance satisfying $\cert$. (If there is no such $K$,
then there is no database instance satisfying $\calA$, and hence $\calA$ is
 vacuously a  certificate!)
Let $I$ and $J$ be two database instances satisfying the argument $\calA$. 
Let $I'$ and $J'$ be obtained from $I$ and $J$ by filling in the $B$-variables
using values from the corresponding $B$-variables from $K$. Then, $I'$ and $J'$
satisfy $\cert$. Consequently, every witness for $Q(I')$ is a witness for
$Q(J')$ and vice versa. But every witness for $Q(I')$ is {\em also} a witness
for $Q(I)$, and every witness for $Q(J')$ is also a witness for $Q(J)$, and vice
versa, because $B$ is a private attribute! Hence, $\calA$ is a certificate as
desired.

Second, we can now assume that $\cert(\rho)$ has no comparison between
$B$-variables. Note that, a variable on a relation $R$ is simply a node on its
search tree. When we change $\rho$ to $\rho'$, some nodes on a variable coming
after $B$ in a relation might collapse into one node because their values
are equal. Call the new node an {\em image} of the old node.
Let $\calA$ be an argument for the $\rho'$ GAO obtained from
$\cert(\rho)$ by replacing every comparison in $\cert(\rho)$ with the comparison
between their images in $\rho'$. Every database satisfying $\calA$ also
satisfies $\cert$, from which we can infer the set of witnesses. Thus $\calA$ is
a certificate, which can be smaller than $\cert(\rho)$ because of the collapsing
of nodes.
\ep

From the above proposition, we know that better certificates can be obtained 
by having  GAOs in which all private attributes come at the end of the order.
Unfortunately, that is as far as the GAO can tell us about
the optimal certificate size. If there were more than one non-private attribute,
then the optimal certificate size is highly data dependent.
The following example illustrates this point further.

\begin{example}[Certificate's dependency on the GAO even without private
attributes]
Consider the join query $$Q = R(A, B) \Join S(A, B).$$
Suppose
\begin{eqnarray*}
R &=& \{ (i, i) \suchthat i \in [N] \}\\
S &=& \{ (N+i, i) \suchthat i \in [N] \}.
\end{eqnarray*}
Then, the optimal certificate for the $(A, B)$ order has size $O(1)$: 
\[ R[N] < S[1], \]
while the optimal certificate for the $(B, A)$ order has size
$\Omega(N)$: 
\[ R[i, N] < S[i, 1], \text{ for all } i \in [N]. \]
And, we can't tell which is which by just looking at the
shape of the search trees for $R$ and $S$.
\end{example}

The following example illustrates that the runtime of 
$O(|\cert(\rho)|+Z)$ 
in the GAO $\rho$ may not be better than the runtime of, say,
$O(|\cert(\rho')|^{w+1}+Z)$ 
in another GAO $\rho'$ for {\em the same} data. 
The notion of {\it nested elimination order} was defined
earlier in Section~\ref{app:sec:gao}.

\begin{example}[Nested elimination order may have large certificate]
It is easy to construct a query and the data so that a nested elimination
order has a much larger optimal certificate than a non-nested elimination
order. Consider the following query
\[ Q = R(A, B, C) \Join S(A, C) \Join T(B, C). \]
This query is $\beta$-acyclic, and $\rho = (C, A, B)$ is a nested elimination order
while $\rho' = (A, B, C)$ is not. \ms runs in time 
$\tilde O(|\cert(C,A,B)|+Z)$ for the former order, and in time
$\tilde O(|\cert(A, B, C)|^3 + Z)$ for the latter.
However, it is entirely possible that $|\cert(A, B, C)|^3 \ll |\cert(C,A,B)|$.
For example, consider
\begin{eqnarray*}
R(A,B,C) &=& \{ (i, i, i) \suchthat i \in [N] \}\\
S(A,C) &=& \{ (N+i, i) \suchthat i \in [N] \}\\
T(B,C) &=& \{ (i, i) \suchthat i \in [N] \}.
\end{eqnarray*}
In this case, similar to the previous example
$|\cert(A, B, C)| = 1$ (where it says $R[N] < S[1]$), while
$|\cert(C, A, B)| = \Omega(N)$.
\end{example}

% ------------------------------------------------------------------------
\subsection{Proof of Proposition~\ref{prop:Omega(|C|)}}
\label{app:sec:prop:Omega(|C|)}
\bp
To prove this proposition, it is sufficient to show that the set of comparisons
issued by an execution of a comparison-based algorithm is a certificate.
To be concrete, we model a comparison-based join algorithm by a decision tree.
Every branch in the tree corresponds to a comparison of the form
\eqref{eqn:comparison-forms}. An execution of the join algorithm is a path
through this decision tree, reaching a leaf node. At the leaf node, the result
$Q(I)$ is labeled. The label at a leaf is the set of tuples the algorithm deems 
the output of the query applied to database instance $I$. 
The collection of comparisons down the path is
an argument $\calA$ which we want to prove a certificate.

First, note that for every tuple $\mv t = (t_1,\dots,t_n) \in Q(I)$, the values
$t_i$ have to be one of the values $R^I[\mv x]$ for some $R \in \atoms(Q)$.
If this is not the case, then we can perturb the instance $I$ as follows: for
every attribute $A_i$ let $M_i$ be the maximum value occurring in any $A_i$-value
overall tuples in the input relations. Now, add $M_i+1$ to every $A_i$-value.
Then, all $A_i$-values are shifted the same positive amount. In this new
database instance $J$, all of the comparisons in the argument have the same
Boolean value, and hence the output has to be the same. Hence, if there was a
value $t_i$ in some output tuple not equal to $R[\mv x]$, the output would be
wrong.

Second, we show that every output tuple can be uniquely identified with a
witness, independent of the input instance $I$.
Recall that a collection $X$ of (full) index tuples is said to
be a {\em witness} for $Q(I)$ if $X$ has exactly one full
index tuple from each relation $R\in \atoms(Q)$, and all index tuples
in $X$ contribute to the same $\mv t \in Q(I)$.

Fix an input instance $I$ and an output tuple $\mv t = (t_1,\dots,t_n)$. Note as
indicated above that the $t_i$ can now be thought of as a variable $R[\mv x]$
for some index tuple $\mv x$ (not necessarily full) and some relation $R\in
\atoms(Q)$. By definition of the natural join operator, there has to be a
witness $X$ for this output tuple $\mv t$.

Consider, for example, a full index tuple $\mv y = (y_1, \dots, y_k)$
from some relation $S$ which is a member of the witness $X$. 
Suppose the relation $S$ is on attributes $(A_{s(1)}, A_{s(2)}, \dots,
A_{s(k)})$.
We show that, for every $j\in [k]$, the argument $\calA$ must imply via the 
transitivity of the
equalities in the argument that $S[y_1, \dots, y_j] = t_{s(j)}$.

Suppose to the contrary that this is not the case. Let $V$ be the set of all 
variables transitively connected to the variable $S[y_1, \dots, y_j]$ by 
the equality comparisons in $\calA$.

Now, construct an instance $J$ from instance $I$ by doing the following
\bi
 \item set $R^J[\mv x] = 2 R^I[\mv x] + 1$ for all variables $R[\mv x]$ appearing
     in the argument $\calA$ but $R[\mv x]$ is not in $V$.
 \item set $R^J[\mv x] = 2 R^I[\mv x] + 2$ for all variables $R[\mv x]$ appearing in
     $V$.
\ei
Then, any comparison between a pair of variables both not in $V$ or 
both in $V$ have the same outcome in both databases $I$ and $J$. For a
pair of variables $R[\mv x] \in V$ and $T[\mv y] \notin V$
the comparison cannot be an equality from the definition of $V$, and hence 
the $<$ or $>$ relationship still holds true. 
This is because if $a$ and $b$ are natural numbers, then 
$a<b$ implies $2a+2<2b+1$ and $2a+1<2b+2$. 
Consequently, the instance $J$ also satisfies all comparisons in the argument
$\calA$.
However, at this point $S[\mv y]$ can no longer be contributing to $\mv t$.
More importantly, {\bf no} full index tuple from $S$ can contribute to $\mv t$
in $Q(J)$. Because, 
\begin{eqnarray*}
  S^J[y_1, \dots, y_{j-1}, y_j-1] 
    &\leq& 2S^I[y_1, \dots, y_{j-1}, y_j-1] +1\\
 &\leq& 2(S^I[y_1, \dots, y_{j-1}, y_j]-1) +1\\
 &=& 2S^I[y_1, \dots, y_{j-1}, y_j] - 1\\
 &=& 2t^I_{s(j)} -1\\
 &<& t^J_{s(j)}.
\end{eqnarray*}
(The first inequality is an equality except when $y_j=1$.)
Similarly,
\begin{eqnarray*}
    S^J[y_1, \dots, y_{j-1}, y_j+1] 
    &\geq& 2S^I[y_1, \dots, y_{j-1}, y_j+1] +1\\
 &\geq& 2(S^I[y_1, \dots, y_{j-1}, y_j]+1) +1\\
 &=& 2S^I[y_1, \dots, y_{j-1}, y_j] + 3\\
 &=& 2t^I_{s(j)} +3\\
 &>& t^J_{s(j)}.
\end{eqnarray*}
(Except when $y_j =
|S[y_1,\dots,y_{j-1},*]|$, the first inequality is an equality.)
\ep

% ------------------------------------------------------------------------
\subsection{Proof of Proposition~\ref{prop:optimal-certificate-upperbound}}
\label{app:sec:prop:optimal-certificate-upperbound}
%\bprop[Upper bound on optimal certificate size]
%Let $Q$ be a general join query on $m$ relations and $n$ attributes. Let $N$ be
%the total number of tuples from all input relations.
%Then, there exists a certificate for $Q$ of size $r\cdot N$, where $r =
%\max\{\arity(R) \suchthat R \in \atoms(Q)\} \leq n$.
%\label{app:prop:optimal-certificate-upperbound}
%\eprop
\bp%[Proof]
We construct a certificate $\cert$ as follows.
For each attribute $A_i$, let $v_1< v_2<\cdots< v_p$ denote the set of {\em all}
possible $A_i$-values present in any relations from $\atoms(Q)$ which has
$A_i$ as an attribute. More concretely,
\[ \{v_1,v_2,\dots,v_p\} := \bigcup_{R \in \atoms(Q), A_i \in \bar A(R)}
\pi_{A_i}(R).
\]
For each $k \in [p]$, let $T_k$ denote the set of all 
tuples from relations containing $A_i$ such that the tuple's $A_i$-value 
is $v_k$. Note that the tuples in $T_k$ can come from
the same or different relations in $\atoms(Q)$. Next, add to $\cert$
at most $|T_k|-1$ equalities connecting all tuples in $T_k$ 
asserting that their $A_i$-values are equal. 
(The reason we may not need exactly $|T_k|-1$ equalities is because there might
be many tuples from the same relation $R$ that share the $A_i$-value, and $A_i$
comes earlier than other attributes of $R$ in the total attribute order.)

Then, for each $k \in [p]$, pick an arbitrary tuple $\mv t_k \in T_k$ 
and add $p-1$ inequalities stating that $\mv t_1.A_i < \mv
t_2.A_i < \cdots < \mv t_p.A_i$. 
(Depending on which relation $\mv t_k$ comes from, the actual syntax for
$\mv t_k.A_i$ is used correspondingly. For example, if $\mv t_k$ is from the
relation $R[A_j, A_i, A_\ell]$, then $\mv t_k.A_i$ is actually 
$R[x_j, x_i]$.)

Overall, for each $A_i$ the total number of comparisons we added is
at most the number of tuples that has $A_i$ as an attribute.
Hence, there are at most $rN$ comparisons added to the certificate $\cert$, 
and they represent all the possible relationships we know about 
the data. The set of comparisons is thus a certificate for this instance.
\ep

% ------------------------------------------------------------------------
\section{Running Time Analysis}
\label{app:sec:rt:analysis}
% ------------------------------------------------------------------------

In this paper, we use the following notion to benchmark the runtime of join algorithms.

\bdefn
We say a join algorithm $\algo$ for a join query $Q$ to be instance optimal for $Q$ with optimality ratio $\alpha$ if the following holds. For every instance for $Q$, the runtime of the algorithm is bounded by $O_{|Q|}(\alpha\cdot |\cert|)$, where $O_{|Q|}(\cdot)$ ignores the dependence on the query size and $\cert$ be the certificate of the smallest size for the given input instance. We allow $\alpha$ to depend on the input size $N$. Finally, we refer to an instance optimal algorithm for $Q$ with optimality ratio $O(\log{N})$ simply as near instance optimal\footnote{Technically we should be calling such algorithms as near instance optimal for certificate-based complexity but for the sake of brevity we drop the qualification. Further, we use the term near instance optimal to mirror the usage of the term near linear to denote runtimes of $O(N\log N)$.} for $Q$.
\edefn

Next, we briefly justify our definition above. First note that we are using 
the size of the optimal certificate as a benchmark to quantify the 
performance of join algorithms. We have already justified this as a natural 
benchmark to measure the performance of join algorithms in 
Section~\ref{sec:certificate}. In particular, recall that 
Proposition~\ref{prop:Omega(|C|)} says that $|\cert|$ is a valid lower bound 
on the number of comparisons made by any comparison-based algorithm that 
``computes" the join $Q$. Even though this choice makes us compare 
performance of algorithms in two different models (the RAM model for the 
runtime and the comparison model for certificates), this is a natural choice 
that has been made many times in the algorithms literature: most notably, 
the claim that algorithms to sort $n$ numbers that run in 
$O(n\log{n})$ time are optimal in the comparison model. (This has also been 
done recently in other works, e.g., in~\cite{self-improving,geometric-io}.)

Second, the choice to ignore the dependence on the query size is standard 
in database literature. In particular, in this work we focus on the data 
complexity of our join algorithms.

Perhaps the more non-standard choice is to call an algorithm with
optimality ratio $O(\log{N})$ to be (near) instance optimal. We made this
choice because this is {\em unavoidable} for comparison-based
algorithm. In particular, there exists a query $Q$ so that every
(deterministic) comparison-based join algorithm for $Q$ needs to make
$\Omega(\log{N}\cdot|\cert|)$ many comparisons on {\em some} input
instance. This follows from the easy-to-verify fact for the selection
problem (given $N$ numbers $a_1,\dots,a_N$ in sorted order, check whether a
given value $v$ is one of them), every comparison-based algorithm
needs to make $\Omega(\log{N})$ many comparisons while every instance
can be ``certified" with constant many comparisons~\cite[Problem
  1(a)]{tim-ps}. For the sake of completeness we sketch the argument
below.

Consider the query $Q=R(A)\Join S(A)$. Now consider the instance where $R(A)=\{a_1,\dots,a_N\}$ and $S(A)=\{v\}$. Note that for this instance, we have $|\cert|\le O(1)$ (and that the output of $Q$ is empty if and only if $v$ does not belong to $\{a_1,\dots,a_N\}$). However, given any sequence of $\lfloor \log{N}\rfloor -1$ comparisons between (the only) element of $S$ and some element of $R$, there always exists two instantiation of $a_1,\dots,a_N$ and $v$ such that in one case the output of $Q$ is empty and is non-empty in the other case. (Basically, the adversary will always answer the comparison query in a manner that forces $v$ to be in the larger half of the ``unexplored" numbers.) 

Finally, we remark that even though this $\Omega(\log{N})$ lower bound on the optimality ratio is stated for the specific join query $Q$ above, it can be easily extended to any join query $Q'$ where at least two relations share an attribute (by ``embedding" the above simple set intersection query $Q$ into $Q'$).

% ------------------------------------------------------------------------
\section{The outer algorithm}
\label{app:sec:outer-analysis}
% ------------------------------------------------------------------------

\subsection{Worked Example of \ms}
\label{app:sec:ms-worked-example}

\begin{example}[$\ms$ in action]
Let $Q_2$ join the following relations:
\begin{eqnarray*}
    R(A_1) &=& [N], \\
    S(A_1,A_2) &=& [N] \times [N], \\
    T(A_2,A_3) &=& \set{(2,2),(2,4)}, \\
    U(A_3) &=& \set{1,3},
\end{eqnarray*}
where $(A_1, A_2, A_3)$ is the global attribute order.

In this example, the value domain of every attribute is $[N]$. The algorithm to compute $Q_2$ will run as follows:
\begin{itemize}
\item First the constraint set $\cds$ is empty.
\item WLOG, assume $\mv t = (-1,-1,-1)$ is 
    the first tuple returned by $\cds.\getpp()$.
\item Step 1, the following constraints will be added to $\cds$:
\begin{eqnarray*}
\langle (-\infty, 1), *, * \rangle &:& \text{from $R$ and $S$}\\
\langle 1, (-\infty, 1), * \rangle &:& \text{from $S$}\\
\langle *, (-\infty, 2), * \rangle &:& \text{from $T$}\\
\langle *, =2, (-\infty, 2) \rangle &:& \text{from $T$}\\
\langle *, *, (-\infty, 1) \rangle &:& \text{from $U$}
\end{eqnarray*}
Then $\cds$ returns, say, $\mv t = (1, 2, 2)$ which does not satisfy any of
the above constraints.
\item Step 2, the following constraint will be added to $\cds$:
\begin{eqnarray*}
\langle *, *, (1, 3) \rangle &:& \text{from $U$}
\end{eqnarray*}
Then $\cds$ returns, say, $\mv t = (1, 2, 3)$ which does not satisfy any of
the above constraints.

\item Step 3, the following constraint will be added to $\cds$:
\begin{eqnarray*}
\langle *, =2, (2, 4) \rangle &:& \text{from $T$}
\end{eqnarray*}
Then $\cds$ returns, say, $\mv t = (1, 2, 4)$ which does not satisfy any of
the above constraints.

\item Step 4, the following constraint will be added to $\cds$:
\begin{eqnarray*}
\langle *, *, (3,+\infty) \rangle &:& \text{from $U$}
\end{eqnarray*}
Then $\cds$ returns, say, $\mv t = (1, 3, 1)$ which does not satisfy any of
the above constraints.

\item Step 5, the following constraint will be added to $\cds$:
\begin{eqnarray*}
\langle *, (3,+\infty), * \rangle &:& \text{from $T$}\\
\langle *, =2, (4,+\infty) \rangle &:& \text{from $T$}
\end{eqnarray*}
At this point no $\mv t \in\outspace$ is free from the constraints and the
algorithm stops, reporting that the output is empty.
\end{itemize}
\label{ex:dung}
\end{example}

% ------------------------------------------------------------------------
\subsection{Proof of Theorem~\ref{thm:analyze-outer-algorithm}}
\label{app:subsec:thm:analyze-outer-algorithm}

\bp 
We account for the maximum number of iterations to be $O(2^r|\cert|+Z)$ as follows.
We give each comparison in the optimal certificate $O(2^r)$ credits and each output 
tuple $O(1)$ credits.
Every iteration is represented by a distinct probe point (or active tuple) $\mv
t$. Hence, instead of counting the number of iterations we count the number of
probe points $\mv t$ returned by the \cds.

Consider a probe point 
$$\mv t = (t_1,t_2,\dots,t_n)$$ returned by the \cds in some iteration of
Algorithm~\ref{algo:outer-algorithm}. 
If $\mv t$ is an output tuple, then we use a credit from the output tuple
to pay for this iteration.
Hence, the hard part is to account for the probe points $\mv t$ that are
not part of the query's output. In these cases we will use the credits
from the comparisons of $\cert$.

\noindent
{\bf Case 1.} First, let us assume that no input relation has a private
attribute\footnote{An attribute is private if it only appears in one relation.}. (Intuitively, private attributes should not be a factor in any join
decision, so this is the harder case.)

Consider a relation $R \in \atoms(Q)$ with $\arity(R) = k$.
Let the attributes of $R$, in accordance with the GAO, be 
$$\bar A(R) = (A_{s(1)}, \dots, A_{s(k)}).$$
(Strictly speaking, the function $s: [k] \to [n]$ depends on $R$, but we will
implicitly assume this dependency to simplify notation.)
Let $p$ be an integer such that $p \in \{0,1,\dots,k-1\}$. 
For any vector 
$$\mv v = (v_1,\dots,v_p) \in \{\ell, h\}^p,$$ the variable
$$R\left[i^{(v_1)}_R, 
        i^{(v_1,v_2)}_R, \dots,
        i^{(v_1,\dots,v_p)}_R, 
        i^{(v_1,\dots,v_p, h)}_R
  \right]$$ 
is said to be {\em $\mv t$-alignable} if all variables
\[ R\left[i^{(v_1)}_R\right], \cdots,
R\left[i^{(v_1)}_R, 
        i^{(v_1,v_2)}_R, \dots,
        i^{(v_1,\dots,v_p)}_R
\right]\] are already $\mv t$-alignable {\bf and} if
\[ R\left[i^{(v_1)}_R, 
        i^{(v_1,v_2)}_R, \dots,
        i^{(v_1,\dots,v_p)}_R, 
        i^{(v_1,\dots,v_p, h)}_R
  \right]\] is either equal to $t_{s(p+1)}$
or it is not involved in any comparison in the certificate $\cert$.
Similarly, we define $\mv t$-alignability for the variable
$$R\left[i^{(v_1)}_R, 
        i^{(v_1,v_2)}_R, \dots,
        i^{(v_1,\dots,v_p)}_R, 
        i^{(v_1,\dots,v_p, \ell)}_R
\right].$$
The semantic of $\mv t$-alignability is as follows.
For any $p \in [k]$, if a $\mv t$-alignable variable 
$$e = R\left[i^{(v_1)}_R, 
        i^{(v_1,v_2)}_R, \dots,
        i^{(v_1,\dots,v_p)}_R
\right]$$
is not already equal to $t_{s(p)}$, setting $e=t_{s(p)}$ will transform the
input instance into another database instance satisfying all
comparisons in $\cert$ without violating the relative order in the relation
that $e$ belongs to.
Following this semantic, any element whose index is out of range is {\em not}
$\mv t$-alignable.

{\bf Claim:}
Since $\mv t$ is not an output tuple, we claim that there must be a relation 
$R \in \atoms(Q)$ with arity $k$, some $p\in \{0,\dots,k-1\}$
and a vector $\mv v \in \{\ell, h\}^p$ for which both variables
$$R\left[i^{(v_1)}_R, 
        i^{(v_1,v_2)}_R, \dots,
        i^{(v_1,\dots,v_p)}_R, 
        i^{(v_1,\dots,v_p, \ell)}_R
\right]$$ and
$$R\left[i^{(v_1)}_R, 
        i^{(v_1,v_2)}_R, \dots,
        i^{(v_1,\dots,v_p)}_R, 
        i^{(v_1,\dots,v_p, h)}_R
\right]$$ 
are {\bf not} $\mv t$-alignable.

To see the claim, suppose for every relation $R$ the above claim does not hold.
Then, for every relation $R \in \atoms(Q)$ there is a vector 
$$\mv v^R = (v^R_1, \dots, v^R_k) \in \{\ell, h\}^k$$ 
with $k = \arity(R)$ such that the variable
$$R\left[i^{(v^R_1)}_R, 
        i^{(v^R_1,v^R_2)}_R, \dots,
        i^{(v^R_1,\dots,v^R_k)}_R \right]$$
is $\mv t$-alignable. By definition of $\mv t$-alignability, {\em all} the variables 
\begin{equation}
    R\left[i^{(v^R_1)}_R, 
        i^{(v^R_1,v^R_2)}_R, \dots,
        i^{(v^R_1,\dots,v^R_j)}_R \right]
\label{eqn:alignable-vars}
\end{equation}
are also $\mv t$-alignable for every $j \in [k]$.

Now, to reach a contradiction we construct two database instances $I$ and $J$
satisfying all comparisons in $\cert$ yet there is a witness for $Q(I)$ which is
{\em not} a witness for $J$.
\bi
\item {\em The database instance $I$.} Keep all variables the same except for the
following: for each $j \in [k]$, we set
\[ R\left[i^{(v^R_1)}_R, i^{(v^R_1,v^R_2)}_R, \dots, i^{(v^R_1,\dots,v^R_j)}_R\right]
   = t_{s(j)},
\]
for every $R\in\atoms(Q)$.
Then, clearly the following set of full index tuples
\begin{equation}
    \left\{ \left(i_R^{(v^R_1)}, i_R^{(v^R_1,v^R_2)}, \dots,
    i_R^{(v^R_1,\dots,v^R_k)}\right)
     \suchthat R\in \atoms(Q) \right\}
\label{eqn:general-witness}
\end{equation}
is a witness for $Q(I)$.
\item {\em The database instance $J$.} Note that we are in the case where 
$\mv t$ is not an output tuple. Hence, of the variables specified in
\eqref{eqn:alignable-vars}, 
there must be at least one relation $R$ of arity $k$ 
and one index $j \in [k]$ for which
\[R\left[i^{(v^R_1)}_R, i^{(v^R_1,v^R_2)}_R, \dots,
    i^{(v^R_1,\dots,v^R_j)}_R \right] \neq t_{s(j)}.
\]
(Note again that $s$ is a function of $R$ too, but we dropped the subscript for
clarity.)
Now, we set all of the alignable variables in \eqref{eqn:alignable-vars} to
be equal to corresponding coordinate in $\mv t$, except for the above. Then, the
set defined in \eqref{eqn:general-witness} is no longer a witness for $Q(J)$.
\ei
This is a contradiction and the claim is thus proved.

Now, fix a relation $R$ for which the pair of variables in the claim exists.
Let $p$ be the smallest integer in the set $\{0,1,\dots, k-1\}$ for which the
pair of variables are not $\mv t$-alignable. 
In particular, for this value of $p$ the pair
$$R\left[i^{(v_1)}_R, 
        i^{(v_1,v_2)}_R, \dots,
        i^{(v_1,\dots,v_p)}_R, 
        i^{(v_1,\dots,v_p, \ell)}_R
\right]$$ and
$$R\left[i^{(v_1)}_R, 
        i^{(v_1,v_2)}_R, \dots,
        i^{(v_1,\dots,v_p)}_R, 
        i^{(v_1,\dots,v_p, h)}_R
\right]$$ 
are not $\mv t$-alignable due to the fact that both of them are not equal to
$t_{s(p+1)}$, not because a prefix variable wasn't alignable.
In particular, $t_{s(p+1)}$ falls strictly in the open interval between these
two variables.

For this pair, the constraint
added in line \ref{line:additional-constraint} is not empty.
And, each variable in this non-$\mv t$-alignable pair is involved
in a comparison in $\cert$.
We will pay for $\mv t$ by charging this pair of comparisons.
(If one end of this pair is out of range, we will only charge the
non-out-of-range end. The other end is either $-\infty$ or $+\infty$.)

Finally, we want to upper bound how many times a pair of comparisons
is charged. 
Consider a pair of non-$\mv t$-alignable variables
\[ e^\ell = R\left[i^{(v_1)}_R, 
        i^{(v_1,v_2)}_R, \dots,
        i^{(v_1,\dots,v_p)}_R, 
        i^{(v_1,\dots,v_p, \ell)}_R
\right]\] and
\[ e^h = R\left[i^{(v_1)}_R, 
        i^{(v_1,v_2)}_R, \dots,
        i^{(v_1,\dots,v_p)}_R, 
        i^{(v_1,\dots,v_p, h)}_R
\right].\]
Each of $e^\ell$ and $e^h$ is involved in a comparison in $\cert$, and  
we need to bound the total charge for these pairs of comparisons.
We think of the pair of comparisons as an interval between 
$e^\ell$ and $e^h$ in a high dimensional space.

To see the charging argument, let us consider a few simple cases.
When $p=0$, then the interval between $e^h = R[i^h_R]$ and $e^\ell =
R[i^\ell_R]$ is a band from one hyperplane $H_1$ to
another hyperplane $H_2$ of the output space $\outspace$. 
This band consists of all points in $\outspace$ whose 
$A_{s(1)}$-values are between $R[i^\ell_R]$ and $R[i^h_R]$.
We call such an interval an $n$-dimensional interval.
Due to the constraint added in line \ref{line:additional-constraint},
a probe point $\mv t$ from a later iteration cannot belong to the
band. However, $\mv t$ might belong to the ``left'' of $H_1$ or
the ``right'' of $H_2$, in which case a new $n$-dimensional
interval might be created that is charged to the comparison involving
$H_1$ or involving $H_2$. Consequently, each comparison from a
$n$-dimensional interval can be charged twice.

When $p=1$, the interval between $e^\ell$ and $e^h$ 
is an $(n-1)$-dimensional interval which is
a band lying inside the hyperplane whose $A_{s(1)}$-value is
equal to $R[i^{(v_1)}_R]$. In this case, each comparison might be charged $4$ 
times: one from one side of the hyperplane, one from the other side, and twice 
from the two sides {\em inside} the hyperplane itself.

It is not hard to formally generalize the above reasoning to show that 
the comparison involving 
$$R\left[i^{(v_1)}_R, 
        i^{(v_1,v_2)}_R, \dots,
        i^{(v_1,\dots,v_p)}_R, 
        i^{(v_1,\dots,v_p, \ell)}_R
\right]$$ or
$$R\left[i^{(v_1)}_R, 
        i^{(v_1,v_2)}_R, \dots,
        i^{(v_1,\dots,v_p)}_R, 
        i^{(v_1,\dots,v_p, h)}_R
\right]$$ might be charged $2^{p+1}$ times.
Hence, the total number of iterations is at most $O(2^r|\cert| + Z)$.

The total number of constraints inserted into the data structure
$\cds$ is at most $O(m4^r|\cert| + Z)$, because if the probe point $\mv t$ is an
output tuple then only one constraint is inserted, and when $\mv t$ is not an
output tuple then at most $m2^r$ constraints are inserted.

As for the total run-time, when $\mv t$ is an output tuple the iteration does
about $O(mr\log N)$ amount of work. When $\mv t$ is not an output tuple, the
amount of work is $O(m2^r\log N)$. Hence, the total runtime is
$O((4^r|\cert|+rZ)m\log N)$, not counting the total time $\cds$ takes.

\noindent
{\bf Case 2.} Now, suppose some input relation has some private
attribute. Although this was supposed to be the easier case, and it is, it
still needs to be handled with delicate care to rigorously go through.
Let us see where the above proof might fail.

The proof fails at the main claim above. 
When we construct the database instance $J$, in order to use the fact that
\[R\left[i^{(v^R_1)}_R, i^{(v^R_1,v^R_2)}_R, \dots,
    i^{(v^R_1,\dots,v^R_j)}_R \right] \neq t_{s(j)}
\]
to conclude that the set \eqref{eqn:general-witness} is no longer a witness for
$Q(J)$, we crucially use the assumption that there is another relation having
the attribute $A_{s(j)}$. But, it might be the case that in the alignable
variables in \eqref{eqn:alignable-vars}, all the variables on non-private 
attributes already aligns perfectly with $\mv t$, and only private attributes
give us the gap (i.e. $\neq t_{s(j)}$ above) we wanted. In this case, the set
\eqref{eqn:general-witness} actually {\em is} a witness for $Q(J)$ too.
What happens really is that the probe point $\mv t$ is in between output
tuples. All of the non-private attributes already align! 
In this case, we actually need to charge one of the output tuples that align
with all the non-private attributes (and align with $\mv t$).
It is not hard to see that each output tuple is still charged only a constant
number of times this way.
\ep

% ------------------------------------------------------------------------
\subsection{Proof of Theorem~\ref{thm:cds-limit}}
\label{app:subsec:thm:cds-limit}

\bp
Recall from Theorem \ref{thm:analyze-outer-algorithm} that the total runtime of
Algorithm \ref{algo:outer-algorithm} is 
\[O\left( \left( 4^r|\cert| + rZ \right) m \log(N) + T(\cds) \right),\] 
where $T(\cds)$ is the total time taken by the
constraint data structure.  And that the algorithm inserts a total of
$O(m4^r|\cert| + Z)$ constraints to the $\cds$ and issues $O(2^r|\cert| + Z)$
calls to $\getpp()$.  Recall also from Proposition
\ref{prop:optimal-certificate-upperbound} that the optimal certificate size is
only $O(N)$. From this, we show the ``negative'' result stated.

We prove this theorem by using the reduction from {\sc unique-clique} to the
natural join evaluation problem. The reduction is standard
\cite{DBLP:conf/pods/PapadimitriouY97}. The {\sc unique-$m$-clique} input
instance ensures that the output size is at most $1$.  In the reduction the
clique size $m$ will become the number of relations.  The result of Chen et al.
\cite{DBLP:conf/soda/ChenLSZ07} implied that if there was an $O(N^{o(m)})$-time
algorithm solving {\sc unique-$m$-clique}, then the exponential time hypothesis
is wrong, and many $\mathsf{NP}$-complete problems have sub-exponential running
time.  Consequently, if there was a constraint data structure satisfying the
stated conditions, the exponential time hypothesis would not hold.
\ep

% ------------------------------------------------------------------------
\section{Details on the \cds}
\label{app:sec:cds}
% ------------------------------------------------------------------------

This section provides more details on how the \cds is built, and analyzes some
of its basic operations.
The \ctree \ data structure is an implementation of the \cds and we use the two
terms interchangeably. We first need two basic building blocks 
called $\arr$ and $\iarr$.

\subsection{The $\arr$ building block}

A $\arr$ data structure $L$ can store $N$ numbers in {\em sorted} order with 
the following operations:
\begin{enumerate}
\item $L.\fnd(v)$ returns {\sc true} if $v \in L$, {\sc false} otherwise.
\item $L.\fndlub(v)$ returns the smallest $v' \geq v$ in $L$. Return {\sc false} if no such $v'$ exists.
% \item $L.\get(i)$ returns the value at the $i$'th position in the sorted order.
\item $L.\ins(v)$ inserts the value $v$ into $L$
\item $L.\del(v)$ deletes value $v$ from $L$ 
\item $L.\deli(\ell,r)$ deletes all the values stored in $L$ that are in the 
interval $(\ell,r)$, where $\ell, r \in \mathbb Z \cup \{-\infty, +\infty\}$.
\end{enumerate}

\begin{rmk}
Even though we defined the $\arr$ data structure for numbers one can of course 
store more complex elements as long as there is a key value whose domain is 
totally ordered.
\end{rmk}

\begin{prop}
\label{prop:slist}
There exists a data structure that implements a $\arr$ $L$ with $N$ elements 
such that the first four operations above can be performed in time 
$O(\log{N})$ in the worst-case and $\deli$ can be implemented in 
$O(\log{N})$ amortized time.
\end{prop}
\begin{proof}
By using any balanced binary search trees such as AVL or Red Black trees,
the claim on the first four operations follow immediately. For the claim on 
$\deli$ note that this implies figuring out the index $i'$ of the smallest 
$\ell'\ge \ell$ and the index $j'$ of the smallest $r'\ge r$ in $L$. This can be 
done by calling $\fndlub$. Then we perform $\del$ operations on elements from 
index $i'$ (and $i'+1$ if $\arr(i') = \ell$) to index $j'-1$. There might be 
many of these $\del$ operations but since each of those deleted elements must 
be added at some earlier point, leading to an overall $O(\log{N})$ 
amortized runtime.
\end{proof}

\subsection{The $\iarr$ building block}

The next building block is called an $\iarr$. This data structure stores 
open intervals $(\ell,r)$, where 
$\ell, r \in \mathbb Z \cup \{-\infty, +\infty\}$, and supports the 
following operations (where $I$ is the $\iarr$):
\begin{itemize}
\item $I.\nxt(v)$ returns the smallest integer $v'$ such that 
(i) $v\le v'$ and (ii) $v'\not\in (\ell,r)$ for every $(\ell,r) \in I$.
\item $I.\covers(v)$ returns whether the integer $v$ is covered by some interval in $I$
\item $I.\ins(\ell,r)$ inserts the interval $(\ell,r)$ into $I$.
\end{itemize}

It is not hard to show that one can use a variation of any {\em segment tree} or
{\em interval tree} data structure to construct an $\iarr$ with 
the above operations taking logarithmic {\em amortized} cost.
In the following proposition, for completeness we describe a simple
implementation of $\iarr$ using $\arr$.

\begin{prop}
\label{prop:iarr}
An $\iarr$ can be implemented on $N$ intervals such that
the $\nxt$ and $\covers$ work in $O(\log{N})$ worst-case time, and 
$\ins$  operates in $O(\log{N})$ amortized time.
\end{prop}
\bp 
The main idea is to store the $N$ intervals as {\em disjoint} intervals. The 
end points are then stored in a $\arr$ and then we use the various operations 
of an $\arr$ to implement the operations of $\iarr$. Details follow.

At any point of time we maintain $m\le N$ {\em disjoint} intervals $(s_i,t_i)$
($i\in [m]$) such that
\[\bigcup_{i=1}^N (\ell_i,r_i)=\bigcup_{j=1}^m (s_j,t_j).\]
We then store all unique numbers $s_1\le t_1\le s_2\le t_2 \le\cdots\le t_m$ in
a $\arr$ $L$ (with two extra bits of information saying whether the number of the
left end point of an interval, right end point of an interval, or both left end
point of some interval and right end point of another interval). For example,
with $2$ intervals $(2,5), (5,9)$, elements in $\arr$ $L$ will be stored as
follows: $(2,L), (5,M), (9,R)$, where $2$ is the left end point of $(2,5)$, $5$
is the mixed point (i.e. the right end point of $(2,5)$ and the left end point
of $(5,9)$), and $9$ is the right end point of $(5,9)$.  Next we show how we
implement the three operations needed on $I$.

We begin with $I.\nxt(v)$. Let $u = L.\fndlub(v)$. If there is no such $u$, then
return $v$. Otherwise, if $u$ is the right end point of an interval or the mixed
point, then $v$ is in an interval whose right end point is $u$ or $v = u$. In
this case, return $u$. Finally, if $u$ is the left end point, then $v$ is not
covered by any interval in $L$. And so, return $v$. By
Proposition~\ref{prop:slist}, all this can be implemented in
$O(\log{m})=O(\log{N})$ time.
 
Next, we consider $I.\covers(v)$. Let $u = L.\fndlub(v)$. If there is no such
$u$, then return false. Otherwise, if $u$ is the right end point or the mixed
point, and $u \neq v$, then $v$ is covered by some interval. In this case, it
returns true. Finally, if $u = v$ or $u$ is the left end point, then $v$ is not
covered by any interval; and so it returns false. It is easy to check that only
$\fndlub$ is used and by Proposition~\ref{prop:slist} this takes $O(\log{N})$
time.

We consider $I.\ins(\ell,r)$ operation. The main idea is to delete all elements
between $\ell$ and $r$ in $\arr$ $L$ and then adjust the left end point and the
right end point. We now present the details. First by $\covers$ operation, we
determine if $l$ and $r$ are covered by some intervals in $I$. Now run
$L.\deli(l,r)$ that will delete all elements in $\arr$ $L$ that are strictly
between $\ell$ and $r$. Finally, we need to adjust the end points $\ell$ and
$r$. If $\ell$ is covered by some interval in $I$, then no action needs to be
taken on this side; the newly inserted interval $(\ell,r)$ will be merged with
the existing interval in $I$ on the left side. Consider the case when there is
an entry $(\ell,b)$ in $\arr$ $L$, where b indicates whether $\ell$ is a left
end point, right end point, or a mixed point. This can be checked by
$L.\fnd(\ell)$. If $\ell$ is the right end point, then update this entry in
$\arr$ $L$ by $(\ell,M)$. In the final case, if $\ell$ is not covered by any
interval and no entry $(\ell,b)$ exists in $\arr$ $L$, then insert $(\ell,L)$
into $L$. That is all about handling the left side. A similar argument handles
the argument for the right side $r$.  To analyze the running time, note that all
operations used are $I.\covers$,$L.\ins$, $L.\fnd$, and $L.\deli$. By
Proposition~\ref{prop:slist} and the above reasoning, this leads to an overall
$O(\log{N})$ amortized runtime.
\ep

% ------------------------------------------------------------------------
\subsection{The $\ctree$ and the $\insconst$ operation}

As described in Section~\ref{sec:cds-outline}, a $\ctree$ is a tree with $n$
levels, one for each of the attributes with the root as the first attribute
in the GAO.  (See also Figure~\ref{fig:ctree}.) The two key data structures
associated with each node are implemented using $\arr$ and $\iarr$: $v.\chld$ is
a $\arr$ and $v.\intv$ is a $\iarr$.

%\subsection{The $\insconst$ operation}
% ------------------------------------------------------------------------

We next describe how the \cds supports $\insconst$. The operation $\insconst$ is
supported by a member function of $\ctree$ called $\inst$ that takes 
as parameter a constraint vector $\mv c= \langle c_1,\dots,c_n \rangle$.
Algorithm~\ref{algo:insert} inserts a constraint vector into a $\ctree$.

\begin{algorithm*}[!thp]
\caption{$\cds.\inst(\vc)$}
\label{algo:insert}
\begin{algorithmic}[1]
\Require{A $\ctree$ $T$ and a constraint vector 
         $\vc= \langle c_1,\dots,c_n \rangle$.}
\Ensure{Update the data structure with $\vc$.}
\Statex

\State $i \gets 1$
\State $v \gets \text{root}(T)$
\While {$c_i$ is not an interval component} \Comment{$c_i \in \mathbb N \cup
\{*\}$}
  \If{$c_i \in \mathbb N$ and $v.\intv.\covers(c_i)$} \label{step:null-check1}
      \State \Return \Comment{$\vc$ is subsumed by an exiting constraint} 
             \label{step:null-check2}
  \ElsIf {$(v.\chld.\fnd(c_i) = \textsc{false})$} 
         \Comment{search even if $c_i = *$}
         \State $v.\chld.\ins(c_i)$ 
         \State Create a new node in $T$ and point $v.\chld(c_i)$ to it 
         \label{step:create-node}
  \EndIf
  \State $v \gets v.\chld(c_i)$
  \State $i \gets i+1$
\EndWhile
\State Suppose $c_i = (\ell,r)$ \Comment{$c_i$ is an interval component}
\State $v.\intv.\ins(\ell,r)$ 
   \Comment{Insert the interval into interval list} \label{step:insert-int}
\State $v.\chld.\deli(\ell,r)$\Comment{Update the $v.\chld$}
\label{step:child-clear}
\end{algorithmic}
\end{algorithm*}

From the description above Proposition~\ref{prop:insert-ctree}
follows straightforwardly.
Note again that when we insert a new interval that covers a lot of existing
intervals we will have to remove existing intervals; hence the cost is
amortized.

% ------------------------------------------------------------------------
\section{$\beta$-acyclic queries}
\label{app:sec:beta-acyclic}
% ------------------------------------------------------------------------

This section analyzes $\getpp$ algorithm for $\beta$-acyclic queries. The key
assumption is that the GAO has to be a nested elimination order, which
as shown in Section~\ref{app:sec:gao} precisely characterizes 
$\beta$-acyclic queries. Since we deal extensively with the partially
ordered sets formed by patterns, Figure~\ref{fig:pattern-of-node} should help 
visualizing these posets.
\begin{figure}[th]
\centerline{\includegraphics[width=2.8in]{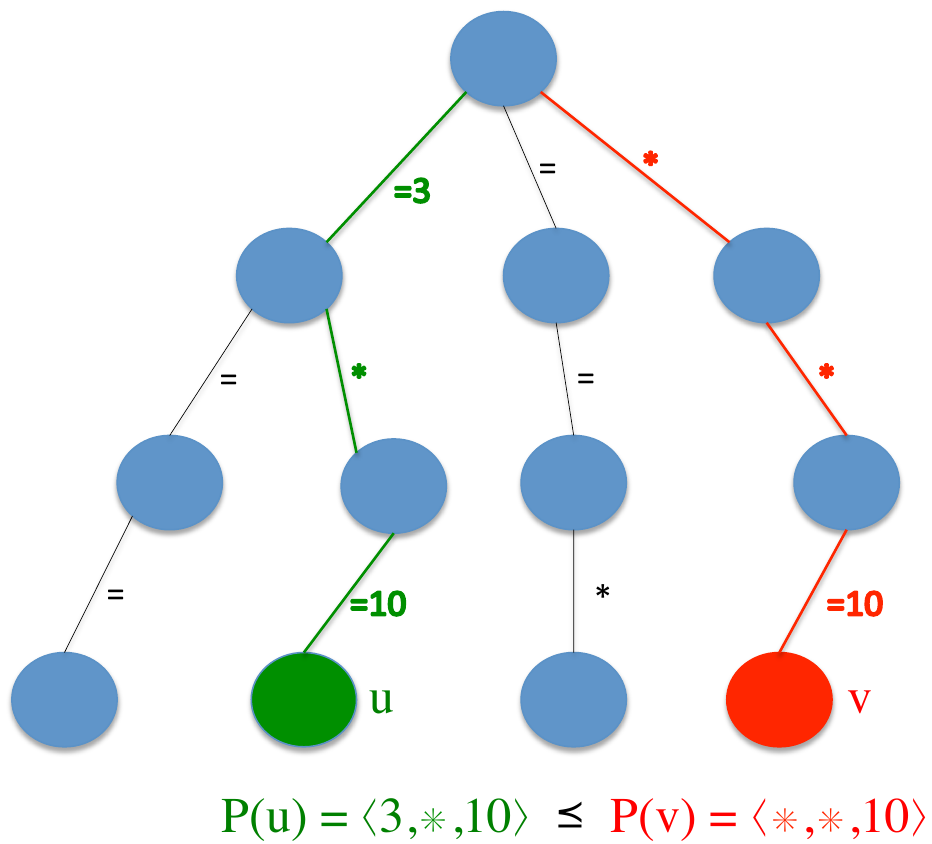}}
\caption{Patterns of nodes and the notion of specialization}
\label{fig:pattern-of-node}
\end{figure}

The meaning of the terms ``specialization'' and ``generalization'' are as
follows. Suppose $P(u)$ is a specialization of $P(v)$. Then, the constraints
stored in $v$ are of ``higher-order'' than the constraints stored in $u$. 
To be more concrete, suppose $P(u) = \langle 3, 5 \rangle$ and
$P(v) = \langle *, 5 \rangle$. Then, for a tuple $\mv t = (t_1,t_2,t_3)$ to
satisfy a constraint stored in $P(u)$, it must be the case that $t_1=3$,
$t_2=5$, and $t_3$ belongs to some interval stored in $u.\intv$.
On the other hand, for the tuple to satisfy $P(v)$ we only need $t_2 = 5$ and
$t_3 \in v.\intv$.

% ------------------------------------------------------------------------
\subsection{Proof of Proposition~\ref{prop:chain}}
\label{app:subsec:prop:chain}

\bp%[Proof of Proposition~\ref{prop:chain}]
Let $t_1,\dots,t_i$ be an arbitrary prefix. Recall that the principal filter
$G = G(t_1,\dots,t_i)$ is a set of nodes $u$ -- or equivalently the set of patterns
$P(u)$ -- which are above the pattern $\langle t_1,\dots, t_i\rangle$ in the
partial order defined in Section~\ref{subsec:pattern-neo}.
In particular, for every pattern $P(u)$ in $G$, its equality component comes
from one of $\{t_1,\dots,t_i\}$. It follows that $G$ is isomorphic to a
sub-poset of the prefix poset $\calP_{i+1}$, which 
is a chain by Proposition~\ref{prop:beta-gao}.

Note the important fact that, strictly speaking, the patterns in
$G(t_1,\dots,t_i)$ might come from constraints inserted from relations, or
constraints inserted by the outputs of the join. The constraints corresponding
to the outputs of the joins always match every entry in a prefix $\langle
t_1,\dots,t_i\rangle$, hence even though Proposition~\ref{prop:beta-gao} only
infers that the patterns from input relations form a chain, we can safely
conclude that the entire poset $G(t_1,\dots,t_i)$ is a chain.
\ep

% ------------------------------------------------------------------------
\subsection{Proof of Lemma~\ref{lmm:amortized-analysis-acyclic-Q}}
\label{app:subsec:lmm:amortized-analysis-acyclic-Q}

This section analyzes the overall run-time of $\ms$ for $\beta$-acyclic
queries. Let us first summarize what we know so far.
Theorem \ref{thm:analyze-outer-algorithm} showed that
the total runtime of $\ms$ (Algorithm \ref{algo:outer-algorithm}) is
\[O\left( \left( 4^r|\cert| + rZ \right) m \log(N) + T(\cds) \right),\]
where $T(\cds)$ is the total time it takes the constraint data structure.
$\ms$ inserts a total of $O(m4^r|\cert| + Z)$ constraints to 
$\cds$ and issues $O(2^r|\cert| + Z)$ calls to $\getpp()$, where $\cert$ is any
certificate, $Z$ is the output size, and $r$ is the maximum arity over all
relations. What we will prove in this section is,
provided that the global attribute order is the nested 
elimination order, we have
\[ T(\cds) = O\left((4^r|\cert|+Z)mn2^n\log N\right). \]
This means the overall runtime is
$O\left( mn2^n \left( 4^r|\cert| + Z \right) \log(N)\right)$. 
Hence, in terms of data complexity the runtime is nearly optimal: $\tilde{O}(|\cert| + Z)$.

\begin{lmm}[Re-statement of Lemma~\ref{lmm:amortized-analysis-acyclic-Q}]
Suppose the input query $Q$ is $\beta$-acyclic, and the global attribute order
$A_1,\dots,A_n$ is a nested elimination order, then each of the two 
operations $\getpp()$ and $\insconst()$ of $\ctree$ takes amortized time 
$O(n2^n\log W)$, where $W$ is the total number of constraints ever 
inserted into $\ctree$.
\end{lmm}
\bp
For each node $u\in \cds$, let $|P(u)|$ denote the number of equality components
in the pattern $P(u)$. For example, if $P(u) = \langle *, *, * \rangle$ then
$|P(u)| = 0$; and if $P(u) = \langle *, 3, 2 \rangle$ then $|P(u)| = 2$.
Note that $|P(u)| \leq n-1$, for all $u\in \cds$.

Our proof strategy is as follows. We equip each of the $\insconst$ and 
$\getpp$ operations with $O(n2^n\log W)$ ``credits.''  We then show that those 
credits are
sufficient to account for the runtime of each operation, {\em and} at the 
same time maintain the following {\em interval-credit invariant}.

{\bf Interval credit invariant:} 
for every node $u\in \cds$, and for every interval in the list $u.\intv$,
there is always a reserve of at least $\bigl(2^{|P(u)|+1}-2\bigr)c \log W$ 
credits at any point in time, where $c$ is a constant to be specified later.
(Note that, by definition, if $|P(u)| = 0$ then the intervals in $u.\intv$ do 
not need any reserve credits to maintain the invariant.)

First, for the $\insconst$ operation, the interval-credit invariant is easy 
to maintain.
From Proposition~\ref{prop:insert-ctree}, $O(n\log W)$ credits 
per operation is already sufficient; furthermore, we have up to
$O(n2^n\log W)$ credits to
spend. Hence, we have more than enough to give 
$(2^n-2)c \log W$ credits to the interval of the new constraint for a large enough constant $c$.
In fact, we will be very generous by assigning credits as follows.
\bi
 \item We give the interval component of the newly inserted constraint 
$(2^n-2)c\log W$ credits to maintain the interval-credit invariant. Note that
$2^n \geq 2^{|P(u)|+1}$ for any node $u$ in the tree.
 \item We give each component $c_i$ (equality or wildcard)
that comes {\em before} the interval component of the new constraint
$5 \cdot 2^n c\log W$ credits. 
How these credits will be used is explained below.
\ei
Overall, each $\insconst$ operation requires at most
\[ n2^{n+3}c\log W = O(n2^n\log W) \] credits as desired.
Note again that $n2^{n+3}c\log W$ credits is a lot more than what is required 
for the $\insconst$ operation by itself.
We need the extra credit to pay for something else down the line.

Next, we consider a $\getpp$ operation. We iterate through the depth $i$ of the
\cds, for $i$ goes from $0$ (the root) to $n-1$ (a leaf). At each depth $i$,
we try to compute the value $t_{i+1}$, backtracking if necessary.
The crucial observation is the following: at each depth $i$ of the
algorithm, 
thanks to Proposition~\ref{prop:beta-gao},
the set $G$ forms a {\em totally ordered set} because the global 
attribute order is a nested elimination order.
Note that Proposition~\ref{prop:beta-gao} only considers
input relations. In the constraint tree there might be constraints
inserted due to the output tuples. However, those constraints are always the most
specific (i.e. they are at the bottom of any poset they participate in),
and thus in any poset $G$ at any depth $i$ the pattern coming from an
output-initiated constraint is a specialization of any 
input-initiated pattern. 
Furthermore, and this is a slightly subtle point, there are also intervals
inserted due to backtracking; but luckily due to the chain property
of the prefix poset $\calP_k$, for $\beta$-acyclic queries the backtracking
intervals have patterns which are just the same as the patterns 
from input-generated intervals.
Since $G$ is a totally ordered set, it has a bottom element $\bar u \in G$.

The basic idea is to show that at each depth $i$ the algorithm takes 
$O(2^n\log W)$-amortized time, accounted for by using newly infused 
$O(n2^n\log W)$ credits from $\getpp$ 
{\em and} the reserved credits from existing intervals guaranteed by the 
invariant.
At the same time, we need to still maintain the interval credit invariant
and thus we cannot abuse the banked reserve of the data structure.
In particular, intervals whose reserved credits have been used up have to
somehow ``disappear'' or be infused with fresh credits to maintain the
invariant.

Specifically, we will equip the $\getpp$ exactly $n2^{n+1}c\log W$ credits,
distributing precisely $2^{n+1}c\log W$ credits to each depth $i$ of the tree.
These credits will be called the {\em depth-$i$} credits of $\getpp$.

Fix an iteration at depth $i\in \{0,\dots,n-1\}$ of the \cds.
If $G = \emptyset$ (line~\ref{line:beta-Gempty} of
Algorithm~\ref{algo:getpp-beta}), then we move on to the next depth and 
hence depth-$i$ credits of $\getpp$ is more than sufficient to spend here, 
assuming $c$ is sufficiently large.
Henceforth, suppose $|G| \geq 1$. Note that we are still considering
depth $i$.

{\bf Case 1.} Let us first assume that there is no backtracking at this 
depth.
Let $\bar u = u_k \prec u_{k-1} \prec \cdots \prec u_1$ be the members of the 
poset $G$, which as explained above is a total order. 
We will show by induction the following claim.

{\bf Claim.} For every $j \in [k]$,
the call $\cds.\nextvalue(x, u_j, G)$ takes amortized time
$$(2^{|P(u_j)+2|}-3) c \log W,$$ while maintaining the interval credit 
invariant.
In other words, we need to use $(2^{|P(u_j)+2|}-3) c \log W$ credits from somewhere
to pay for this call.

From the claim, and from the fact that
\[ n-1 \geq |P(u_k)| > |P(u_{k-1})| > \cdots > P(u_1) \geq 0 \]
the initial call $\cds.\nextvalue(-1, u_k, G)$ 
(line~\ref{line:initial-call-beta} of Algorithm~\ref{algo:getpp-beta}) 
takes time at most $(2^{n+1}-3) c \log W = O(2^n\log W)$.
Consequently, the depth-$i$ credits of $\getpp$ is sufficient to pay for the
call.

We next prove the claim by induction. The base case is when $j=1$,
i.e. when we are calling
$u_1.\intv.\nxt(x)$. Line~\ref{line:nextvalue-basecase}
of Algorithm~\ref{algo:ctree-nextvalue-beta} 
takes $O(\log W)$-time, thanks to Proposition~\ref{prop:iarr}. 
Note that, 
\[ (2^{|P(u_1)|+2}-3) c \log W \geq c\log W. \]
Hence, with $c$ large we have enough credits to pay for the call.

Next, consider $j \geq 2$ and assume the claim holds for $j-1$. 
Consider a call to $\cds.\nextvalue(x, u_j, G)$. 
An iteration of Algorithm~\ref{algo:ctree-nextvalue-beta} has two
steps: 
(a) line \ref{line:z-beta} takes $(2^{|P(u_{j-1})|+2}-3)c\log W$-credits 
by the claim's induction hypothesis,
and 
(b) line \ref{line:y-beta} takes $c \log W$-time for $c$ large.
In total, each iteration takes time at most
\[ \bigl(2^{|P(u_{j-1})|+2}-3\bigr)c\log W + c \log W \leq
   \bigl(2^{|P(u_j)|+1}-2\bigr)c\log W.
\]
If we had to continue on with the next iteration ($y\neq z$), then it must be the 
case that $z \in (\ell, y)$ for some interval $(\ell, y) \in u_j.\intv$.
By the interval credit invariant, this interval $(\ell, y)$ has
a credit reserve of 
$\bigl(2^{|P(u_j)|+1}-2\bigr)c\log W$, which by the above inequality is
sufficient to pay for the next iteration!
This process repeats itself.

Consequently, the reserves of credits at $u_j$-intervals that $z$ hits pay for 
subsequent iterations.
We are left to pay for
$(i)$ the first iteration,
$(ii)$ the insertion of the new interval in line~\ref{line:new-constraint-beta}, 
{\em and} 
$(iii)$ fresh credits to deposit to the newly inserted interval to maintain
the invariant. 
It is crucial to notice that the newly inserted interval ``consumes'' all
intervals whose credits we have used up to pay for subsequent iterations.
Hence, by paying for $(i)$, $(ii)$, and $(iii)$ above we are done.

With large $c$, the insertion in line~\ref{line:new-constraint-beta} takes time
at most $c \log W$. Hence,
the missing amount in all these three unpaid operations is at most
%\begin{multline*}
$$
    \left((2^{|P(u_{j-1})|+2}-3)c\log W + c \log W\right) + c \log W %\\
   + (2^{|P(u_j)|+1} -2) c \log W
   \leq (2^{|P(u_j)|+ 2}-3)c\log W.
$$
%\end{multline*}
This proves the claim.
Figure~\ref{fig:|G|=2} illustrates the induction reasoning and where all the
credits go.
\begin{figure*}[!thp]
\centerline{\includegraphics[width=4.5in]{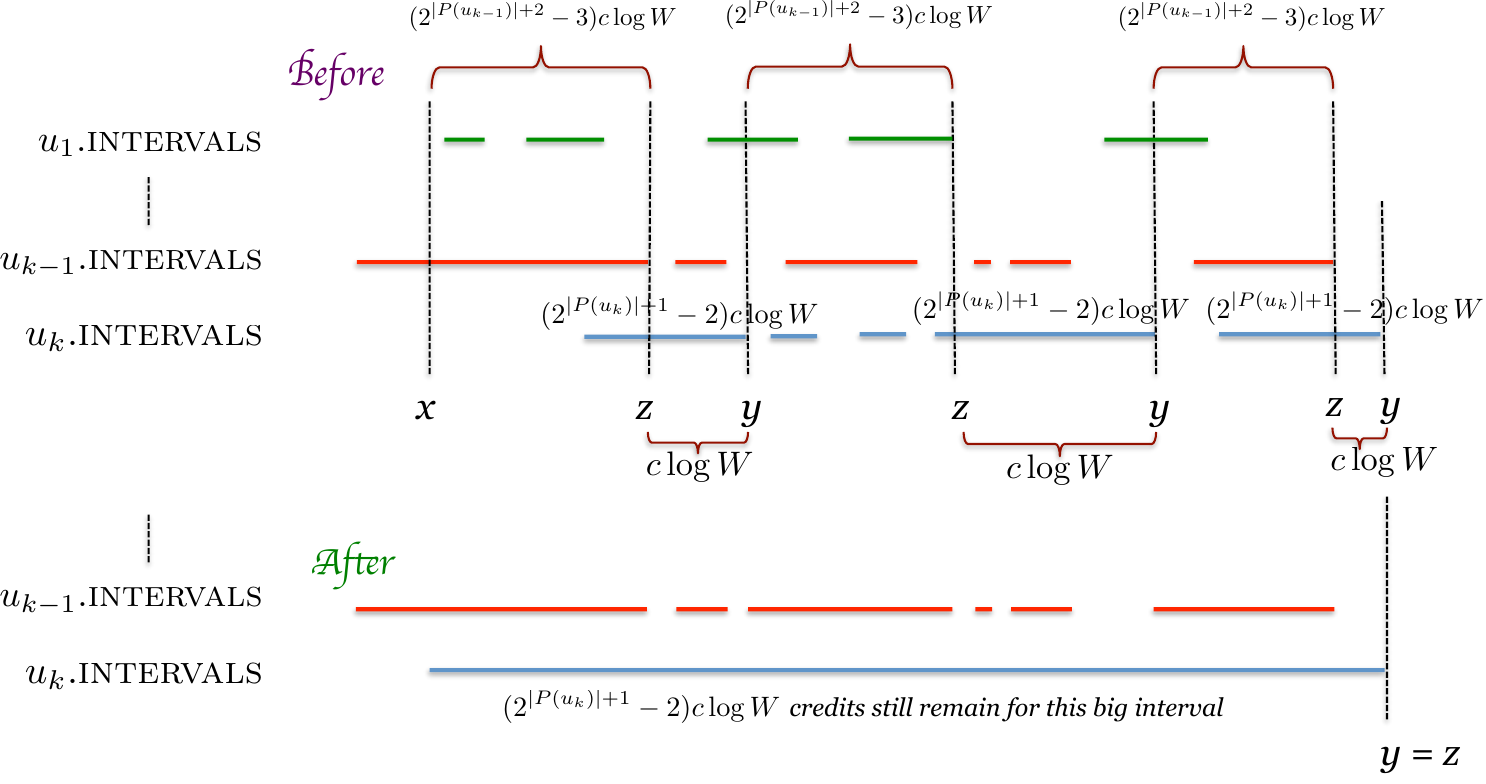}}
\caption{Illustration of the analysis of Algorithm~\ref{algo:ctree-nextvalue-beta}}
\label{fig:|G|=2}
\end{figure*}

{\bf Case 2.} Next, we consider the case when there is some backtracking.
When backtracking occurs in line~\ref{line:backtrack}, we have 
$5 \cdot 2^n c \log W$ credits for {\em each} of the components of the 
constraint $\bar u$ from $\bar p_{i_0}$ to $\bar p_i$.
We will use these credits as follows.
\bi
 \item $2^nc \log W$ credits of $\bar p_{i_0}$ is used for the insertion in
 line~\ref{line:backtrack} itself, 
 \item $2\cdot 2^nc \log W$ credits of $\bar p_{i_0}$ is now considered fresh 
 depth-$(i_0-1)$ credits of the $\getpp$ operation. 
 \item $2\cdot 2^nc \log W$ credits of $\bar p_{i'}$ for every $i'$ from $i_0$
 to $i$ are now considered fresh depth-$i'$ credits of $\getpp$.
 We need these fresh credits because the depth-$i'$ credits from $\getpp$ 
 have been used up when we visited depth up to $i$ before this backtracking
 step. 
 Luckily when we backtrack we will never visit the points
 $\bar p_{i_0}, \dots, \bar p_i$ again due to the constraint inserted in
 line~\ref{line:backtrack},
 and hence their credits can be used freely.
\ei
\ep

% ------------------------------------------------------------------------
\subsection{Proof of Proposition~\ref{prop:no-io-non-beta}}
\label{app:subsec:prop:no-io-non-beta}

To prove Proposition~\ref{prop:no-io-non-beta} we need an auxiliary lemma.
\begin{lmm}[\cite{3sum}]
\label{lem:mihai}
If $3$-\textsf{SUM} problem does not have a sub-quadratic algorithm, then for every $c\ge 3$, there exist $c$-partite graphs $G$ such that listing all $O(|E|)$ $c$-cycles in $G$ needs to take time $\Omega\left(|E|^{4/3-\eps}\right)$ for any $\eps>0$. Further, the graph $G$ can be written as $(V_1,V_2,\dots,V_c;E)$ where $E$ can be written as the disjoint union of edge sets $E_{i,i+1\mod{c}+1}\subseteq V_i\times V_{i+1\mod{c}+1}$.
\end{lmm}
\bp
The result for $c=3$ appears in~\cite{3sum}. The extension to the case of $c>3$ is simple: pick any two partitions and replace each edge (between the two partitions) by a path of length $c-2$. Note that the resulting graph is $c$-partite (with the claimed special structure on the edge set), has $O(c|E|)$ edges and has a $c$-cycle if and only if the original tri-partite graph has a triangle.
\ep

\bp[Proof of Proposition~\ref{prop:no-io-non-beta}]
Consider an arbitrary $\beta$-cyclic query $Q$ (with attribute set $\{A_1,\dots,A_n\}$ and relations/hyperedges $R_1,\dots,R_m$). Note that this implies that $Q$ has a $\beta$-cycle of length $c\ge 3$. W.l.o.g. assume that the cycle involves relations $R_1,\dots,R_c$ and the attributes $A_1,\dots,A_c$. In other words, for every $1<i\le c$, $\{A_{i-1},A_i\} = R_i\cap \{A_1,\dots,A_c\}$ and $\{A_1,A_c\}= R_1\cap \{A_1,\dots,A_c\}$. The idea is to embed the hard instance for listing $c$-cycles from Lemma~\ref{lem:mihai} into this cycle. Details follow.

Define
\[ Q'= \ \Join_{i=1}^c R_i.\]
Let $G=(V_1,V_2,\dots,V_c;E)$ be the hard instance for listing $c$-cycles from Lemma~\ref{lem:mihai}, where $E$ is the disjoint union of $E_{i,i+1\mod{c}+1}\subseteq V_i\times V_{i+1\mod{c}+1}$. %Define $E'_{i,j}=\{(u,v)| (u,v)\text{ or } (v,u)\in E_{i,j}\}$.
Further, define the instance for $Q$ as follows:
\begin{eqnarray*}
R_1 &=& \left\{(u,1,\dots,1,v)| (u,v)\in E_{1,c}\right\}
\times\left(\bigtimes_{c<j\le m: A_j\in R_1} \{1\}\right),\\
R_i&=&\left( \bigtimes_{j<i-1: A_j\in R_i} \{1\}\right)\times E_{i-1,i}\times
\left(\bigtimes_{j>i: A_j\in R_i} \{1\}\right) \text{ for } 1<i\le c,\\
R_i&=& \left(\bigtimes_{j\in [c]: A_j\in R_i} V_j\right) \times \left(\bigtimes_{c<j\le m: A_j\in R_i} \{1\}\right) \text{ for } c<i\le m.
\end{eqnarray*}
Note that the size of the output of $Q$ in the instance above is exactly the same as the size of the output of $Q'$. Further, there is a one-to-one correspondence between an output tuple of $Q'$ (and hence $Q$) and a $c$-cycle of $G$. (Indeed for each $c$-cycle $(v_1,\dots,v_c,v_{c+1}=v_1)$ in $G$ (where $v_j\in V_j$), the output tuple of $Q$ assigns $v_j$ to attribute $A_j$ for every $j\in [c]$ and every other attribute is assigned $1$.)

Since the relations $R_i$ for $i\in [c]$ are of size $O(|E|)$, Proposition \ref{prop:optimal-certificate-upperbound} implies that
the optimal certificate for $Q'$ has size $O(|E|)$. We claim that such a certificate can be extended to a certificate for $Q$
also of size $O(|E|)$ as we did in the proof of Proposition~\ref{prop:tw-eth}. Indeed, the certificate for $Q'$ is enough to pinpoint
which tuples are the output tuples with $O(|E|)$ comparisons
(or certify that the join is empty). Then $O(n\log(|E|))$ more comparisons can
verify whether each of the output tuples of $Q'$ can be extended to an output
tuple of $Q'$. (The formal argument is pretty much the same as the argument for
the proof of Proposition~\ref{prop:tw-eth}. The only difference is that $Q$ can
have attributes that are not in $Q'$ and for such attributes we have to check
that they have the same values in the projections to $R_j$ for appropriate
$c<j\le m$ but this can be done in the claimed time.)
Since $G$ has $O(|E|)$ $c$-cycles, this will produce an $O(|E|+c\cdot n\cdot|E|\cdot \log(|E|))=\tilde{O}_n(|E|)$ sized certificate for $Q$. Thus, an $O(|\cert|^{4/3-\eps})$ time algorithm to solve $Q$ for any $\eps>0$ would (by Lemma~\ref{lem:mihai}) imply a sub-quadratic time algorithm for $3$\textsf{SUM}, which is a contradiction.
\ep

% ------------------------------------------------------------------------
\section{General queries}
\label{app:sec:tw}
% ------------------------------------------------------------------------

In this section, we show that $\ms$ runs in time roughly 
$\tilde O( \cert^{w+1} + Z)$ for general queries ($\beta$-acyclic or not) where
$w$ is the elimination width of the GAO. In particular, if the query has
treewidth $w$, then there exists a GAO for which the above runtime holds,
thanks to Proposition~\ref{prop:gao-tw}.

The algorithm for general query is the same as that of the $\beta$-acyclic case;
the only (slight) difference is in $\getpp$, which is described next.

% ------------------------------------------------------------------------
\subsection{Algorithms}
\label{app:subsec:tw-algorithms}

The $\getpp$ algorithm for general queries
remains very similar in structure to that of the $\beta$-acyclic case
(Algorithm~\ref{algo:getpp-beta}), and if the input query is $\beta$-acyclic
with the nested elimination order as the global attribute order, then 
the general $\getpp$ algorithm is exactly Algorithm~\ref{algo:getpp-beta}.
The new issue we have to deal with lies in the fact that the poset $G$ at
each depth is not necessarily a chain.
Our solution shown in Algorithm~\ref{algo:getpp-tw}
is very simple and quite natural: we mimic the behavior of
Algorithm~\ref{algo:getpp-beta} on a ``shadow'' of
$G$ that {\em is} a chain and make use of both the algorithm and the analysis
for the $\beta$-acyclic case.

The ``shadow'' of $G$ is constructed as follows. 
Let $u_1,\dots,u_k$ be an arbitrary linearization of nodes in $G$, 
i.e. if $1 \leq i<j\leq k$, then either $P(u_i) \preceq P(u_j)$ or 
$P(u_i)$ and $P(u_j)$ are incomparable using the relation $\preceq$.
A linearization always exists because $\preceq$ is a partial order.
Now, for $j \in [k]$, define the patterns
\[ \bar P(u_j) = \bigwedge_{i=j}^k P(u_i). \]
Here $\wedge$ denotes ``meet'' under the partial order $\preceq$.
Then, obviously the shadow patterns form a chain:
\[ \bar P(u_1) \preceq \bar P(u_2) \preceq \cdots \preceq \bar P(u_k). \]
Note that it is possible for $\bar P(u_i) = \bar P(u_j)$ for $i\neq j$.
For example, suppose the patterns of nodes in $G$ are
\[ 
\langle a,b,* \rangle,
\langle *,b,* \rangle,
\langle *,*,* \rangle,
\langle a,*,c \rangle,
\langle *,b,c \rangle;
\]
and suppose we pick the following linearization of these patterns:
\[ 
\langle a,*,c \rangle,
\langle *,b,c \rangle,
\langle a,b,* \rangle,
\langle *,b,* \rangle,
\langle *,*,* \rangle.
\]
Then, the $\bar P$ patterns are as follows.
\[
%\small
\begin{matrix}
\text{Linearization:} &
\langle a,*,c \rangle &
\langle *,b,c \rangle &
\langle a,b,* \rangle &
\langle *,b,* \rangle &
\langle *,*,* \rangle\\
\text{The $\bar P$ patterns:} &
\langle a,b,c \rangle &
\langle a,b,c \rangle &
\langle a,b,* \rangle &
\langle *,b,* \rangle &
\langle *,*,* \rangle.
\end{matrix}
\]
It should be apparent from the above example the two claims we made earlier: the
shadow patterns form a chain, and some shadow patterns are the same.
To continue with the above example, $\getpp$ is supposed to return a free value
$d$ on attribute $D$ which does not belong to any interval in the interval
lists of the nodes 
\[ \langle a,*,c \rangle,
\langle *,b,c \rangle,
\langle a,b,* \rangle,
\langle *,b,* \rangle,
\langle *,*,* \rangle. \]
For each node $u$, we will operate as if its pattern was actually $\bar P(u)$.
Algorithm~\ref{algo:getpp-tw} has the details.

\begin{algorithm*}[!thp]
\caption{$\cds.\getpp()$ for general queries}
\label{algo:getpp-tw}
\begin{algorithmic}[1]
\Require{A $\ctree$ $\cds$}
\Ensure{Returns a tuple $\mv t=(t_1,\dots,t_n)$ that does not satisfy any stored
constraint}
\Statex

\State $i \gets 0$
\While {$i < n$}
  \State $G \gets \setof{u \in \cds}{(t_1,\dots,t_i) \preceq P(u) \text{ and }
  u.\intv \neq \emptyset}$
  %\Comment{$G$ instead of $G(t_1,\dots,t_i)$ for short}
  \If {($G = \emptyset$)} \Comment{$G$ will be empty for all later values of
  $i$}
    \State $t_{i+1} \gets -1$
    \State $i \gets i+1$
  \Else
    \State Let $G = \{u_1,\dots,u_k\}$, where $u_1,\dots,u_k$ is a linearization of $G$
    \State $\bar G \gets \emptyset$ \Comment{start constructing the shadow
    chain} \label{line:shadowchain-start}
    \For {$j \gets 1$ {\bf to} $k$}
      \State $\bar P(u_j) \gets \bigwedge_{\ell=j}^k P(u_\ell)$
    \Comment{$\wedge$ denotes {\em meet} under partial order $\preceq$}
      \If {$\cds$ has no node with pattern $\bar P(u_j)$} 
        \Comment{Create the shadow nodes}
        \State $\cds.\insconst(\langle \bar P(u_j), (-\infty, 0)\rangle)$
      \EndIf
      \State Add the pattern $\bar P(u_j)$ to $\bar G$ \Comment{$\bar G$ is a
      {\bf multiset}, also a chain, break ties arbitrarily}
    \EndFor \label{line:shadowchain-end}
    \State Let $\bar{\mv p} = \langle \bar p_1, \dots, \bar p_i \rangle$ be the
    {\em bottom} element of the poset $\bar G$ \label{line:bottom:multiset}
    \State Let $\bar u \in \cds$ be the node for which $P(\bar u) = \bar{\mv p}$
    \State $t_{i+1} \gets \cds.\nextvaluefromshadow(-1, \bar u, \bar G)$
    \Comment{Algorithm~\ref{algo:ctree-nextvalue}} \label{line:initial-call-tw}
    \State $i_0 \gets \max\{ k \suchthat k \leq i, \bar p_k \neq *\}$
    \If {($t_{i+1} = +\infty$) and $i_0=0$}
      \State \Return \NULL \Comment{No tuple $\mv t$ found}
\label{step:getpp-full}
    \ElsIf {($t_{i+1} = +\infty$)} 
      \State $\cds.\insconst(\langle \bar p_1,\dots, \bar p_{i_0-1}, (\bar
              p_{i_0}-1, \bar p_{i_0}+1) \rangle$ %\Comment{Subtle!}
      \label{line:backtrack-tw}
      \State $i \gets i_0-1$ \Comment{Back-track} 
    \Else
      \State $i \gets i+1$ \Comment{Advance $i$}
    \EndIf
  \EndIf
\EndWhile
\State \Return{$\mv t = (t_1,\dots,t_n)$}
\end{algorithmic}
\end{algorithm*}

%\ar{Is the bottom element in Line~\ref{line:bottom:multiset} in Algorithm~\ref{algo:getpp-tw} well defined since $G$ can be a multiset?}

\begin{algorithm*}[!thp]
\caption{$\cds.\nextvaluefromshadow(x, u, \bar G)$, where $\bar G$ is a chain}
\label{algo:ctree-nextvalue}
\begin{algorithmic}[1]
\Require{A $\ctree$ $T$, a node $\bar u \in \bar G$ to start the recursion with}
\Require{A (multiset) chain $\bar G$ of nodes, and a starting value $x$}
\Ensure{the smallest value $y\geq x$ not covered by {\em any}
$v.\intv$, for all $v \in \bar G$ such that $P(\bar u) \preceq P(v)$}
\Comment{note that $v$ could be an original node or a shadow node}
\Statex
\State Let $u$ be the node that $\bar u$ is a shadow of 
       \Comment{$u$ could be the same as $\bar u$}
\If {there is no $v \in \bar G$ for which $P(\bar u) \precdot P(v)$} 
     \Comment{At the top of the chain $\bar G$} \label{step:nextval-end}
  \State \Return $\cds.\nextvalue(x,\bar u, \{\bar u, u\})$
  \Comment{Algorithm~\ref{algo:ctree-nextvalue-beta}} \label{line:nextvalue-basecase-tw}
\Else
  \State $y \gets x$
  \Repeat
    \State Let $v\in \bar G$ such that $P(u) \precdot P(v)$ 
       \Comment{Next node up the shadow chain}
    \State $z \gets \cds.\nextvaluefromshadow(y, v, \bar G)$ \label{line:z-tw}
    \Comment{first ``free value'' $\geq y$ up the chain}
    \State $y \gets \cds.\nextvalue(z, \bar u, \{\bar u, u\})$ \label{line:y-tw}
    \Comment{first ``free value'' $\geq z$ at $u$}
  \Until {$y = z$}
  \State $\cds.\insconst(\langle P(u), (x-1,y)\rangle)$ 
  \label{line:new-constraint-tw}
  \Comment{All values from $x$ to $y-1$ are not available}
  \State \Return $y$
\EndIf
\end{algorithmic}
\end{algorithm*}

% ------------------------------------------------------------------------
\subsection{Proof of Theorem~\ref{thm:tw}}
\label{app:subsec:thm:tw}

\bp%[Proof of Theorem~\ref{thm:tw}]
Let us first go through the skeleton of Algorithm~\ref{algo:getpp-tw}.
We encourage the reader to view Algorithm~\ref{algo:getpp-beta} and
Algorithm~\ref{algo:getpp-tw} side-by-side. Their structures are identical
except in two places.
First, lines \ref{line:shadowchain-start} to \ref{line:shadowchain-end}
of Algorithm~\ref{algo:getpp-tw}
build a shadow chain because $G$ itself is not necessarily a chain as is the
case in Algorithm~\ref{algo:getpp-beta}.
Second, the call to $\nextvalue$ on line \ref{line:initial-call-beta} in
Algorithm~\ref{algo:getpp-beta} is replaced by the call 
on line \ref{line:initial-call-tw} 
of Algorithm~\ref{algo:getpp-tw}
to $\nextvaluefromshadow$.

For the moment, suppose the calls to $\nextvalue$ and 
$\nextvaluefromshadow$ take the same amount of time, then the only extra work
that Algorithm~\ref{algo:getpp-tw} does compared to 
Algorithm~\ref{algo:getpp-beta} per depth of the $\cds$ comes from building up
the shadow poset. It takes time $O(mn\log W)$ for each shadow poset. 
(Recall that $W$ is the number of intervals ever inserted into $\cds$ by
$\ms$.)
And, if we
also want to maintain the interval credit invariant then it takes 
$O(mn2^n\log W)$-time extra per depth per $\getpp$ operation.
Up to this point, we can mimic the proof from the $\beta$-acyclic case and
assign each $\getpp$ operation $O(mn^22^n\log W)$ credits.
The only difference (so far) from the analysis of the $\beta$-acyclic case is
that the amount of credits per depth assigned to $\getpp$ is blown up by a factor of
$mn$ due to the shadow poset construction and the extra credits needed for the
shadow intervals.
Note also that, modulo the difference between the $\nextvalue$ call
and the $\nextvaluefromshadow$ call, if the poset $G$ is a chain, then 
$\bar G$ is exactly $G$ and every node is a shadow of itself!
In this case, we do not need to do the extra work of building up the shadow
poset and indeed we ``get back'' Algorithm~\ref{algo:getpp-beta}.

Next, let us look at Algorithm~\ref{algo:ctree-nextvalue-beta}
and Algorithm~\ref{algo:ctree-nextvalue} side by side.
The key difference is in the calls to $\nxt$ on an interval list, we make a call
to $\nextvalue$ on the chain $\{\bar u, u\}$, where $\bar u$ is the shadow of
$u$. (Since $\bar u$ is the meet of all nodes from $u$ and above in the
linearization, $\bar u \preceq u$ and hence $\{\bar u, u\}$ is a chain.)
If we maintain the interval credit invariant, then each of these calls
to $\nextvalue(z, \bar u, \{\bar u, u\})$ 
takes amortized time $O(2^n\log W)$, a $2^n$-blowup compared to $\nxt$.
So far, we have an $mn2^n$ blowup factor (with a very loose analysis),
relative to the $\beta$-acyclic case.

Note again that, if the poset $G$ is a chain, implying $\bar G = G$, and 
every node is a shadow of itself, then the calls to 
$\intv.\nxt$ and to $\nextvalue(u, z, \{u\})$ are
identical. In this case $\nextvaluefromshadow$ is the same as $\nextvalue$
and we get back to the $\beta$-acyclic case.

In general, however, we cannot maintain the interval credit invariant by 
simply giving each inserted interval $O(n2^n\log W)$ credits (blown up by
$mn2^n$ more) as we have done in the $\beta$-acyclic case because the same
node $u$ might have different shadows depending on the prefix we are working on.
For example, the node $u = \langle *, b, * \rangle$ might have the shadows 
$\bar u^1 = \langle a, b, * \rangle$,
$\bar u^2 = \langle a', b, c' \rangle$,
$\bar u^3 = \langle *, b, c'' \rangle$,
and so forth. 
In this example, the number of credits we give to an interval in
the list $u.\intv$ has to be at least
three times as much as that in the 
$\beta$-acyclic case because $u$ might participate in the 
$\nextvalue$ calls with each of its shadows.
Consequently, we will have to give each inserted interval many
more credits than $O(mn^24^n\log W)$.
The key question is: {\em how many more credits}?

The number of credits assigned to each interval depends on the size of its
pattern, and on what type of interval it is.
From Theorem~\ref{thm:analyze-outer-algorithm}, we know that the number of 
intervals inserted into the $\cds$ is $O(m4^r|\cert| + Z)$. 
In fact, there are two types of intervals inserted: 
$Z$ intervals inserted by the output called the 
{\em output-generated intervals}, 
and $O(m4^r|\cert|)$ intervals inserted by the input relations, called the 
{\em input-generated intervals}.
There are also {\em backtracking intervals} created by the algorithm.
Hence, overall we have three types of intervals.

The overall credit-assignment scheme is intimately tied to the size of 
interval pattern under consideration.
Recall that 
the {\em size} of a pattern is defined to be
the number of equality components of the pattern.
For example, the pattern $P(u) = \langle a, b, * \rangle$ has size $2$
and the pattern $P(u) = \langle *,*,*\rangle$ has size $0$.

Consider the simplest case when $P(u)$ has size $0$, such as $P(u) =
\langle *,*,*\rangle$. In this case, the node $u$ is always on top of the
linearization $G$ and thus in $\bar G$ it is the shadow of itself. In other
words, it does not really have any shadow. 
The symmetric situation is when $P(u)$ has {\bf no} wildcard pattern. All of the
intervals that come from the outputs are of this type. 
This includes backtracking intervals which are created from a prefix of 
output-generated intervals.
In this case, $u$ is always at the {\em bottom} of the linearization
$G$ and thus in $\bar G$ it is also a shadow of itself.
For these types of intervals --
intervals in $u.\intv$ where $u$ is self-shadowed -- we can
give them the same credits as they get in the $\beta$-acyclic case 
(multiplied by the blowup factor).

Next, we analyze how many credits we need for intervals whose patterns 
are not self-shadowed at some point in the execution of $\ms$.
(A pattern might be self-shadowed at one point, but then not self-shadowed at
another time due to a different prefix.)
Let $P(u)$ be one such pattern with size $s$; 
let $k-1$ be the length of the pattern $P(u)$, i.e. intervals of $u$ are on
attribute $A_k$ in the global attribute order.
Then, the support of $P(u)$ (the positions of equality components) is precisely
a subset of the universe $U(\calP_k)$ of the prefix poset 
$\calP_k$ as defined in Section~\ref{app:subsec:prefix-poset}.
From Proposition~\ref{prop:gao-tw} we know that $s \leq w$,
and hence $1 \leq s \leq w-1$. 
(If $s \in \{0,w\}$ then $P(u)$ is self-shadowed.)
Since $s$ out of $w$ components in $P(u)$ are already fixed, the number
of different shadows $\bar P(u)$ of $u$ is at most $(m4^r|\cert|)^{w-s}$: it has
$w-s$ degrees of freedom, each of which can be attributed to some interval in
the set of input-generated intervals.
Consequently, we can give each interval in $u.\intv$ the following 
number of credits to pay for all
operations it's involved in: 
$(m4^r|\cert|)^{w-s} \cdot O(mn2^n\log W)$.
In the above accounting we did not need to distinguish between input-generated
intervals or backtracking intervals, namely $P(u)$ can be the prefix of a
backtracking interval too.

Let us summarize what we know thus far:
\bi
 \item Intervals whose patterns are self-shadowed get the same credits as in the
 $\beta$-acyclic case.
 \item Other intervals with size-$s$ patterns get $(m4^r|\cert|)^{w-s} \cdot
 O(mn2^n\log W)$ credits.
\ei
(Again, the credits are supposed to be multiplied by $mn2^n$.)
So our next task is to sum up all the credits we need and that will be the final
(amortized) runtime of $\ms$.

\bi
 \item Each output-generated interval gets $O(mn^24^n\log W)$ credits. (We
 already multiplied in the blowup factor.) The total contribution of
 output-generated intervals to the overall runtime is thus 
 $O(mn^24^nZ\log W)$.
 \item Each input-generated interval with pattern of size $s \in \{0,w\}$ gets
 $O(mn^24^n\log W)$ credits. Each input-generated interval with pattern of size
 $s \in [w-1]$ receives 
 \[ O\left( (m4^r|\cert|)^{w-s} \cdot mn^24^n\log W \right) \]
 credits.
 Since there were at most $m4^r|\cert|$ input-generated intervals, the total
 number of credits infused is at most
 %\begin{multline*}
 $$
     O\left( m4^r|\cert| \cdot (m4^r|\cert|)^{w-1} \cdot mn^24^n\log W \right)\\
    = O\left( m^{w+1}n^24^{n+rw} |\cert|^w \log W \right).
 $$
 %\end{multline*}
 \item Lastly, we account for the backtracking intervals whose patterns are not
 self-shadowed. The number of such intervals with a size-$s$ pattern can 
 be upperbounded by $n\binom w s (m4^r|\cert|)^{s+1}$, 
 because every backtracking 
 interval must have come from a pattern (of size $s+1$) and each equality 
 component of the
 pattern can be attributed to an input-generated interval.
 Each such interval, as analyzed above, gets
 $O\left( (m4^r|\cert|)^{w-s} \cdot mn^24^n\log W \right)$ credits.
 Hence, overall we need
 %\begin{multline*}
 $$
     O\left( \sum_{s=1}^{w-1} n\binom w s (m4^r|\cert|)^{s+1} \cdot 
 (m4^r|\cert|)^{w-s} \cdot mn^24^n\log W \right)\\
 = O\left( n2^w (m4^r|\cert|)^{w+1} mn^24^n \log W \right).
 $$
 %\end{multline*}
\ei

Overall, over-estimating by a lot, we need to pump in 
\begin{eqnarray*}
&&mn^24^nZ\log W + m^{w+1}n^24^{n+rw} |\cert|^w \log W
     + n2^w (m4^r|\cert|)^{w+1} mn^24^n \log W\\
&\leq& 2mn^24^n( nm^{w+1} 8^{r(w+1)} |\cert|^{w+1} + Z) \log W\\
&=& O\left( m^3n^34^n \left( nm^{w+1} 8^{r(w+1)} |\cert|^{w+1} + Z\right) 
     \log N \right).
\end{eqnarray*}
To get the last inequality, we bound $W$ -- the number of intervals ever
inserted in to the $\cds$ -- as follows. $W$ is at most the number of
input-generated intervals plus the number of output generated intervals plus the
number of backtracking intervals:
\[ W \leq m4^r|\cert| + N^m + n2^w |\cert|^w \leq 3mn4^nN^m. \]
\ep

% ------------------------------------------------------------------------
\subsection{Proof of Proposition~\ref{prop:tw-eth}}
\label{app:subsec:prop:tw-eth}

\bp
We prove this result by using the reduction from {\sc unique-clique}
to the natural join evaluation problem. 
The {\sc unique-$k$-clique}
input instance ensures that the output size is at most $1$. Let's say the input graph is $G=(V,E)$ (with no self loops),
which is guaranteed to either have no clique or exactly one clique. Then
consider the following query
\[Q_k = \left(\Join_{i\neq j\in [k]} R_{i,j}(v_i,v_j)\right)\Join U(v_1,\dots,v_k),\]
where the domain of each $v_i$ for $i\in [k]$ is $V$, 
\[ R_{i,j} = \bigcup_{(u,v)\in E}\{(u,v),(v,u)\}, \] 
and $U=V^k$.
Note that $Q_k$ is empty if $G$ has no clique otherwise $Q_k$ has exactly $k!$ tuples (corresponding to each of the $k!$ assignments of the vertices of the clique in $G$ to $v_1,\dots,v_k$).
Further, $Q_k$ is $\alpha$-acyclic because of the presence of $U$. Finally, it is easy to verify that $Q_k$ has treewidth $k-1$.
%\ar{Do we need to argue the claim on treewidth?}

Next, we argue that the certificate size of the query above is $O(|E|)$. To see this, first consider the sub-query
\[Q'_k=\Join_{i\neq j\in [k]} R_{i,j}(v_i,v_j).\]
Since all the relations are of size $2|E|$, Proposition \ref{prop:optimal-certificate-upperbound} implies that
the optimal certificate for $Q'_k$ has size $O(|E|)$. We claim that such a certificate can be extended to a certificate for $Q_k$
also of size $O(|E|)$. Indeed, the certificate for $Q'_k$ is enough to pinpoint 
which tuples are the $k!$ output tuples with $O(|E|)$ comparisons
(or certify that the join is empty). Then $O(k\log(|E|))$ more comparisons can verify whether each of the $k!$ tuples is in $U$ or not.
This will produce an $O(|E|+k\cdot k!\cdot\log(|E|))=O_k(|E|)$ sized certificate for $Q_k$.\footnote{
More formally, let $\mv t\in Q'_k$ and define $\mv t_i=\pi_{R_{i,i+1\mod{k}+1}}(\mv t)$. Further, let $\mv t_i=(R_{i,i+1\mod(k)+1}[x_i],R_{i,i+1\mod(k)+1}[x_i,y_i])$ for $i\in [k]$. Then to ``pinpoint" whether $\mv t\in U$, we run the following $k$ binary searches: for $i\in [k-1]$, perform binary search to compute a $z_i$ such that $R_{i,i+1\mod(k)+1}[x_i]=U[z_1,\dots,z_{i-1},z_i]$. Then perform the binary search to compute $z_k$ such that $R_{k,1}[x_k,y_k]=U[z_1,\dots,z_k]$. If any of the $z_i$'s do not exist, then $\mv t$ is not in $Q_k$ otherwise it does. It is easy to check that the above set of comparisons (along with the comparisons in the certificate for $Q'_k$) constitute a valid certificate for $Q_k$ (in the sense of Definition~\ref{defn:certificate}). Finally, note that each binary search can be done with $O(\log(|E|))$ comparisons, which implies the claimed certificate size of $Q_k$.
}
Thus, if there were an $|\cert|^{o(k)}$ algorithm for $Q_k$, it would determine if $G$ has a clique or not in time $\tilde{O_k}(|E|^{o(k)})=\tilde{O_k}(|V|^{o(k)})$.

However,  Chen et al. \cite{DBLP:conf/soda/ChenLSZ07} showed that
if there was an $O(|V|^{o(k)})$-time algorithm solving {\sc unique-$k$-clique},
then the exponential time hypothesis is wrong, and many $\mathsf{NP}$-complete
problems have sub-exponential running time. This implies that for large enough $k$, the above $|\cert|^{o(k)}$ time algorithm will be 
a contradiction.
\ep

% ------------------------------------------------------------------------
\subsection{Proof of Proposition~\ref{prop:ms-tw-lb}}
\label{app:subsec:prop:ms-tw-lb}

\bp
We will in fact prove this result using the same query family as in Proposition~\ref{prop:tw-eth}. In particular, we define
\[Q_w = \left(\Join_{i\neq j\in [w+1]} R_{i,j}(v_i,v_j)\right)\Join U(v_1,\dots,v_{w+1}).\]
As was observed in the proof of Proposition~\ref{prop:tw-eth}, $Q_w$ is both $\alpha$-acyclic and has treewidth $w$.

W.l.o.g. assume that the global attribute order is $v_1,\dots,v_{w+1}$. Now consider the following input instance:
\[ U = [m]^{w+1},\]
\[R_{i,j}=[m]\times [m]\text{ for every } (i,j)\in [w]\times[w],\]
\[R_{i,w+1}=[m]\times \{1\}\text{ for every } i\in [w-1],\]
and
\[R_{w,w+1}=[m]\times \{2\}.\]
It is easy to check that the output of $Q_w$ on the input above is empty and that $|\cert|\le O(wm)$. To see why the latter is true note that with $m-1$ equalities one can certify $\pi_{v_{w+1}}(R_{i,w+1})$ for every $i\in [w]$. Further with $w-1$ further equalities and one comparison one can certify that the output is empty. Thus, we need overall $O(wm)$ comparisons. To complete the proof, we will show that \ms\ on the input above runs in time $\Omega(m^w)$, which would prove the result (since we are ignoring the query complexity). In fact, we will prove this claim by showing that Line~\ref{line:initial-call-tw} in Algorithm~\ref{algo:getpp-tw} is executed $\Omega(m^w)$ times.

For simplicity, we will assume that \ms\ always has the interval $(-\infty,0]$ inserted in all branches of its CDS.

We will argue \ms\ has to consider all possible prefixes of size $w$:
$(t_1,\dots,t_w)\in [m]^w$. In particular, for each such prefix
Algorithm~\ref{algo:getpp-tw} executes Line~\ref{line:initial-call-tw}. One can
show  (e.g. by induction) that for any such prefix $(t_1,\dots,t_w)$, the only
constraints in the CDS that can rule them out are of the form $\langle
*,*,\cdots,*,t_i,*,\cdots,*,(1,\infty)\rangle$ for $i\in [w-1]$ and $\langle
*,\cdots,*,t_w,(0,2)\rangle$. However, this implies that to rule this prefix
(specifically the potential tuple $(t_1,\dots,t_w,1)$) out,
Algorithm~\ref{algo:getpp-tw} has to ``merge" at least two of these constraints,
which means that Line~\ref{line:initial-call-tw} has to be executed at least
once\footnote{Note that all the constraints listed above might not exist in
which case Line~\ref{line:initial-call-tw} might not be able to rule the tuple $(t_1,\dots,t_w,1)$ out but that is fine since the outer algorithm will rule this tuple out.}, as desired.
\ep

% ------------------------------------------------------------------------
\section{end-to-end results for the set intersection query}
\label{app:sec:intersection}
% ------------------------------------------------------------------------

This section describes our results specialized to intersection
queries, which have been discussed by previous work. The purpose is
for the interested reader to both see all the tools used in this simple
example and be able to more directly compare our results with previous
results on set intersection. Also, this query and the bowtie query in the
next section are both $\beta$-acyclic (with any GAO); these two sections
illustrate many of the key ideas in our outer algorithm analysis, the design and
analysis of \cds and $\getpp$, without too much abstraction.

% ------------------------------------------------------------------------
\subsection{The set intersection query}
\label{app:subsec:ms-set-intersection}

\bdefn[Set intersection query]
The {\em set intersection query} $Q_\cap$ is defined as
\[ Q_\cap = S_1(A) \Join S_2(A) \Join \cdots \Join S_m(A). \]
In this query, each input relation $S_i$ is unary over the same attribute $A$.
So each input relation $S_i$ can be viewed as a set of (distinct)
values over the domain $\mv D(A)$.
The output $Q_\cap$ is simply the intersection of all input relations $S_i$.
In this case, $\atoms(Q_\cap) = \{S_1, \dots, S_m\}$,
$\bar A = \bar A(S_i) = \{A\}$, for all $i\in [m]$, and an output ``tuple''
is a one-dimensional vector of the form 
$\mv t = (t)$, where $t \in S_1 \cap \cdots \cap S_m$.
\edefn

In this section, we present $\ms$ specialized to the intersection query
$Q_\cap$. To recap, consider the following problem.
We want to compute the intersection of $m$ sets $S_1, \cdots, S_m$.
Let $n_i = |S_i|$. We assume that the sets are sorted, i.e.
\[ S_i[1] < S_i[2] < \cdots < S_i[n_i], \forall i \in [m]. \]
The set elements belong to the same domain $\mv D$, which is a totally ordered
domain. Without loss of generality, we will assume that $\mv D = \mathbb N$.

In ``practice'' it might be convenient to think of $\mv D$ as the 
{\em index set} to another data
structure that stores the real domain values.
For example, suppose the domain values are \texttt{string}s and there are 
only $3$ strings \texttt{this}, \texttt{is}, \texttt{interesting} in the domain.
Then, we can assume that those strings are stored in a $3$-element array,
and the value $a \in \mv D$ is one of the three indices $0, 1, 2$ into the
array.

% ------------------------------------------------------------------------
\subsection{The \cds for $Q_{\cap}$}
% ------------------------------------------------------------------------

The \cds for $Q_\cap$ is a data structure that stores a collection of 
{\em open} intervals of the form $(a, b)$, where $a$ and $b$ are in the set 
$\mathbb N \cup \{-\infty, +\infty\}$.
When two intervals overlap they are automatically merged.
We overload notation and refer to both the data structure and the set of
intervals stored as \cds.
The data structure supports two operations: $\insconst$ and $\getpp$.
\bi
 \item The $\insconst$ operation takes an open interval and inserts it into 
 the \cds.
 \item The $\getpp$ operation either returns an integer $t$ that does not belong
to any stored interval, or returns $\NULL$ if no such $t$ exists.
 \item If $\cds$ is empty, then $\cds.\getpp()$ returns an arbitrary integer.
     We use $-1$ as the default.
\ei

\paragraph*{Two options for implementing the constraint data structure $\cds$}

If we implement the data structure $\cds$ straightforwardly, then we can do the
following.
We give each input interval one credit, $1/2$ to each end of the interval.
When two intervals are merged, say $(a_1, b_1)$ is merged with $(a_2,
b_2)$ to become $(a_1, b_2)$, we use $1/2$ credit from $b_1$ and $1/2$
credit from $a_2$ to pay for the merge operation.
If an interval is contained in another interval, only the larger interval is
retained in the data structure.
By maintaining the intervals in sorted order,
in $O(1)$-time the data structure can either return a probe point
$t$ that does not belong to any stored interval, or correctly report (return
$\NULL$) that no such $t$ exists.
In other words, each call to $\getpp$ takes amortized constant time.
Inserting a new interval into $\cds$ takes $O(\log W)$-amortized time where 
$W$ is the maximum number of intervals ever inserted into $\cds$, using the
credit scheme described above.

On the other hand, if we apply the strategy of always returning the least value
of $t$ that does not belong to any stored interval, then it is easy to see that
$\cds$ essentially only needs to maintain {\em one} interval $(-\infty, t)$.
Initially when $\cds$ is empty $t=-1$ is returned.
After that -- referring forward to the outer algorithm presented in the next
section -- the newly inserted intervals always contain $t$ and the
new single interval maintained in $\cds$ becomes $(-\infty, t')$ for some 
$t'>t$. Insertion of a new interval only takes constant time because we only
need to compare the high-end of the interval with the current $t$ value.
In this case, the algorithm becomes the minimum-comparison method
in~\cite{DBLP:conf/soda/DemaineLM00} and it is the same as a typical $m$-way
merge join algorithm.

% ------------------------------------------------------------------------
\subsection{The outer algorithm for intersection}

The outer algorithm for \ms specialized to evaluate $Q_\cap$
is presented in Algorithm~\ref{algo:set-intersection}.  In this case, each {\em
  constraint} is an open interval of the form $(a, b)$, where $a$ and
$b$ are {\em integers}.  An interval $(a, b)$ is inserted into the
constraint data structure $\cds$ if the algorithm has determined that
the interval $(a, b)$ contains {\bf no} output.  Note that $a$ and/or
$b$ themselves might be part of the output, but any value in between
is not.  In particular, the constraint data structure $\cds$ stores a
set of constraints and thus we will use the term {\em constraint set}
to refer to the set of intervals stored in $\cds$. %We will return to
%this example in Section~\ref{sec:cds:cap}.

\begin{algorithm*}[!htp]
\caption{$\ms$ for computing the intersection of $m$ sets}
\label{algo:set-intersection}
\begin{algorithmic}[1]
\Require{$m$ sorted sets $S_1, \cdots, S_m$, where $|S_i| = n_i$, $i\in [m]$}
\Require{Elements of $S_i$ are $S_i[1], \cdots, S_i[n_i]$}
\Require{Implicitly $S_i[0] = -\infty$, $S_i[n_i+1] = +\infty$, following the
conventions stipulated in \eqref{eqn:-infty-convention} and
\eqref{eqn:+infty-convention}}
%\Ensure{The run-time is $\tilde O(|\text{optimal certificate}|)$}
\State Initialize the constraint data structure $\cds \leftarrow \emptyset$
\While{$( (t \la \cds.\getpp()) \neq \NULL)$}
\For{ $i = 1,\dots,m$}
    \State $x^h_i \la \min \{ j \suchthat S_i[j] \geq t\}$ \label{line:e_i}
    \State $x^\ell_i \la \max \{ j \suchthat S_i[j] \leq t\}$
    \label{line:e'_i} \Comment{It is possible that $x^h_i =x^\ell_i$}
\EndFor
\If {$S_i[x^h_i] = t$ for all $i\in [m]$} \Comment{Then all $S_i[x^\ell_i]=t$ too}
  \State \textbf{Output} $t$
  \State $\cds.\insconst(t-1, t+1)$ \label{line:Qcap-ruleoutoutput}
\Else
  \For {each $i\in [m]$ such that $S_i[x^h_i]>t$} \Comment{$S_i[x^\ell_i] < t$
  and $x^\ell_i=x^h_i-1$ for such index $i$}
    \State $\cds.\insconst(S_i[x^\ell_i], S_i[x^h_i])$
  \EndFor
\EndIf
\EndWhile
\end{algorithmic}
\end{algorithm*}

We next run through the elements of the argument and the algorithm
for the case of $Q_{\cap}$.

\subsection{Analysis}
\label{app:subsec:set-intersection-certificate}
% ------------------------------------------------------------------------

We specialize notions of argument and certificate to $Q_{\cap}$ in
order to illustrate these concepts.

\begin{defn}[Argument for $Q_\cap$]
An \emph{argument} is a finite set of symbolic equalities and
inequalities, or \emph{comparisons}, of the following forms: (1)
$(S_s[i] < S_t[j])$ or (2) $S_s[i] = S_t[j]$ for $i,j \geq 1$ and $s,t
\in [m]$. An instance satisfies an argument if all the comparisons in
the argument hold for that instance.
\end{defn}

\begin{defn}[Certificate for $Q_\cap$]
An argument $\calA$ is called a {\em certificate} if any collection of 
input sets $S_1, \dots,S_m$ satisfying $\calA$ must have the ``same'' output,
in the following sense.
Let $R_1, \dots, R_m$ be an arbitrary set of unary relations such that 
$|R_j| = |S_j|$, for all $j\in [m]$, and
that $R_1, \dots, R_m$ satisfy all comparisons in the certificate $\calA$, then
the following must hold: 
$$S_1[i_1] = S_2[i_2] = \cdots = S_m[i_m]$$
\text{ if and only if }
$$R_1[i_1] = R_2[i_2] = \cdots = R_m[i_m].$$
The tuple $(i_1,\dots,i_m)$ is called a {\em witness} for this instance of
the query. Another way to state the definition is that, an argument is a
certificate iff all instances satisfying the argument must have the same set of
witnesses.
\end{defn}

The {\em size} of a certificate is the number of comparisons in it.
The {\em optimal certificate} for an input instance is the smallest-size
certificate that the instance satisfies.
The optimal certificate size measures the information-theoretic lowerbound on 
the number of comparisons that any comparison-based join algorithm has to 
discover. Hence, if there was an algorithm that runs in linear time in the 
optimal certificate size, then that algorithm would be instance-optimal.

The following theorem along with Proposition~\ref{prop:Omega(|C|)}
imply that Algorithm \ref{algo:set-intersection} has a near
instance-optimal run-time. 
up to an $m\log N$ factor. 
Since $m$ is part of the query size and the output has to be reported, 
$\ms$ is instance-optimal in terms of data complexity up to a $\log$
factor for this query.

\bthm[\ms is near instance optimal for $Q_\cap$]
Algorithm \ref{algo:set-intersection} runs in time
$O((|\cert|+Z)m\log N)$, where $\cert$ is {\bf any} certificate for the instance,
$N = \sum_{i=1}^m n_i$ is the input size, and $Z$ is the output size.
\label{thm:ms-io-set-intersection}
\ethm
\bp
We show that the number of iterations of Algorithm
\ref{algo:set-intersection} is $O(|\cert|+Z)$, and that each iteration takes
time $O(m\log N)$.

To upperbound the number of iterations, the key idea is to ``charge'' each 
iteration of the main while loop
to either a distinct output value {\em or} a pair of comparisons in the 
certificate $\cert$ such that no comparison will
ever be charged more than a constant number of times.
Each iteration in the loop is represented by a distinct probe value $t$.
Hence, we will find an output value or a pair of comparisons to ``pay'' 
for $t$ instead of paying for the iteration itself.

Let $t$ be a probe value in an arbitrary
iteration of Algorithm \ref{algo:set-intersection}.
Let $x^h_i$ and $x^\ell_i$ be defined as in lines
\ref{line:e_i} and \ref{line:e'_i} of the algorithm.

First, consider the case when
$S_i[x^h_i] = t$ for all $i\in [m]$, i.e. $t$ is an output value.
Note that in this case $x^\ell_i = x^h_i$ for all $i\in [m]$.
We pay for $t$ by charging the output $t$. 
The new constraint inserted in line~\ref{line:Qcap-ruleoutoutput} ensures that
we will never charge an output twice.

Second, suppose $S_i[x^h_i]>t$ for some $i$, i.e. $t$ is not an output.
(Note that by definition it follows that $S_i[x^\ell_i] < t$.)
For each $i\in [m]$, the variable $S_i[x^h_i]$ is said to be {\em $t$-alignable} 
if $S_i[x^h_i]$ is already equal to $t$ (in which case $x^\ell_i=x^h_i$)
or if $S_i[x^h_i]$ is not part of 
{\em any} comparison ($=,<,>$) in the certificate $\cert$.
Similarly, we define the notion of {\em $t$-alignability} for
$S_i[x^\ell_i]$, $i\in [m]$.

When $S_i[x^h_i]$ is $t$-alignable, setting $S_i[x^h_i] = t$ will not violate 
any of the comparisons in the certificate $\cert$.
Similarly, we can transform the input instance to another input instance
satisfying $\cert$ by setting $S_i[x^\ell_i] = t$, provided $S_i[x^\ell_i]$ is 
$t$-alignable.

{\bf Claim:} if $t$ is not an output, then there must exist some 
$\bar i\in [m]$ for which both 
$S_{\bar i}[x^\ell_{\bar i}]$ and 
$S_{\bar i}[x^h_{\bar i}]$ are {\em not} $t$-alignable.
In particular, in that case
\[ S_{\bar i}[x^\ell_{\bar i}] < t < S_{\bar i}[x^h_{\bar i}] \]
and both 
$S_{\bar i}[x^\ell_{\bar i}]$ and 
$S_{\bar i}[x^h_{\bar i}]$ are involved in comparisons in the certificate.

Before proving the claim, let us assume it is true and finish off the charging
argument. We will pay for $t$ using any comparison involving
$S_{\bar i}[x^h_{\bar i}]$ and any comparison involving 
$S_{\bar i}[x^\ell_{\bar i}]$.
Because they are not $t$-alignable, each of them must be part of some comparison
in $\cert$.
Since we added the interval $(S_i[x^\ell_i], S_i[x^h_i])$ to the constraint 
data structure $\cds$, in later iterations $t$ will never hit the same 
interval again.
Each comparison involving one variable will be charged at most $3$ times: one
from below the variable, one from the above the variable, 
and perhaps one when the variable is output.

\noindent
{\bf Proof of claim.}
Suppose to the contrary that at least one member of every pair 
$S_i[x^\ell_i]$ and $S_i[x^h_i]$, $i\in [m]$, {\em is} $t$-alignable. 
Let $v(i) \in \{\ell, h\}$ such that $S_i[x^{v(i)}_i]$ is $t$-alignable,
$i\in [m]$.
Let $i_0$ be such that $S_i[x^{v(i_0)}_i] \neq t$. The value $i_0$ must exist
because $t$ is not an output.
First, by assigning 
$S_i[x^{v(i)}_i] = t$ for all $i$, we obtain an instance
satisfying the certificate for which
\begin{equation}
S_1[x^{v(1)}_1] = 
S_2[x^{v(2)}_2] =  \cdots
S_m[x^{v(m)}_m].
\label{eqn:intersection-witness}
\end{equation}
Second, by assigning 
$S_i[x^{v(i)}_i] = t$ for all $i \neq i_0$, we obtain an instance
also satisfying the certificate for which
\eqref{eqn:intersection-witness} {\em does not} hold!
This contradicts the certificate definition; hence, the claim holds.

We have already discussed how the constraint data structure \cds can be 
implemented so that insertion takes amortized constant time in the number 
of intervals inserted, and querying (for a new probe point $t$) takes 
constant time. 
Given a probe point $t$, searching for the values $x^h_i$ and $x^\ell_i$ takes
$O(\log N)$-time, for each $i\in [m]$.
Hence, each iteration of the algorithm takes time at most $O(m\log N)$. 
\ep

\brmk
In fact, if we implement $\ms$ using the galloping/leapfrogging strategy shown 
in~\cite{DBLP:conf/soda/DemaineLM00} and~\cite{DBLP:journals/corr/abs-1210-0481}, 
then we can speed up the search for the values $x^h_i$ and $x^\ell_i$ of Algorithm~\ref{algo:outer-algorithm}
slightly in terms of asymptotic runtime.
Those ideas in fact work very well in practice! However, they 
are regarded as ``implementation details'' in this paper and will 
not be discussed further. We are happy with a $\log$-factor loss.
\ermk

% ------------------------------------------------------------------------
\section{End-to-end results for the bowtie query}
\label{app:sec:bowtie}
% ------------------------------------------------------------------------

To illustrate the key ideas of the analysis of \ms, we present in this section
the second end-to-end set of results on a query that is slightly more complex
than the intersection query from Section~\ref{app:sec:intersection}. The hope
is, without burdening the reader with the heavy notation from the general
algorithm, the so-called {\em bowtie query} is able to illustrate many key ideas. 
We will define what the query is, what are arguments and certificates for this
query, what are the constraints and the \cds, the outer algorithm, and finally
the analysis. It turns out that the \cds for this query is very simple.
Additionally, the bowtie query {\em is} $\beta$-acyclic, and {\em any} GAO is a
nested elimination order!

% ------------------------------------------------------------------------
\subsection{The bowtie query, arguments, and certificates}

\bdefn[Bow-tie query]
\label{def:bowtie}
The {\em bow-tie query} is defined as
\[ Q_{\bowtie} = R(X) \Join S(X,Y) \Join T(Y). \] 
In this case, $\atoms(Q_{\bowtie}) = \{R, S, T\}$,
$\bar A = (X, Y)$, and a tuple $\mv t = (x, y)$ is in the output if
and only if $x \in R$, $(x,y) \in S$, and $y \in T$.
\edefn

Due to symmetry, the global attribute order (GAO), without loss of generality,
can be assumed to be $(X, Y)$. The relations $R$, $S$, and $T$ are assumed to be
already indexed, allowing for the following kind of access.
\bi
 \item $R[*]$ is the set of all values in $R$.
 \item $R[i]$ is the $i$th smallest value in $R$, where $i$ is the index and the
       value $R[i]$ belongs to the domain $\mv D(X)$ of attribute $X$.
 \item Similarly, $T[*] = T$, and $T[j] \in \mv D(Y)$ is the $j$th value in $T$.
 \item $S[*]$ is the set of all $X$-values in $S$
 \item $S[i]$ is the $i$th smallest $X$-value in $S$
 \item $S[i, *]$ is the set of all $Y$-values among tuples $(x,y) \in S$ with $x
 = S[i]$.
 \item $S[i,j]$ is the $j$th $Y$-value among all tuples $(x,y) \in S$ with
       $x = S[i]$.
\ei
In the above, when we say $i$th smallest value we use the set semantic. There is
no duplicate value and thus no need to break ties.
We next specialize notions of argument and certificate to this particular 
query in order to illustrate these concepts.

\bdefn[Argument for $Q_{\bowtie}$]
An {\em argument} for the bow-tie query $Q_{\bowtie} = R(X) \Join S(X, Y) \Join
T(Y)$ is a set of comparisons in one of the following three formats:
\begin{eqnarray*}
R[i_r] & \theta & S[i_s], \ \ \ \ \ \text{  (a comparison on $X$-value)} \\
S[i_s,j_s] & \theta & T[j_t],\ \ \ \ \text{  (a comparison on $Y$-value)} \\
S[i_s,j_s] & \theta & S[i'_s, j'_s]. \text{  (a comparison on $Y$-value between
$S$-tuples)}
\end{eqnarray*}
where $\theta \in \{<, =, >\}$ is called a {\em comparison}.
\edefn

Since the $X$-values in $R$ and $Y$-values in $T$ are distinct, there was no
need to allow for comparisons between tuples in $R$ or between tuples in $T$.
Allowing for such comparisons does not change our analysis.

\bdefn[Certificate for $Q_{\bowtie}$]
For the bow-tie query, an {\em argument} $\calA$ is called a {\em certificate} if
the following conditions hold.
Let $R'(X), S'(X,Y), T'(Y)$ be three arbitrary relations such that
\begin{eqnarray*}
|R'| &=& |R|\\
|T'| &=& |T|\\
|S'[*]| &=& |S|\\
|S'[i, *]| &=& |S[i,*]|, \ \forall i, 1 \leq i \leq |S[*]|
\end{eqnarray*}
and that the $R',S',T'$ and $R, S, T$ satisfy all the comparisons in $\calA$.
Then, for any triple $\{ i, (j, k), \ell \}$,
 \[ R[i] = S[j] \text{ and } S[j, k] = T[\ell] \]
if and only if
 \[    R'[i] = S'[j] \text{ and } S'[j, k] = T'[\ell]. \]
Such a triple is called a {\em witness} for the instance $R, S, T$; and, it is
also a witness for the instance $R',S',T'$.
\edefn

Following the lead from Example~\ref{ex:same-relation-and-equalities}, it is not
hard to construct an example showing that comparisons between $Y$-variables
between tuples in $S$ are sometimes crucial to reduce the overall certificate
size.

% ------------------------------------------------------------------------
\subsection{Constraints and the \cds}

For the bow-tie query $Q_{\bowtie}$ every constraint is of one of the 
following three forms:
\bi
 \item $\langle (a, b), *\rangle$,
 \item $\langle p, (a, b)\rangle$,
 \item or $\langle*, (a, b)\rangle$.
\ei
where $p \in \mathbb N$, $a, b \in \{-\infty, +\infty\} \cup \mathbb N$.

A tuple $\mv t = (x,y)$ satisfies the constraint 
$\langle (a, b), *\rangle$ if $x \in (a, b)$; it satisfies the constraint
$\langle p, (a, b)\rangle$ if $x=p$ and $y\in (a, b)$;
and it satisfies the constraint
$\langle*, (a, b)\rangle$ if $y \in (a, b)$.

Each constraint can be thought of as an ``interval'' in the following sense.
The first form of constraints consists of all two-dimensional (integer) points
whose $X$-values are between $a$ and $b$. We think of this region as a
2D-interval (a vertical strip).
Similarly, the second form of constraints is a 1D-interval,
and the third form of constraints is a 2D-interval (a horizontal strip).

We store these constraints using a two-level tree data structure (which is a $\ctree$ specialized to the two attribute case).
Figure \ref{fig:constraint-data-structure-bowtie} illustrates the data
structure.
\begin{figure}[th]
\centerline{\includegraphics[width=3in]{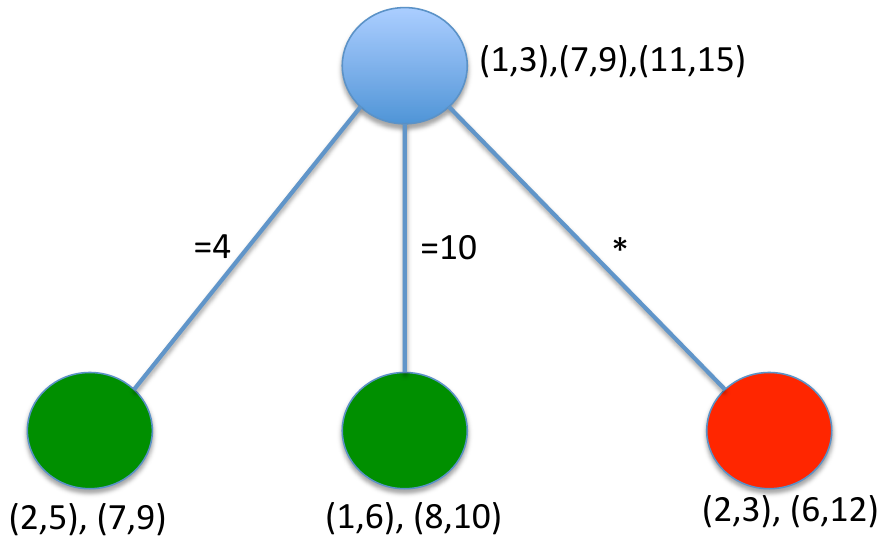}}
\caption{Constraint tree data structure for the bow-tie query}
\label{fig:constraint-data-structure-bowtie}
\end{figure}

In the first level (the root node), there is a collection of intervals
indicating the ruled out $X$-values. Then there are branches to the second
level. Each branch is marked with an $*$, or $=p$ for some integer $p$ that
does not belong to any interval stored at the root.

In the second level of the tree, every node has a collection of open intervals.
Intervals belonging to the same node are merged when they overlap. So, the
collection of intervals at each node are disjoint too.  To analyze the cost of
merging we use a credit based argument. Every inserted interval is given one
credit.  We use the  trick of giving the low end and the high end of an interval
half a credit to pay for the merging of two intervals. If the second level of a
$[=p]$-branch covers the entire domain, then the $[=p]$-branch is turned into a
$\langle (p-1, p+1), *\rangle$ constraint that can further be merged (at the
root level). Effectively, the interval $(p-1,p+1)$ is inserted into the interval
list of the root node. We will give an extra credit to the $[=p]$-branch so that
when the branch is turned into a $\langle (p-1,p+1), *\rangle$ interval both of
the end points has half a credit as any other interval of the root.

Inserting a new constraint takes amortized logarithmic time, as we keep the
branches sorted, and the new constraint might ``consume'' existing intervals.
(This logarithmic factor can be improved to constant time if the \ms algorithm
and the constraint data structure work in concert, but we will not dig deeper
into this detail at this point. We rather keep the description generic, and
separate as much as possible the inner workings of the algorithm from the data 
structure.) In amortized constant time,
the data structure \cds is able to report a new tuple $\mv t=(x,y)$ that does
not satisfy any of its constraints, or correctly report that no such $\mv t$
exists.
To find $\mv t$, we apply the following strategy:
\bi
 \item We first find $x$ such that $x$ does not belong to any root-level
 interval. This value of $x$, if it exists, can easily be found in constant 
 time by taking the right end point of the lowest interval from the interval
 list at the root. (Recall the invariant that intervals are disjoint!). 
 If there is no first level interval, then we can set $x=-1$
 \item If $x$ is found and there is no $[=x]$ branch, then we find a 
 value $y$ that does not belong to any second level interval on the $*$-branch.
 If there is no $*$ branch, then set $y=-1$.
 If no $y$ exists then no such $\mv t$ exists, the algorithm terminates.
 \item If $[=x]$ is a first-level branch, we find a
 value $y$ under the $[=x]$-branch that does not belong to any interval
 under that branch. We call such a $y$ a ``free value'' $y$.
 The tuple $\mv t=(x,y)$ might still violate a $\langle *, (a, b)\rangle$
 constraint in the $*$-branch. In that case, we insert the constraint $\langle x, (a, b)\rangle$ into the tree. Then we find the next smallest ``free'' value $y$
 under the $[=x]$-branch again and continue with the ``ping-pong'' with the
 $*$-branch until a good value of $y$ is found. 
 The intervals under $[=x]$-branch might be merged with an interval taken from
 the $*$-branch, but if we give each constraint $\langle x, (a,b)\rangle$
 a constant number of credits, we can pay for all the merging operations.
\ei

To summarize, insertion and querying for a new probe point into the above data
structure takes at most logarithmic time in the amortized sense (over all
operations performed).

% ------------------------------------------------------------------------
\subsection{The outer algorithm}

\begin{algorithm*}[!tph]
\caption{$\ms$ for evaluating the bow-tie query $R(X) \Join S(X,Y) \Join T(Y)$.}
\label{algo:bowtie-join}
\begin{algorithmic}[1]
\Require{Following the conventions stipulated in \eqref{eqn:-infty-convention}
and \eqref{eqn:+infty-convention}, the following are implicit:}
\Require{$R[0] = S[0] = T[0] = -\infty$} 
\Comment{out-of-range indices}
\Require{$R[|R|+1] = T[|T|+1] = S[|S[*]|+1] = +\infty$}
\Comment{out-of-range indices}
\Require{$S[i, 0] = -\infty$, $S[i, |S[i, *]|+1] = +\infty$, $\forall
1\leq i \leq |S[*]|$}
\Comment{out-of-range indices}
\While{$( (\mv t \la \cds.\getpp()) \neq \NULL)$}
  \State Say $\mv t = (x,y)$
  \State $(i^{\ell}_R, i^h_R) \la R.\findgap( (), x)$\label{line:i^ell_R}
  \State $(i^{\ell}_T, i^h_T) \la T.\findgap( (), y)$\label{line:i^ell_T}
  \State $(i^{\ell}_S, i^h_S) \la S.\findgap( (), x)$\label{line:i^ell_S}

  \State $(i^{\ell\ell}_S, i^{\ell h}_S) \la S.\findgap( (i^{\ell}_S),y)$
  \State $(i^{h\ell}_S, i^{hh}_S) \la S.\findgap( (i^{h}_S),y)$  
\label{line:i^{hh}_S}

  \If {($R[i^h_R] = S[i^h_S] = x$ {\bf and} $S[i^h_S,i^{hh}_S] = T[i^h_T] = y$)}
  \label{line:if} \Comment{Not true if any index is out of range}
    \State \textbf{Output} the tuple $\mv t = (x,y)$ \label{line:output-t}
    \State $\cds.\insconst\left(\langle x, (y-1, y+1) \rangle\right)$ \label{line:constraint-tt}
  \Else \label{line:else}
    \State $\cds.\insconst\left(\left\langle \left(R[i^\ell_R], R[i^h_R]\right),
    * \right\rangle\right)$ \Comment{Interval on $X$}
    \State $\cds.\insconst\left(\left\langle \left(S[i^\ell_S], S[i^h_S]\right),
    * \right\rangle\right)$ \Comment{Interval on $X$}
    \State $\cds.\insconst\left(\left\langle *, \left(T[i^\ell_T],
    T[i^h_T]\right) \right\rangle\right)$; \Comment{Interval on $Y$}
      \label{line:c1}
    \If {($i^h_S$ is not out of range)}
       \State $\cds.\insconst\left(\left\langle S[i^h_S], \left(S[i^h_S,i^{h\ell}_S], S[i^h_S, i^{hh}_S]\right) \right\rangle\right)$
      \label{line:c2h}
    \EndIf
    \If {($i^\ell_S$ is not out of range)}
      \State $\cds.\insconst\left(\left\langle S[i^\ell_S], \left(S[i^\ell_S,i^{\ell\ell}_S], S[i^\ell_S, i^{\ell h}_S]\right) \right\rangle\right)$
      \label{line:c2ell}
    \EndIf
  \EndIf
\EndWhile
\end{algorithmic}
\end{algorithm*}

We next describe the outer algorithm of \ms specialized to the bowtie query.
The key of Algorithm~\ref{algo:bowtie-join} is the loop. 
We begin with a point $\mv t=(x,y)$ that has not been ruled out by the 
constraint data structure at this point of the algorithm. 
Intuitively, our goal is to determine if $\mv t$ is
in the output; if $\mv t$ it is not in the output, then our goal is to
find some gap that could rule out $\mv t$. 
We find the gaps by probing each of the relations $R$, $S$, and $T$ 
``around'' the point $\mv t$. 

The probes of $R$ and $T$ are straightforward: we find gaps around the value 
of $x$ and $y$ in each relation. 
With respect to $S$, if we only needed to verify that $\mv t$ is not an output
tuple then a gap around the value $(x,y)$ in $S$ would suffice. 
However, we do a bit more work in Lines~\ref{line:i^ell_S}
to~\ref{line:i^{hh}_S}; the reason for that is explained below. 
Before then, observe that if $\mv t$ is in the
output then $S[i_{S}^{h},i_{S}^{hh}] = (x,y)$ (in fact,
$S[i_{S}^{z},i_{S}^{zz'}]=(x,y)$ for all $z,z' \in \{l,h\}$). Thus,
the condition in Line~\ref{line:if} is a sound and complete check for
$\mv t$ to be in the output.

\begin{figure*}
\centering
\begin{tabular}{c@{\hskip 1in}c}
\includegraphics[height=2in]{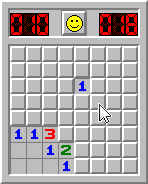} &
\includegraphics[height=2in]{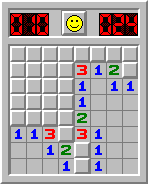}\\
(a) & (b)
\end{tabular}
\caption{In the classic game Minesweeper, when you click on a square
  without a bomb (a), it can reveal many other squares (b), which is
  analogous to our gap exploration. See
  \url{http://en.wikipedia.org/wiki/Minesweeper_(video_game)} for a
  history of this classic game.}
\label{fig:winmine}
\end{figure*}
As we noted, \ms does more work than is necessary to certify that 
$\mv t$ is not an output tuple. This is because
\ms is not only searching for a gap that contains $\mv t$ to rule $\mv t$
out, but it also ``looks for'' a variable that is involved in {\em any}
certificate (including the optimal certificate); this is done so that the
algorithm does not explore too deep into --
and thus spend too much time in -- regions that the optimal certificate wastes
very few comparisons to rule out.
At the same time, we also have to ensure that we don't explore too many gaps
just to capture one single comparison in the optimal certificate.
\ms steps on this fine line roughly as follows.

Since $\mv t$ could have been excluded by a gap that does not have
$(x,y)$ on its boundary, \ms -- analogously to the
eponymous windows game when one chooses a square without a bomb (see
Figure~\ref{fig:winmine}) -- finds the biggest gaps around the point $\mv
t$. In this case, we observe that for each $\mv t$, \ms examines a
constant number of gaps. From Lines~\ref{line:else} onward, we simply
insert all the gaps that we found above. Later, we will reason that
{\em some} comparison in any optimal certificate is found by at least
one of the $\findgap$ searches. In this example, we explored a
constant number of gaps (here $5$) for every $\mv t$; for more complex
queries, the number of gaps may grow, but it will grow with the number
of attributes in $Q$, i.e., the number of gaps we explore does {\em
 not} depend on the data.

Algorithm~\ref{algo:bowtie-join} presents the formal details of the algorithm.
Notations are redefined here for completeness.
Let $|R|, |S|, |T|$ denote the number of tuples in the corresponding
relations, $S[*]$ the set of $X$-values in $S$,
$S[i]$ the $i$'th $X$-value,
$S[i,*]$ the set of $y$'s for which $(S[i],y) \in S$,
and $S[i,j]$ the $j$'th $Y$-value in the set $S[i,*]$.

Figure \ref{fig:bowtie-pic} illustrates the choices of various parameters 
in the algorithm.
Some of the constraints might look unintuitive and perhaps redundant at first.
The picture shown in Figure~\ref{fig:bowtie-pic} should give the reader the
correct geometric intuition behind the constraints: all points in $\outspace$
satisfying the constraints are guaranteed to be not part of the output.

\begin{figure*}
\centerline{\includegraphics[width=3.5in]{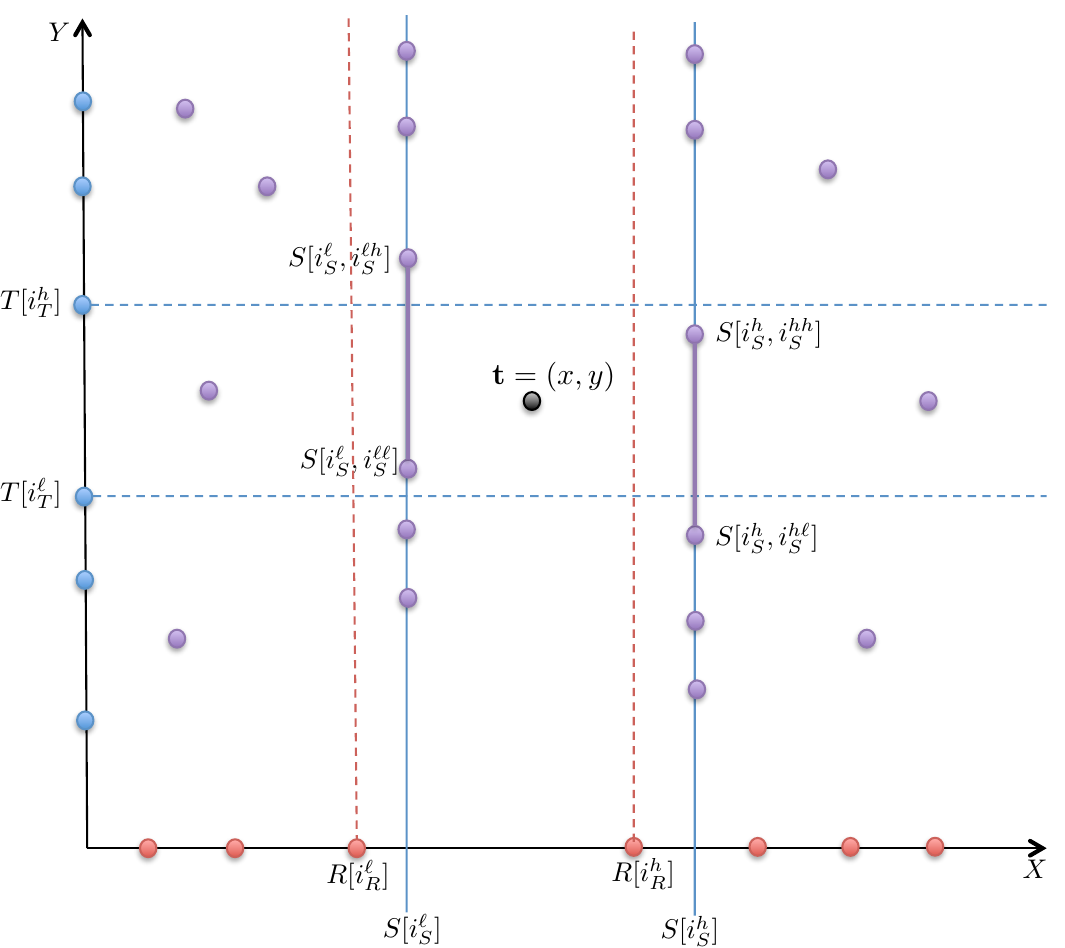}}
\caption{Illustration for Algorithm \ref{algo:bowtie-join}}
\label{fig:bowtie-pic}
\end{figure*}
To convey the subtlety in gap exploration, let us consider a simple idea.
The first candidate gap from $S$ that comes to mind is perhaps
the gap between $(x_-, y_-)$
and $(x_+,y_+)$ where $(x_-, y_-)$ is the largest tuple in $S$ that is 
smaller than $(x,y)$ lexicographically, and
$(x_+,y_+)$ is the smallest tuple in $S$ that is greater than $(x,y)$ 
lexicographically.
The problem with this simple idea is that this gap might actually fail to 
capture {\em any} variable involved in an optimal certificate comparison at all. 
Consider, for example, the following input:
\begin{eqnarray*}
R[X] &=& \{2\}, \\
T[Y] &=& \{N+1\},\\
S[X,Y] &=& \left\{ (1, N+1+i) \suchthat i\in [N] \right\} \cup 
          \left\{ (3, i) \suchthat i \in [N] \right\}.
\end{eqnarray*}
The certificate $\cert = \{ S[1, 1] > T[1], S[2, N] < T[1] \}$ is an optimal
certificate for this instance. Suppose the probe point is
$\mv t = (2, N+1)$, then $(x_-,y_-) = (1, 2N+1)$ and $(x_+, y_+) = (3, 1)$
both of which do not have anything to do with the optimal certificate above.

% ------------------------------------------------------------------------
\subsection{Analysis}

We next show Algorithm~\ref{algo:bowtie-join} 
is near instance optimal (modulo the time spent in the \cds). This Theorem
parallels the content of Theorem~\ref{thm:analyze-outer-algorithm} so the reader
can map back this special case to the general statement.

\bthm[runtime of \ms on $Q_{\bowtie}$]
Let $N$ denote the total number of tuples from the input relations,
$Z$ the total number of output tuples.
Let $\cert$ be an arbitrary certificate for the input query.
Then, the total runtime of Algorithm \ref{algo:bowtie-join} is 
\[O\left( \left( |\cert| + Z \right)  \log(N) + T(\cds) \right),\]
where $T(\cds)$ is the total time it takes the constraint data structure.
The algorithm inserts a total of $O(|\cert| + Z)$ constraints to $\cds$ and
issues $O(|\cert| + Z)$ calls to $\getpp$.
\label{thm:analyze-outer-algorithm-bowtie}
\ethm
\bp%[Proof of Theorem~\ref{thm:analyze-outer-algorithm-bowtie}]
We show that the number of iterations of Algorithm \ref{algo:bowtie-join} is
$O(|\cert|+Z)$.  Since the amount of work done in each iteration is $O(\log N)$
and the number of calls to $\getpp()$ and the number of inserted constraints are
linear in the number of iterations, the proof is complete.

We ``pay'' for each iteration of the algorithm, represented by the tuple 
$\mv t=(x,y)$ that the $\cds$ returns in the while predicate for that 
iteration, by ``charging'' either an output tuple or a pair of comparisons in 
the certificate $\cert$. We show that each output tuple and each comparison will
be charged at most $O(1)$ times.  To this end, we define a couple of terms.

Any variable
$$e \in \left\{ 
R[i^h_R],
R[i^\ell_R],
S[i^h_S],
S[i^\ell_S]
\right\}$$
is said to be {\em $\mv t$-alignable} if either
$e$ is already equal to $x$ or $e$ is not involved in any comparison
in the certificate $\cert$.
By convention, any variable whose index is out of range is {\bf not}
$\mv t$-alignable.
The semantic of $\mv t$-alignability is as follows.
If a $\mv t$-alignable variable $e$ is not already equal to $x$, setting
$e=x$ will transform the input into another database instance satisfying all
comparisons in $\cert$ without violating the relative order in the relation
that $e$ belongs to.

A variable 
$$e \in \left\{ T[i^h_T], T[i^\ell_T] \right\}$$
is said to be {\em $\mv t$-alignable} if either
$e$ is already equal to $y$ or $e$ is not involved in any comparison
in the certificate $\cert$.
A variable $$e \in \left\{ S[i^h_S, i_S^{h\ell}], S[i^h_S, i^{hh}_S] \right\}$$
is $\mv t$-alignable if $S[i^h_S]$ is $\mv t$-alignable
and either $e$ is already equal to $y$ or
$e$ is not part of any comparison in the certificate $\cert$.
Similarly, we define $\mv t$-alignability for a variable
$$e \in \{ S[i^\ell_S, i_S^{\ell\ell}], S[i^\ell_S, i^{\ell h}_S]\}.$$

Next, we describe how to ``pay'' for the tuple $\mv t=(x,y)$.

\paragraph*{Case 1}
Line \ref{line:output-t} is executed. We pay for $\mv t$ by charging the output
tuple $\mv t$.
The constraint added in line \ref{line:constraint-tt} ensures that we won't
have to pay for the same output $\mv t$ again.

\paragraph*{Case 2} The else part (line \ref{line:else}) is executed,
i.e. $\mv t$ is not an output tuple.
We claim that one of the following five cases must hold:

\bi
 \item[(1)] both $R[i^h_R]$ and $R[i^\ell_R]$ are not $\mv t$-alignable,
 \item[(2)] both $S[i^h_S]$ and $S[i^\ell_S]$
 are not $\mv t$-alignable,
 \item[(3)] both $S[i^h_S, i_S^{h\ell}]$ and $S[i^h_S, i^{hh}_S]$
are not $\mv t$-alignable,
 \item[(4)] both $S[i^\ell_S, i_S^{\ell\ell}]$ and $S[i^\ell_S, i^{\ell h}_S]$
are not $\mv t$-alignable,
 \item[(5)] both $T[i^h_T]$ and $T[i^\ell_T]$ 
are not $\mv t$-alignable.
\ei

Suppose otherwise that at least one member in each of the five pairs above is
$\mv t$-alignable. For example, suppose the following variables are
$\mv t$-alignable:
\[ R[i^h_R],
S[i^\ell_S],
S[i^\ell_S, i_S^{\ell\ell}],
T[i^h_T]. 
\]
Then, we construct two database instances as follows.

\bi
\item {\em Database instance $I$.} In this instance, we keep all current 
variable values except that we set $R[i^h_R] = S[i^\ell_S] = x$ and
$S[i^\ell_S, i_S^{\ell\ell}] = T[i^h_T] = y$. Then, clearly in this instance the
set of index tuples 
$\left\{ i^h_R, (i^\ell_S, i_S^{\ell\ell}), i^h_T \right\}$
is a witness for $Q_{\bowtie}(I)$.
\item {\em Database instance $J$.} This instance requires a little bit more
care. Recall that we are in the case when $\mv t$ is {\em not} an output tuple.
Hence, it cannot possibly be the case that 
$R[i^h_R] = S[i^\ell_S] = x$ and
$S[i^\ell_S, i_S^{\ell\ell}] = T[i^h_T] = y$ already. 
Assume, for example, that $S[i^\ell_S, i_S^{\ell\ell}] \neq y$. Then, the
database instance $J$ is constructed by setting
$R[i^h_R] = S[i^\ell_S] = x$ and
$T[i^h_T] = y$. Then, in this case
$\left\{ i^h_R, (i^\ell_S, i_S^{\ell\ell}), i^h_T \right\}$
is {\bf not} a witness for $Q_{\bowtie}(J)$.
\ei
Note that both $I$ and $J$ satisfy $\cert$. Hence, we reach a contradiction
because $\cert$ is a certificate. The claim is thus proved.

The key idea is, by definition 
each non-out-of-range variable $e$ that is not $\mv t$-alignable 
must be involved in a comparison in the certificate $\cert$.
There is an exception, something like $S[i^\ell_S, i_S^{\ell\ell}]$
might be non-$\mv t$-alignable because its prefix variable $S[i^\ell_S]$ is not
alignable. But in that case $S[i^\ell_S]$ must be involved in a comparison in
the certificate, and that's all we need for the reasoning below.

Instead of charging a comparison, we will charge a 
non-out-of-range non-$\mv t$-alignable variable. 
{\bf If} each non-out-of-range non-$\mv t$-alignable variable is charged 
$O(1)$ times, then each comparison will be charged $O(1)$-times.

We pay for $\mv t$ by charging any pair of non-out-of-range 
non-$\mv t$-alignable variables out of the five pairs above that are involved in
comparisons in $\cert$.
We call each of those five pairs an ``interval."
The pairs of the type (1), (2), and (5) are 2D-intervals,
and the pairs of the type (3), (4) are 1D-intervals.

Due to the constraints added on lines \ref{line:c1}, \ref{line:c2h},
and \ref{line:c2ell}, the 2D-intervals are charged at most once.
Since two 2D-intervals might share an end point, each non-out-of-range non-$\mv
t$-alignable variable from a 2D-interval might be charged twice.
The 1D-intervals are charged at most twice.
Each non-out-of-range non-$\mv t$-alignable variable from a 1D-interval might be
charged at most four times. Consequently, each comparison is charged $O(1)$
times.
\ep

% ------------------------------------------------------------------------
\section{Counter examples}
\label{app:sec:counter-examples}
% ------------------------------------------------------------------------

In this section, we present a family of $\beta$-acyclic join queries and
instances on which none of
Leapfrog-Triejoin~\cite{DBLP:journals/corr/abs-1210-0481} (\lb henceforth), the
algorithm of~\cite{DBLP:conf/pods/NgoPRR12} (\nprr henceforth) or Yannakakis'
algorithm~\cite{DBLP:conf/vldb/Yannakakis81} are instance optimal (i.e. they
don't run in $O(|\cert| + Z)$-time). Furthermore, the gap between those
algorithms and \ms can be arbitrarily large.

To simplify the argument, let us first consider a simpler family of instances.
The query is the following.
\[ Q = \ \Join_{i=1}^m R_i(A_i,A_{i+1}). \]
The above query is $\beta$-acyclic, and the GAO $A_1, \dots, A_{m+1}$ is a
nested elimination order. The main idea behind the instance is to ``hide" the
certificate along a long path in the query. All three algorithms \nprr, \lb, and
Yannakakis do not explore the attributes globally as \ms does, and hence they
will get stuck looking for {\em many} partial tuples that do not contribute to
the output.

The relations are constructed as follows. Each attribute $A_i$ will have as its 
domain the set $[mM]$, where $m\geq 5$ and $M$ is a large positive integer.
Each relation $R_i$ will have $m$ ``chunks,'' where the $j$th chunk 
($1\le j\le m$) is a subset of the set 
$$[(j-1)M+1,jM]\times [(j-1)M+1,jM].$$ 
More precisely, relation $R_i$ is defined as follows. 
\bi
 \item For every $j \in [m] - \{i, i-1\}$, the $j$th chunk of $R_i$ is exactly
exactly $$[(j-1)M+2,jM]\times [(j-1)M+2,jM].$$ 
 \item The $i$th chunk consists of one single tuple 
$\bigl((i-1)M+1,(i-1)M+1\bigr)$.
 \item And the $(i-1)$'th chunk is empty.
\ei
If $i=1$, then we interpret $i-1$ as $m$. Namely, the $m$th chunk of $R_1$ is
empty. Note that every relation is of size $N = \Theta(mM^2)$.

It is not hard to see that the output of the above instance is empty.
Furthermore, there is a certificate of size $O(mM)$.
Hence, by Theorem~\ref{thm:io-beta-neo} \ms takes $O(mM\log{M})$ time 
since $A_1,\dots,A_{m+1}$ is a nested elimination order.

This certificate consists of the following comparisons
\begin{eqnarray*}
    R_1[1,1] &<& R_2[1]\\
    R_1[i,1] &>& R_2[1], \text{ for } i > 1\\
%    R_2[1,1] &>& R_3[M], \\
%    R_2[1,1] &<& R_3[M+1], \\
    R_2[i,1] &>& R_3[M+1], \text{ for } i > 1\\
%    R_3[M+1,1] &>& R_4[2M], \\
%    R_3[M+1,1] &<& R_4[2M+1], \\
    R_3[i,1] &>& R_4[2M+1], \text{ for } i > M+1\\
    R_4[i,1] &>& R_5[3M+1], \text{ for } i > 2M+1\\
    \vdots & \vdots & \vdots\\
    R_{m-1}[i,1] &>& R_m[(m-2)M+1], \text{ for } i > (m-3)M+1.
\end{eqnarray*}
To see why this is a certificate that the output is empty, consider an arbitrary
witness for this instance: 
$$X = \left\{(i^{(1)},j^{(1)}), \dots, (i^{(m)}, j^{(m)})\right\}.$$
From the first two (sets of) inequalities above, we know
$i^{(2)} > 1$. Then, from the next inequality we know $i^{(3)}>M+1$.
This inference goes on until the last inequality, which does not leave any room
for $j^{(m)}$. Hence, such a witness cannot exist.

Now, for every $i \in [m]$, the semijoin $R_i \lJoin R_{i+1}$ has size
$\Omega(mM^2)$. Hence, Yannakakis algorithm runs in time at least
$\Omega(mM^2)$.

For \lb and \nprr, it takes slightly more work to be rigorous, but the key ideas
are as follows. Consider {\em any} attribute ordering that \lb adopts. 
Say the attribute ordering is $A_{i_1}, \dots, A_{i_{m+1}}$, for some permutation 
$\{i_1,\dots,i_{m+1}\}$ of $[m+1]$.
\lb will compute the intersection on $A_{i_1}$, and for each value $a$ in the
intersection it will compute the join on $A_{i_2}$ using $a$ as an anchor, 
and so on.
If $|i_1-i_2|>1$, then clearly the runtime is $\Omega(mM^2)$ because the
intersection on each attribute is $\Omega(mM)$.
If $|i_1-i_2|=1$, then \lb will go through every tuple in the input relation 
$R(A_{i_1},A_{i_2})$ after it is semijoin-reduced on $A_{i_1}$ and $A_{i_2}$.
And even after two such reductions, the size of the relation is still 
$\Omega(mM^2)$.
The algorithm \nprr suffers the same drawback.

There are two potential unsatisfactory aspects of the instance above: (i) The gap between $\ms$ and the worst-case algorithms is only quadratic and (ii) the example only considers path type queries. Next, we handle these two shortcomings.

We can increase the gap in the above example by considering the following join query:
\[ Q = R_1(A_1,\dots,A_k) \Join R_2(A_2, \dots, A_{k+1}) \Join \dots 
       \Join R_m(A_m, \dots, A_{m+k-1}).
\]
Then, each relation $R_i$ still has $m$ blocks like before, where the 
$j$th block for $j \in [m]-\{i,i-1\}$ is $[(j-1)M+2,jM]^k$.
For $j=i$ the block has only one tuple $((i-1)M+1)^{(k)}$, 
and the $(i-1)$th block is empty.
It is not hard to see that there is a certificate of size
$O(mM)$, and all three algorithms Yannakakis, \nprr, and \lb run in time
at least $\Omega(mM^k)$. The reasoning is basically identical to the previous
example.

Finally, we tackle the class of $\beta$-acyclic queries for which our quadratic gap holds. We note that our argument holds for any $\beta$-acyclic query into which we can embed the $5$-path query. In other words, as long as a $\beta$-acyclic query $Q$ has attributes $A_{i_1},\dots,A_{i_6}$ and relations $R_{i_1},\dots,R_{i_5}$ such that $A_{i_j}$ is only present in $R_{i_{j-1}}$ and $R_{i_j}$ (except for the cases $j=1$ in which case $A_{i_1}$ only exists in $R_{i_1}$ and the case of $j=6$ in which case $A_{i_6}$ only exists in $R_{i_5}$), we can embed the hard instance above into a hard instance for $Q$, where we extend the values for other attributes and relations as we did in the proof of Proposition~\ref{prop:no-io-non-beta} (in Appendix~\ref{app:subsec:prop:no-io-non-beta}). Recall that the proof in Appendix~\ref{app:subsec:prop:no-io-non-beta} shows that this does not change the certificate size (which implies that the runtime of $\ms$ remains the same) while it is not hard to check that the worst-case optimal algorithms still are quadratically slower. We note that we did not try to optimize the length of the shortest path for which our hard instance still works. However, we note that the argument cannot work for path of length $3$ (since $\lb$ is instance optimal for the query $R_1(A_1,A_2)\Join R_2(A_2,A_3)\Join R_3(A_3,A_4)$, which is essentially the bowtie query). Note that the class of $\beta$-acyclic queries that has a $5$-path embedded in it as above is a fairly rich subset of $\beta$-acyclic queries.

\section{Our certificate vs. the notion of proof from \dlm.}
\label{app:sec:C-vs-dlm}

It is perhaps instructive to compare and contrast our notion of certificate from
the similar notion of ``proof'' from \dlm.
The obvious difference is that our certificate is defined for a 
natural join query, while \dlm's ``proof'' is only defined for the $Q_\cap$
query.
The difference, however, is subtler than that.
Consider only the $Q_\cap$ case. $\dlm$'s proof is output-specific while our
certificate does not need any specific mentioning of the output at all.
For some subset $B$ of {\em domain values}, they defined a {\em $B$-proof} to be
an argument for which each value $b \in B$ is certified with a spanning tree of
equalities, and each of the set of values in between consecutive values in $B$
must be ``emptiness-certified'' with an $\emptyset$-proof.
An $\emptyset$-proof in $\dlm$ is defined in terms of ``eliminating'' elements
by the $<$ comparisons. (See $\dlm$'s Lemmas 2.1 and 2.2 and Theorem 2.1.)
Our notion of certificate is stronger. For example, our certificate does not
need {\em any} $<$ comparison to certify that the output is empty.
For example, consider intersection of three sets $R,S,T$ of the same size
$N$. An emptiness certificate may consist of two {\em equalities}:
$\{ (S[N] = R[1]), (R[2] = T[1]) \}$.
It is obvious that any instance satisfying those two inequalities must have an
empty intersection because we can infer that $S[N] < T[1]$.
Subtler than that, Example~\ref{ex:same-relation-and-equalities} 
points to the fact that the
use of equalities can asymptotically reduce the certificate size even when the
output is empty.

%\iffalse
%\section{Experiment}
%\ar{If we do end up having experiments this section needs details on the machine specs.}
%We implemented \ms inside LogicBlox database engine. Three graph datasets, Facebook social network, collaboration network of Arxiv General Relativity, and Gnutella peer to peer network, are obtained from Stanford Large Network Dataset Collection. 
%Figure~\ref{fig:graph:datasets} shows some features about these graph datasets.
%
%The first join query that is run over those three datasets is the edge intersection query, which is $R(A,B) \Join S(A,B)$. For each graph dataset, $R$ and $S$ are two subsets that are chosen randomly. 
%
%The second query, $3$-path query, is $R(A) \Join S(A,B) \Join S(B,C) \Join S(C,D) \Join T(D)$. For each graph dataset, relation $S$ is given by that data, and $R$ and $T$ are the first and the second column of $S$, respectively.
%
%Finally, 4-clique query is $R(A,B) \Join R(A,C) \Join R(A,D) \Join R(B,C) \Join R(B,D) \Join R(C,D)$, where relation $R$ is given by three graph datasets.
%
%\begin{figure}
%\centering
%\begin{tabular}{|c||c|c|c|c|c|c||}
%\hline
%Dataset & Type & Nodes & Edges\\
%\hline
%Facebook & Undirected & 4,039 &88,234\\
%ca-GrQc & Undirected& 5,242 & 28,980 \\
%p2p-Gnutella04 & Directed & 10,876 & 39,994\\
%\hline
%\end{tabular}
%\caption{Graph datasets}
%\label{fig:graph:datasets}
%\end{figure}
%\fi

\section{The Triangle Query}
\label{app:triang}

The goal of this section is to prove Theorem~\ref{thm:triang}. We begin with the data structure $\cds$.

\subsection{The $\cds$}

We begin by collecting some properties of a data structure that maintain interval lists that will be useful in our new definition of $\cds$.

\subsubsection{Dyadic Tree for Intervals}

Let $N=2^{d}$ for ease throughout. 

\begin{itemize}
\item The data structure consists of nodes that contain an interval
  list.  We index the nodes of the tree by binary strings of length
  less than or equal to $d$, i.e,.  
\[ x \in \{0,1\}^{\leq d} = \cup_{j=0}^{d} \{0,1\}^{j} \]

\noindent
e.g., $x =
  ()$ is the root, $x=(0)$ is the left child, and while
  $(1,1,\dots,1)$ is the rightmost leaf. We can think of string of
  length $j$ as denoting the interval each as a binary expansion of a
  value, i.e., $[b(x) 2^{d-j}, (b(x) + 1) 2^{d-j})$.

\item Each node $x \in \{0,1\}^{\leq d}$ is associated with an
  interval list on domain $[N]$ denoted $I(x)$.

\item Given an interval $[a_1, a_2]$, we will need its dyadic
  decomposition, i.e., one of the intervals above and we denote it

\[ d([a_1,a_2]) = \{ J_1, \dots, J_k \} \text{ and } [a_1,a_2] = \bigcup_{i=1}^{k} J_i \text{ where } J_i = [b(x) 2^{d-j}, (b(x) + 1) 2^{d-j}) \text{ for } x \in \{0,1\}^{\leq d}, j \in [d]\]

For any interval, in $[N]$ we have $k \leq 2 \log N$. Let
$x(J_i)$ be the string associated with the dyadic interval $J_i$.

\item For every dyadic interval node $x$, there is an $\iarr$ $I(x)$. We also insist that any consecutive interval $[b_1,b_2]$ is stored as the collection $d([b_1,b_2])$. (This makes the upcoming arguments simpler.)

\item $\ins( [a_1,a_2], [b_1, b_2])$: For each $J \in
  d([a_1,a_2])$, set 
\[ I( x(J) ) \leftarrow I( x(J) ) \cup [b_1,b_2] \]

This takes $O( \log^2 N )$ since there are $O(\log{N})$ dyadic intervals $J$ and insertion into an $\iarr$ takes $O(\log{N})$ time.

\item $\inter(x,[b_1,b_2])$ for dyadic $[b_1,b_2]$: Returns the interval list for $I(x)\cap [b_1,b_2]$. This can be done in $O(y\cdot \log{N})$ time, where $y$ is the total number of (dyadic) intervals that need to be output. Note that to perform this task one has to find the correct position of $b_1$ and $b_2$ in $I(x)$ (which can be done in $O(\log{N})$ time) and then returning the corresponding intervals. There are two cases: (i) The output is $[b_1,b_2]$, in which case we ``charge" the intersection to this interval (note that we have $y=1$ in this case)  or (ii) The output is a subset of intervals from $I(x)$: in this case we ``charge" the intersection to these set of intervals. (Note that we can only have $y>1$ in case (ii).) This charging scheme will be useful in the proof of Proposition~\ref{prop:upward-maintain} below.
\end{itemize}

\subsubsection{The New $\cds$}

For notational convenience we will denote the interval lists for prefixes $\eps,(=a),(*),(=a,*), (*,=b)$ and $(=a,=b)$  simply as $I(), I(*), I(=a,*), I(*,=b)$ and $I(=a,=b)$.

Our new $\cds$ would be very similar to our earlier $\cds$ in Appendix~\ref{app:sec:cds} except in the following way:

The $*$-branch for variable $A$ would be replaced by a dyadic tree containing all the intervals of the form $\langle *,b,[c_1,c_2]\rangle$. In other words in addition to maintaining the interval lists $I(*,=b)$, it will also maintain for dyadic interval $x$ of $[N]$ the interval list $I(*,=x)$, which will always satisfy the following invariant:

\begin{equation}
\label{eq:upward-eq}
I(*,x)=I(*,=x\circ 0)\cap I(*,=x\circ 1).
\end{equation}

Not surprisingly, we will maintain the interval lists $I(*,=x)$ as a dyadic tree as outlined in the previous section. For the rest of the argument we will show that the $\ins$ operation on the new $\cds$ can be done in amortized $O(\log^3{N})$ time. For all constraints except of the form $\langle *,b,[c_1,c_2]\rangle$ the $\ins$ can be done in amortized $O(\log{N})$ time by Proposition~\ref{prop:insert-ctree}. So we only need to show the following:

\begin{prop}
\label{prop:upward-maintain}
Given a dyadic tree on domain $N$ with $M$ insertions of the form $\langle *,b,[c_1,c_2]\rangle$,
then \eqref{eq:upward-eq} 
can be ensured in time $O( M \log^3 N)$.
\end{prop}
\bp
We use the natural algorithm to implement the insert of constraints of the form $\langle *,b,[c_1,c_2]\rangle$, which we outlined next. Let $L$ be $[c_1,c_2]\setminus I(*,=b)$. Then perform our original $I(*,=b).\ins([c_1,c_2])$. Then do the following for every $J\in L$. Let $L'\gets \inter(\sbl(b),J)$, where $\sbl(b)$ is the sibling of $b$ in the dyadic tree. If $L'$ is empty we stop otherwise we recurse with this algorithm at the parent $x$ of $b$ (where we want to insert all intervals in $L'$ into $I(*,=x)$).

We first assume that the algorithm only deals with dyadic intervals throughout. We claim that under this assumption we would be done by assigning $O(\log^2{N})$ credits to each inserted interval. To see this first consider the simple case, where we always have $|L|=|L'|=1$, i.e. we always need to insert one interval. In this case, for each recursive level we do $O(\log{N})$ amounts of work and we have $O(\log{N})$ recursive calls  (up the path in the dyadic tree from the leaf corresponding to $b$ to the root) overall. We expand a bit on the $O(\log{N})$ work on each recursive call. Recall that the analysis of the $\inter$ procedure: we first need $O(\log{N})$ work to figure out the correct position of $J$ in the $\iarr$ of $\sbl(b)$. To pay for this we use up credits from $J$. We still have to pay for the computation of $L'$. We pay for this by charging $J$ or $L'$ as appropriate. The important point to note that is that once an interval stops ``floating" up, it can only be pushed by fresh intervals that arrive at its sibling node at a later point of time.

For the more general case (but still with dyadic interval), note that any interval that ``floats" up is either one of the $M$ inserted intervals or is ``sandwiched" between two such inserted intervals. Adjusting the constant for the number of tokens appropriately takes care of this issue.

Finally to handle the general case, we can replace any interval by $O(\log{N})$ dyadic intervals. To deal with this we need to increase the number of credits to $O(\log^3{N})$ from the previous $O(\log^2{N})$ credits. This completes the proof.
\ep

Finally, we will also ensure the following:
\begin{quote}
For every constraint  $\langle *,[b_1,b_2],*\rangle$ that is inserted into
$I(*)$, we also insert $[N]$ into the interval list $I(*,=J)$ for every $J\in
d([b_1,b_2])$.\footnote{Note that Proposition~\ref{prop:upward-maintain} only
talks about singleton $b$, it can easily be checked to see that it can handle this more general case. Basically we can make all such $J$ to be the leaves in the dyadic tree and the argument in Proposition~\ref{prop:upward-maintain} can handle insertions into leaves.}
\end{quote}

\subsection{The Algorithm}

The outer algorithm for $Q_{\triangle}$ will be the same as in Algorithm~\ref{algo:outer-algorithm}. We will have to change the $\getpp$ algorithm, which is
formally  presented in Algorithm~\ref{algo:getpp-triang}.

\begin{algorithm}[th]
\caption{$\getpp$ for evaluating the triangle query $R(A,B) \Join S(B,C) \Join T(A,C)$.}
\label{algo:getpp-triang}
\begin{algorithmic}[1]
\Require{$\cds$ as outlined earlier.}
\Statex
\State $i\gets 0$
\While{$i<3$}
	\label{step:A-start}
	\If{$i=0$} \Comment{Handling $A$}
	\State $a\gets I().\nxt(-1)$
	\If{$a=\infty$}
		\State \Return{$\NULL$}
	\EndIf
	\State $i\gets i+1$
	\label{step:A-end}
	\EndIf
	\label{step:B-start}
	\If{$i=1$} \Comment{Handling $B$}
	\State $b\gets \nxtu(I(=a),I(*),-1)$.
	\label{step:B-end}
	\If{$b=\infty$}
		\State $i\gets 0$
	\Else
		\State $i\gets i+1$
	\EndIf
	\EndIf
	\label{step:A-empty-start}
	\If{$i=2$} \Comment{Handling $C$}
		\If{$I(=a,*).\nxt(-1)=\infty$} % or $\nxtu(I(=a,*),I(*,[N]),-1)=\infty$}
		\label{step:A-empty-cond}
			\State $\cds.\insconst\left(\left\langle (a-1,a+1),*,* \right\rangle\right)$
			\label{step:A-empty-end}
			\State $i\gets 0$
		\Else \Comment{There is a probe point with prefix $(a,b)$}
			\State $x\gets \eps$ \Comment{Initializing $x$ to the root}
			\label{step:while-start}
			\While{$x$ is not a leaf}
				\State $z\gets \getcache(a,x)$
				\State $c\gets \nxtu(I(=a,*), I(*,x),z)$	
				\label{step:nxtu-dyadic}
				\State $\cache(a,x,c)$
				\label{step:nxtu-cache}
				\If{$c=\infty$} \Comment{No viable $b$ in interval $x$}
					\State $\cds.\insconst\left(\left\langle =a,x,*\right\rangle\right)$
					\label{step:ab-rule-out}
					\State $y\gets \nxtsbl(x)$
					\If{$y=\NULL$} %\Comment{\atri{I don't think this will be executed}}
						\State $\cds.\insconst\left(\left\langle (a-1,a+1),*,* \right\rangle\right)$
						\label{step:root-full}
                        			\State $i\gets 0$
						\State Exit While loop
					\EndIf
					\State $x\gets y$
				\Else
					\State $x\gets x\circ 0$
				\EndIf
			\EndWhile
			\If{$x$ is a leaf} \Comment{Found the probe point}
				\State \Return{$(a,b,c)$}
			\EndIf
		\EndIf
	\EndIf
\EndWhile
\end{algorithmic}
\end{algorithm}

Algorithm~\ref{algo:getpp-triang} uses the following helper functions:
\begin{itemize}
\item $\nxtu(I_1,I_2,v)$: Finds the smallest value $v'\ge v$ that is not covered by $I_1\cup I_2$. We implement this algorithm by the MERGE algorithm.
\item $\nxtsbl(x)$: Returns the next node in the dyadic tree by the pre-order traversal (returns $\NULL$ if the traversal is done). It is simple to implement: let $x=(x_1,\dots,x_j)$. Let $1\le i\le j$ be the largest index such that $x_i=0$. Then return $(x_1,\dots,x_{i-1},1)$. If no such $i$ exists, return $\NULL$.
\item We maintain a data structure, which keeps track of the last ``uncovered" value considered by the algorithm in the union of $I(=a,*)$ and $I(*,x)$. The function $\getcache(a,x)$ returns this cached value while $\cache(a,x,c)$ updates the cached value to $c$. 
\begin{quote}
Further, we will assume that when the outer algorithm outputs a tuple $(a,b,c)$ and adds the constraint $\langle a,b,(c-1,c+1)\rangle$, there is an accompanying call to $\cache(a,b,c+1)$.
\end{quote}
\end{itemize}

It is not hard to see that Algorithm~\ref{algo:getpp-triang} is correct. We state this fact without proof:
\blmm
\label{lem:getpp-triang-correct}
Algorithm~\ref{algo:getpp-triang} correctly returns a tuple $(a,b,c)$ that is
not covered by any existing constraints. If no such tuple exists then it
correctly outputs $\NULL$.
\elmm

%The correctness of Algorithm~\ref{algo:getpp-triang} (should?) follow from the following lemmas (and Proposition~\ref{prop:ineq}).

\begin{lmm}
\label{lem:AB}
Time spent (except those involving backtracking intervals from later part of the algorithm) in Steps~\ref{step:A-start}-\ref{step:B-end} is $O(|\cert|\log{N})$.
\end{lmm}
\bp This follows from the fact that the time spent is in some sense running Algorithm~\ref{algo:getpp-beta} on $R(A,B)\Join S(B)\Join T(A)$, which is a $\beta$-acyclic query. So our earlier proof can be easily adopted to prove the lemma.
\ep

\blmm
\label{lem:A-empty}
Time spent over prefixes that satisfy the condition in Step~\ref{step:A-empty-cond} is upper bounded by $O(|\cert|\log{N})$.
\elmm
\bp
This just follows from the fact that the total time spent is bounded by (up to constants):
\[\sum_a |I(=a,*)|\log{N} \le |\cert|\log{N},\]
as desired.
\ep

The next couple of lemmas need the following definition:
\begin{defn}
For any $A$ value $a$ define
\[B(a)=\{b| I(=a,*) \cup I(*,=b)\subset [N]\},\]
where $I(=a,*)$ and $I(*,b)$ were the interval list the first time Algorithm~\ref{algo:getpp-triang} deals (i.e. it reaches Step~\ref{step:while-start}) with the prefix $(a,b)$. If Algorithm~\ref{algo:getpp-triang} never reaches Step~\ref{step:while-start} for some $a$, then define $B(a)=\emptyset$.
\end{defn}

We first argue that the number of pairs $(a,b)$ with $b\in B(a)$ is bounded:
\blmm
\label{lem:Ba}
\[\sum_a |B(a)| \le O\left(|\cert|\right).\]
\elmm
\bp
Note that we only need to consider the values $a$ for which $B(a)\neq\emptyset$. Fix such an arbitrary $a$ and consider an arbitrary $b\in B(a)$. Now when the Algorithm gets to Step~\ref{step:while-start} we know the following:
\begin{itemize}
\item $a\not\in I()$
\item $b\not\in I(=a)\cup I(*)$.
\item $I(=a,*) \cup I(*,=b)\subset [N]$
\end{itemize}
All of the above imply that there exists a $c$ such that the tuple $(a,b,c)$ is not ruled out by the current set of constraints. This implies that Algorithm~\ref{algo:getpp-triang} will return such a tuple $(a,b,c)$ (by Lemma~\ref{lem:getpp-triang-correct}). This probe point will then be used by the outer algorithm to either (i) discover a new constraint or (ii) recognize it as an output tuple.

If there is even one $c$ such that the returned tuple $(a,b,c)$ fall in category
(i) above, then note that we can assign a unique inserted constraint to the
prefix $(a,b)$ (among all such prefixes that have $c$ such that $(a,b,c)$ falls
in category (i)). This by the argument for the runtime of the outer algorithm
implies that the number of such prefixes is bounded by $O(|\cert|)$.

Thus, we only have to consider prefixes $(a,b)$ such that every tuple $(a,b,c)$ returned by Algorithm~\ref{algo:getpp-triang} turns out to be an output tuple. We now claim that each such pair $(a,b)$ must be certified via equalities in the certificate to be present as a tuple in the relation $R(A,B)$. If this were not the case then one can come up with two database instances that satisfy all the comparisons in the certificate but in one instance $(a,b,c)$ is in the output while in the other it is not. This contradicts the definition of a certificate. Thus, we can assign each such prefix with a unique pair (one for $a$ and one for $b$) of equalities in the certificate. Further since each pair involves the tuple $(a,b)\in R$, these assignments are unique and thus, we have the number of prefixes $(a,b)$ such that all its extensions lead to output tuples is upper bounded by $|\cert|$. This completes the proof.
\ep

%\blmm
%\label{lem:ab-probe}
%Total time spent by the algorithm in prefixes $(a,b)$ that lead to a probe point is bounded by
%\[O\left(\sum_{(a,b)\mathrm{leads~to~ a~ probe}} \min(|I(=a,*)|,I(*,=b)|)\log{N}\right).\]
%\elmm
%
%\atri{Lemma~\ref{lem:ab-probe} should be correct but the argument I had in mind was incorrect. Note that for such a prefix the algorithm will go down from the root to the leaf with $\{b\}$ via all dyadic intervals $x$ that contain $b$. Now for each $O(\log{N})$ intervals, we have runtime of $\min(|I(=a,*)|,|I(*,x)|)$. Now it is true that all the points covered by $I(*,x)$ are a subset of all the points covered by $I(*,\{b\})$. Originally, I had thought that this implied that $|I(*,x)|\le |I(*,\{b\})$ (which would have been enough to prove Lemma~\ref{lem:ab-probe}). However, this inequality is {\em not} true. Thus, we need a different argument.}
%
%\blmm
%\label{lem:ab-not-probe}
%Total time spent by the algorithm in prefixes $(a,b)$ that do not lead to a probe (i.e. one that gets ruled out in Step~\ref{step:ab-rule-out}) is bounded by
%\[O\left(\sum_{(a,b')\in L} \min(|I(=a,*)|,I(*,=b)|)\log{N}\right),\]
%where $|L|=\tilde{O}(|\cert|)$.
%\elmm

\blmm
\label{lem:ab-sum}
Total time spent by Algorithm~\ref{algo:getpp-triang} on prefixes $(a,b)$ for which it reached Step~\ref{step:while-start} is bounded by
\begin{equation}
\label{eq:lem-ab-sum}
O\left(\sum_{(a,x):~ x\in B'(a)} \min(|I(=a,*)|,I(*,=x)|)\log^3{N}\right) + O\left(\sum_a |I(=a,*)|\log{N}\right) + O(Z\log^2{N}),
\end{equation}
where $B(a)\subseteq B'(a)$ is a set of disjoint dyadic intervals and $|B'(a)|\le O(|B(a)|\log{N})$.
\elmm
\bp
Let us first consider the values $a$ for which we have $I(=a,*)\cup I(*,\eps)=[N]$. In this case there is only one iteration of the While loop and Step~\ref{step:root-full} is executed. Other than Step~\ref{step:nxtu-dyadic} all the other steps take $O(\log{N})$ time. Since $\nxtu$ is implemented as the MERGE algorithm, Step~\ref{step:nxtu-dyadic} runs in time at most (up to constants):
\[\min\left(|I(=a,*)|,|I(*,\eps)|\right)\le |I(=a,*)|,\]
since the MERGE algorithm will run till it has skipped over at least one of the two interval lists. Summing up the above run-time for all $a$ such that $I(=a,*)\cup I(*,\eps)=[N]$, gives the second term in the claimed runtime.

For the rest of the proof we consider the $a$'s such that $I(=a,*)\cup
I(*,\eps)\subset [N]$: fix such an arbitrary $a$. Now consider all the possible
$B$ values. Mark a $b$ as a {\em comparison-probe} if there exists a $c$ such
that Algorithm~\ref{algo:getpp-triang} returns the tuple $(a,b,c)$, which is
used by the outer algorithm to discover a new constraint. We will mark $b$ as an
{\em output-probe} if for {\em every} $c$ such that
Algorithm~\ref{algo:getpp-triang} returns the tuple $(a,b,c)$, it is used by the
outer algorithm to discover a new output tuple. For notational convenience we
will call  $b$ a probe value if it is marked either as a comparisons-probe or an
output-probe. Note that there are exactly $|B(a)|$ probe values. Now sort the
$B$ values and consider two probe values $b< b'$ such that there are no probe
values in $(b.b')$. For the time being assume that $(b,b')$ is dyadic. Note that
in this case we execute Step~\ref{step:ab-rule-out} for $x=(b,b')$ and no
children of $x$ in the dyadic tree is explored. Now denote certain  nodes in the
dyadic tree as $\ell_1,\dots,\ell_m$ for some $m\le O(|B(a)|\log{N})$ as
follows. Each probe value $b$ (which corresponds to a singleton interval in the
dyadic tree) gets its own $\ell_i$. For any two consecutive probe values $b<b'$
(in sorted order of $B$ values) each of  $O(\log{N})$ dyadic intervals in
$(b,b')$ gets its own $\ell_i$. Since there are $|B(a)|$ probe values, there are at most $|B(a)|+1$ intervals of consecutive non-probe values. Further, each such interval gets partitioned into $O(\log{N})$ dyadic intervals, which means that we will have $m=O(|B(a)|\log{N})$ nodes $\ell_i$ overall, as desired.

Consider the subtree of the dyadic tree whose leaves are $B'(a)\stackrel{\mathrm{def}}{=}\{\ell_1,\dots,\ell_m\}$. (We will overload notation by referring to the dyadic interval corresponding to $\ell_i$ as just $\ell_i$.) We will show that the total time spent by Algorithm~\ref{algo:getpp-triang} on pairs $(a,x)$ on Step~\ref{step:while-start} and beyond is bounded by
\begin{equation}
\label{eq:time-bound}
O\left(\sum_{x\in B'(a)} \min(|I(=a,*)|,I(*,=x)|)\log^3{N}\right) +O(Z_a\log^2{N}),
\end{equation}
where $Z_a$ is the number of output tuples with $A$ value as $a$.
Summing above the above bound over all values of $a$ proves the claimed runtime bound.

To complete the proof, we prove \eqref{eq:time-bound}. First we note that for a given pair $(a,x)$, the total time spent by Algorithm~\ref{algo:getpp-triang} on Step~\ref{step:nxtu-dyadic} where $z\neq c$ is upper bounded by (up to constants)
\begin{equation}
\label{eq:step22}
\min(|I(=a,*)|,|I(*,x)|)\log{N}.
\end{equation}
The above follows from the fact that when $z\neq c$, it means that the MERGE algorithm actually made an advance and that the total number of times we can advance is upper bounded by the size of the shorter list. (Recall that each advance needs a binary search and thus takes $O(\log{N})$ time.)

Notice that all Steps other than Step~\ref{step:nxtu-dyadic} can be implemented in $O(\log{N})$ time.\footnote{The time bound is amortized for Steps~\ref{step:ab-rule-out} and~\ref{step:root-full}. However, the number of times these steps are run is bounded by $\sum_a |B'(a)|$. So the overall time spent on these steps will be bounded by $O(\sum_a |B'(a)|\log{N})$, which is subsumed by the bound in \eqref{eq:lem-ab-sum}.} 
Now note that Algorithm~\ref{algo:getpp-triang} (for the given value of $a$) only considers pairs $(a,x)$ such that $x$ contains at least one $\ell_i$. This implies two things. First, due to \eqref{eq:step22} the total time spent on Step~\ref{step:nxtu-dyadic} where $z\neq c$ is upper bounded by
\begin{equation}
\label{eq:time-bound-move}
\sum_{x:\ell_i\subseteq x\mathrm{~for~ some}~ i} \min(|I(=a,*)|,|I(*,x)|)\log{N}.
\end{equation}
Second, we have
that Algorithm~\ref{algo:getpp-triang} exactly traces all paths from the root to one of the $\ell_i$'s.  Let us consider the different cases of $\ell_i$:
\begin{enumerate}
\item ($\ell_i$ is an interval for which we run Step~\ref{step:ab-rule-out}.) If
    we exclude the time spent from the root to $\ell_i$ that is accounted for in
    \eqref{eq:time-bound-move}, then we essentially go along a path of length
    $O(\log{N})$ doing $O(\log{N})$ amount of work at each interval in the path.
    Thus, the overall time spent on these paths (excluding time spent in \eqref{eq:time-bound-move}) is bounded by (up to constant factors):
\[\sum_{x\in B'(a)} \log^2{N}.\]
\item ($\ell_i\in B(a)$ and is marked output-probe.) For such cases note that we will always have $z=c$ on all intervals in the path. So again with an argument as in the last case we spend time at most (up to constant factors) 
\[Z_a\log^2{N}.\]
\item ($\ell_i\in B(a)$ and is marked comparison-probe.) Note that the number of $c$ values for which we get $(a,b,c)$ probe points is upper bounded by $O(\min(|I(=a,*)|,|I(*,=\ell_i)|))$. Thus, total time spent in this case outside of the time accounted for in  \eqref{eq:time-bound-move}, by an argument similar to the earlier cases is upper bounded by (up to constants)
\[\sum_{b\in B(a)} \min(|I(=a,*)|, I(*,=b)|)\log^2{N}.\]
\end{enumerate}

Adding up the bounds above with \eqref{eq:time-bound-move} implies that the time bound we are after is at most (up to constants)
\[\sum_{x:\ell_i\subseteq x\mathrm{~for~ some}~ i} \min(|I(=a,*)|,|I(*,x)|)\log^2{N} +Z_a\log^2{N}.\]
To complete the proof, we will argue that
\begin{equation}
\label{eq:sum-levels}
\sum_{x:\ell_i\subseteq x\mathrm{~for~ some}~ i} \min(|I(=a,*)|,|I(*,x)|) \le O\left(\sum_{y\in B'(a)} \min(|I(=a,*)|,|I(*,y)|)\log{N}\right),
\end{equation}
since the above will imply \eqref{eq:time-bound}. To see why the above is true, note that by \eqref{eq:upward-eq}, we have
\[I(*,y)= \cap_{i:\ell_i\subseteq y} I(*,\ell_i),\]
where the intersection is over the set of points covered by the interval lists. This in turn implies that
\[|I(*,y)|\le \sum_{i:\ell_i\subseteq y} |I(*,\ell_i)|.\]
The above in turn implies that
\[\min(|I(=a,*)|,|I(*,y)|)\le  \sum_{i:\ell_i\subseteq y} \min(|I(=a,*)|,|I(*,\ell_i)|).\]
Noting that for intervals $y\neq y'$ of the same size, the set of $\ell_i$'s contained in them are disjoint and that there are $O(\log{N})$ distinct sizes for dyadic intervals, the above implies \eqref{eq:sum-levels}, as desired.
\ep

We are finally ready to prove the runtime for Algorithm~\ref{algo:getpp-triang}:
\bthm
\label{thm:getpp-triang}
Over all calls to Algorithm~\ref{algo:getpp-triang} from the outer algorithm, the total time spent is bounded by
\[O\left(|\cert|^{3/2}\log^{7/2}{N}+Z\log^2{N}\right).\]
\ethm
\bp
The total time spent is bounded by the sum of the time bounds in Lemmas~\ref{lem:AB},~\ref{lem:A-empty} and~\ref{lem:ab-sum}. The first two terms are subsumed by the bound in this lemma. Thus, we only need to bound
\[O\left(\sum_{(a,x):~ x\in B'(a)} \min(|I(=a,*)|,I(*,=x)|)\log^3{N}\right) + O\left(\sum_a |I(=a,*)|\log{N}\right) + O(Z\log^2{N}).\]
Since $\sum_a |I(=a,*)|\le|\cert|$, the last two terms in the sum above are subsumed by the bound in this lemma. So we are left with the bound
\begin{equation}
\label{eq:cert-tt-bound}
\sum_a \sum_{x\in B'(a)} \min(|I(=a,*)|,I(*,=x)|).
\end{equation}
Next we note the following:
\[\sum_a |I(=a,*)|\le |\cert|,\]
and
\[\sum_{x\in B'(a)} |I(*,=x)| \le \sum_b |I(*,=b)| \le |\cert|.\]
In the above the first inequality follows from the fact that every $x\neq x'\in
B'(a)$ are disjoint and by the argument used in proof of Lemma~\ref{lem:ab-sum},
$|I(*,x)|\le \sum_{b\in x} |I(*,=b)|$. Then
Lemmas~\ref{lem:chris-ineq},~\ref{lem:Ba} and~\ref{lem:ab-sum} imply that \eqref{eq:cert-tt-bound} is bounded by $|\cert|^{3/2}\sqrt{\log{N}}$, as desired.
\ep

\blmm
\label{lem:chris-ineq}
For any two vectors $u,v\in\mathbb{R}^M_{\ge 0}$ and a set $J\subseteq [M]\times [M]$, we have
\[\sum_{(i,j)\in J} \min(u_i,v_i)\le \sqrt{|J|\cdot \norm{u}_1\cdot\norm{v}_1}.\]
\elmm
\begin{proof}
For notational convenience, define
\[J[i]=\{j|(i,j)\in J\}.\]
Now consider the following sequence of relationships:
\begin{align*}
\sum_{ (i,j) \in J} \min \{u_i, v_j\} &\leq \sum_{ (i,j) \in J} \sqrt{u_iv_j}\\
&= \sum_i \sqrt{u_i}\sum_{j\in J[i]} \sqrt{v_j}\\
&\le \sum_i \sqrt{u_i}\cdot\sqrt{|J[i]}|\cdot \sqrt{\norm{v}_1}\\
&= \sqrt{\norm{v}_1}\cdot \sum_i \sqrt{u_i}\cdot\sqrt{|J[i]}|\\
&\le \sqrt{\norm{v}_1}\cdot \sqrt{\norm{u}_1}\cdot \sqrt{\sum_i |J[i]|}\\
&=\sqrt{\norm{v}_1}\cdot\sqrt{\norm{u}_1}\cdot \sqrt{|J|},
\end{align*}
where the inequalities follow from Cauchy-Schwarz inequality.
\end{proof}

\subsection{Wrapping it up}

It is easy to see that Theorem~\ref{thm:getpp-triang},
Proposition~\ref{prop:upward-maintain}, Proposition~\ref{prop:insert-ctree} and
Theorem~\ref{thm:analyze-outer-algorithm} prove Theorem~\ref{thm:triang}.


\begin{thebibliography}{10}

\bibitem{DBLP:books/aw/AbiteboulHV95}
{\sc S.~Abiteboul, R.~Hull, and V.~Vianu}, {\em Foundations of Databases},
  Addison-Wesley, 1995.

\bibitem{MR2351517}
{\sc I.~Adler, G.~Gottlob, and M.~Grohe}, {\em Hypertree width and related
  hypergraph invariants}, European J. Combin., 28 (2007).

\bibitem{geometric-io}
{\sc P.~Afshani, J.~Barbay, and T.~M. Chan}, {\em Instance-optimal geometric
  algorithms}, in FOCS, 2009, pp.~129--138.

\bibitem{self-improving}
{\sc N.~Ailon, B.~Chazelle, K.~L. Clarkson, D.~Liu, W.~Mulzer, and
  C.~Seshadhri}, {\em Self-improving algorithms}, SIAM J. Comput., 40 (2011),
  pp.~350--375.

\bibitem{MR599482}
{\sc N.~Alon}, {\em On the number of subgraphs of prescribed type of graphs
  with a given number of edges}, Israel J. Math., 38 (1981).

\bibitem{MR985145}
{\sc S.~Arnborg and A.~Proskurowski}, {\em Linear time algorithms for {NP}-hard
  problems restricted to partial {$k$}-trees}, Discrete Appl. Math., 23 (1989),
  pp.~11--24.

\bibitem{FOCS:AtsGroMar08}
{\sc A.~Atserias, M.~Grohe, and D.~Marx}, {\em Size bounds and query plans for
  relational joins}, 2008, pp.~739--748.

\bibitem{DBLP:conf/soda/BarbayK02}
{\sc J.~Barbay and C.~Kenyon}, {\em Adaptive intersection and t-threshold
  problems}, in SODA, 2002, pp.~390--399.

\bibitem{DBLP:journals/talg/BarbayK08}
\leavevmode\vrule height 2pt depth -1.6pt width 23pt, {\em Alternation and
  redundancy analysis of the intersection problem}, ACM Transactions on
  Algorithms, 4 (2008).

\bibitem{DBLP:conf/sccc/BarbayL09}
{\sc J.~Barbay and A.~L{\'o}pez-Ortiz}, {\em Efficient algorithms for context
  query evaluation over a tagged corpus}, in SCCC, M.~Arenas and B.~Bustos,
  eds., IEEE Computer Society, 2009, pp.~11--17.

\bibitem{Beeri:1981:PAD:800076.802489}
{\sc C.~Beeri, R.~Fagin, D.~Maier, A.~Mendelzon, J.~Ullman, and M.~Yannakakis},
  {\em Properties of acyclic database schemes}, in STOC, New York, NY, USA,
  1981, ACM, pp.~355--362.

\bibitem{Beeri:1983:DAD:2402.322389}
{\sc C.~Beeri, R.~Fagin, D.~Maier, and M.~Yannakakis}, {\em On the desirability
  of acyclic database schemes}, J. ACM, 30 (1983), pp.~479--513.

\bibitem{Berge:1985:GH:1096893}
{\sc C.~Berge}, {\em Graphs and Hypergraphs}, Elsevier Science Ltd, 1985.

\bibitem{braultbaron:LIPIcs:2012:3669}
{\sc J.~Brault-Baron}, {\em {A Negative Conjunctive Query is Easy if and only
  if it is Beta-Acyclic}}, in {CSL 12}, vol.~16, 2012, pp.~137--151.

\bibitem{brouwer-kolen-1980}
{\sc A.~Brouwer and A.~Kolen}, {\em A super-balanced hypergraph has a nest
  point},  (1980).
\newblock Tech. Report.

\bibitem{DBLP:conf/stoc/ChandraM77}
{\sc A.~K. Chandra and P.~M. Merlin}, {\em Optimal implementation of
  conjunctive queries in relational data bases}, in STOC, 1977.

\bibitem{DBLP:journals/tcs/ChekuriR00}
{\sc C.~Chekuri and A.~Rajaraman}, {\em Conjunctive query containment
  revisited}, Theor. Comput. Sci., 239 (2000), pp.~211--229.

\bibitem{DBLP:conf/soda/ChenLSZ07}
{\sc J.~Chen, S.~Lu, S.-H. Sze, and F.~Zhang}, {\em Improved algorithms for
  path, matching, and packing problems}, in SODA, 2007, pp.~298--307.

\bibitem{journals/ai/DechterP89}
{\sc R.~Dechter and J.~Pearl}, {\em Tree clustering for constraint networks.},
  Artificial Intelligence, 38 (1989), pp.~353--366.

\bibitem{DBLP:conf/soda/DemaineLM00}
{\sc E.~D. Demaine, A.~L{\'o}pez-Ortiz, and J.~I. Munro}, {\em Adaptive set
  intersections, unions, and differences}, in SODA, 2000, pp.~743--752.

\bibitem{Duris:2008:HAE:1381308.1382241}
{\sc D.~Duris}, {\em Hypergraph acyclicity and extension preservation
  theorems}, in LICS, 2008, pp.~418--427.

\bibitem{Duris2012}
\leavevmode\vrule height 2pt depth -1.6pt width 23pt, {\em Some
  characterizations of $\gamma $ and $\beta $-acyclicity of hypergraphs.},
  Information Processing Letters, 112 (2012).

\bibitem{adaptive-sort}
{\sc V.~Estivill-Castro and D.~Wood}, {\em A survey of adaptive sorting
  algorithms}, ACM Comput. Surv., 24 (1992), pp.~441--476.

\bibitem{DBLP:journals/jacm/Fagin83}
{\sc R.~Fagin}, {\em Degrees of acyclicity for hypergraphs and relational
  database schemes}, J. ACM, 30 (1983), pp.~514--550.

\bibitem{fagin-io}
{\sc R.~Fagin, A.~Lotem, and M.~Naor}, {\em Optimal aggregation algorithms for
  middleware}, J. Comput. Syst. Sci., 66 (2003), pp.~614--656.

\bibitem{Fagin:1982:SUR:319732.319735}
{\sc R.~Fagin, A.~O. Mendelzon, and J.~D. Ullman}, {\em A simplied universal
  relation assumption and its properties}, TODS, 7 (1982).

\bibitem{DBLP:journals/jacm/FlumFG02}
{\sc J.~Flum, M.~Frick, and M.~Grohe}, {\em Query evaluation via
  tree-decompositions}, J. ACM, 49 (2002), pp.~716--752.

\bibitem{Freuder:1990:CKS:1865499.1865500}
{\sc E.~C. Freuder}, {\em Complexity of k-tree structured constraint
  satisfaction problems}, in AAAI, AAAI'90, AAAI Press, 1990, pp.~4--9.

\bibitem{Goodman:1982:TQS:319758.319775}
{\sc N.~Goodman and O.~Shmueli}, {\em Tree queries: a simple class of
  relational queries}, ACM Trans. Database Syst., 7 (1982).

\bibitem{DBLP:conf/wg/GottlobGMSS05}
{\sc G.~Gottlob, M.~Grohe, N.~Musliu, M.~Samer, and F.~Scarcello}, {\em
  Hypertree decompositions: Structure, algorithms, and applications}, in WG,
  D.~Kratsch, ed., vol.~3787 of LNCS, Springer, 2005.

\bibitem{DBLP:journals/jcss/GottlobLS02}
{\sc G.~Gottlob, N.~Leone, and F.~Scarcello}, {\em Hypertree decompositions and
  tractable queries}, J. Comput. Syst. Sci., 64 (2002), pp.~579--627.

\bibitem{Gottlob:2009:GHD:1568318.1568320}
{\sc G.~Gottlob, Z.~Mikl\'{o}s, and T.~Schwentick}, {\em Generalized hypertree
  decompositions: Np-hardness and tractable variants}, J. ACM, 56 (2009),
  pp.~30:1--30:32.

\bibitem{DBLP:conf/pods/Grohe01}
{\sc M.~Grohe}, {\em The parameterized complexity of database queries}, in
  PODS, 2001, pp.~82--92.

\bibitem{Heggernes:2008:FCM:1654201.1654211}
{\sc P.~Heggernes and B.~W. Peyton}, {\em Fast computation of minimal fill
  inside a given elimination ordering}, SIAM J. Matrix Anal. Appl., 30 (2008),
  pp.~1424--1444.

\bibitem{MR0297088}
{\sc F.~K. Hwang and S.~Lin}, {\em A simple algorithm for merging two disjoint
  linearly ordered sets}, SIAM J. Comput., 1 (1972), pp.~31--39.

\bibitem{DBLP:journals/siamcomp/ItaiR78}
{\sc A.~Itai and M.~Rodeh}, {\em Finding a minimum circuit in a graph}, SIAM J.
  Comput., 7 (1978), pp.~413--423.

\bibitem{DBLP:journals/im/KolountzakisMPT12}
{\sc M.~N. Kolountzakis, G.~L. Miller, R.~Peng, and C.~E. Tsourakakis}, {\em
  Efficient triangle counting in large graphs via degree-based vertex
  partitioning}, Internet Mathematics, 8 (2012), pp.~161--185.

\bibitem{Maier:1982:CAH:588111.588118}
{\sc D.~Maier and J.~D. Ullman}, {\em Connections in acyclic hypergraphs:
  extended abstract}, in PODS, ACM, 1982, pp.~34--39.

\bibitem{DBLP:journals/talg/Marx10}
{\sc D.~Marx}, {\em Approximating fractional hypertree width}, ACM Transactions
  on Algorithms, 6 (2010).

\bibitem{DBLP:conf/pods/NgoPRR12}
{\sc H.~Q. Ngo, E.~Porat, C.~R{\'e}, and A.~Rudra}, {\em Worst-case optimal
  join algorithms: [extended abstract]}, in PODS, 2012, pp.~37--48.

\bibitem{ordyniak_et_al:LIPIcs:2010:2855}
{\sc S.~Ordyniak, D.~Paulusma, and S.~Szeider}, {\em {Satisfiability of Acyclic
  and Almost Acyclic CNF Formulas}}, in {FSTTCS 2010}, vol.~8, 2010.

\bibitem{DBLP:conf/pods/PaghP06}
{\sc A.~Pagh and R.~Pagh}, {\em Scalable computation of acyclic joins}, in
  PODS, 2006, pp.~225--232.

\bibitem{DBLP:conf/pods/PapadimitriouY97}
{\sc C.~H. Papadimitriou and M.~Yannakakis}, {\em On the complexity of database
  queries}, in PODS, 1997, pp.~12--19.

\bibitem{3sum}
{\sc M.~P{\v a}tra{\c s}cu}, {\em Towards polynomial lower bounds for dynamic
  problems}, in STOC, 2010, pp.~603--610.

\bibitem{Ramakrishnan:2002:DMS:560733}
{\sc R.~Ramakrishnan and J.~Gehrke}, {\em Database Management Systems},
  McGraw-Hill, Inc., New York, NY, USA, 3~ed., 2003.

\bibitem{tim-notes}
{\sc T.~Roughgarden}, {\em Lecture notes for {CS369N} ``beyond worst-case
  analysis"}.
\newblock
  \href{http://theory.stanford.edu/~tim/f09/f09.html}{http://theory.stanford.edu/~tim/f09/f09.html},
  2009.

\bibitem{tim-ps}
\leavevmode\vrule height 2pt depth -1.6pt width 23pt, {\em Problem set \#1
  ({CS369N}: Beyond worst-case analysis)}.
\newblock
  \href{http://theory.stanford.edu/~tim/f11/hw1.pdf}{http://theory.stanford.edu/~tim/f11/hw1.pdf},
  2011.

\bibitem{Schafhauser2006}
{\sc W.~Schafhauser}, {\em {New Heuristic Methods for Tree Decompositions and
  Generalized Hypertree Decompositions}}, Master's thesis, 2006.

\bibitem{DBLP:conf/www/SuriV11}
{\sc S.~Suri and S.~Vassilvitskii}, {\em Counting triangles and the curse of
  the last reducer}, in WWW, 2011, pp.~607--614.

\bibitem{DBLP:conf/stoc/Vardi82}
{\sc M.~Y. Vardi}, {\em The complexity of relational query languages (extended
  abstract)}, in STOC, 1982, pp.~137--146.

\bibitem{STOC:VasWil06}
{\sc V.~Vassilevska and R.~Williams}, {\em Finding a maximum weight triangle in
  $n^{3-\delta}$ time, with applications}, in STOC, 2006, pp.~225--231.

\bibitem{Vassilevska:2009:FMC:1536414.1536477}
{\sc V.~Vassilevska and R.~Williams}, {\em Finding, minimizing, and counting
  weighted subgraphs}, in STOC, ACM, 2009, pp.~455--464.

\bibitem{DBLP:journals/corr/abs-1210-0481}
{\sc T.~L. Veldhuizen}, {\em Leapfrog triejoin: a worst-case optimal join
  algorithm}, ICDT,  (2014).
\newblock To Appear.

\bibitem{DBLP:journals/jcss/Willard02}
{\sc D.~E. Willard}, {\em An algorithm for handling many relational calculus
  queries efficiently}, J. Comput. Syst. Sci., 65 (2002), pp.~295--331.

\bibitem{DBLP:conf/vldb/Yannakakis81}
{\sc M.~Yannakakis}, {\em Algorithms for acyclic database schemes}, in VLDB,
  1981, pp.~82--94.

\end{thebibliography}
\end{document}